\documentclass[12pt]{main}  

\usepackage[utf8]{inputenc}
\usepackage[T1]{fontenc}
\usepackage{graphicx}
\usepackage{mathrsfs}
\usepackage{hyperref}
\hypersetup{colorlinks,linkcolor={blue},citecolor={blue},urlcolor={red}}
\usepackage[english]{babel}
\usepackage{verbatim}
\usepackage{amssymb}
\usepackage{amsmath}
\usepackage[english]{babel}
\usepackage{amsthm}
\usepackage{mathtools}
\usepackage{amsfonts,bbm}
\usepackage{amssymb,amscd}
\usepackage{enumitem}
\usepackage{tikz}

\newtheorem{theorem}{Theorem}
\newtheorem{hypothesis}{Hypothesis}
\newtheorem*{Srednicki}{Srednicki's ETH}
\usepackage[skip=0.33\baselineskip]{caption}
\usepackage{ragged2e}
\usepackage{booktabs, multirow, tabularx}
\newcolumntype{L}{>{\RaggedRight}X}
\usepackage{adjustbox}
\usepackage{chngpage}
\usepackage{colortbl,xcolor,color}
\definecolor{mypink1}{rgb}{0.858, 0.188, 0.478}
\definecolor{mypink2}{RGB}{219, 48, 122}
\definecolor{mypink3}{cmyk}{0, 0.7808, 0.4429, 0.1412}
\definecolor{mygray}{gray}{0.6}

\newcommand{\R}{\mathbb{R}}
\newcommand{\e}{\mathrm{e}}
\newcommand{\Hilb}{\mathcal{H}}

\newcommand{\Ham}{\varmathbb{H}}

\DeclareMathOperator{\Tr}{Tr}

\newcommand{\1}{\mathbbm{1}}

\newcommand{\Ket}[1]{ \left| #1 \right\rangle}
\newcommand{\Bra}[1]{ \left\langle #1 \right|}
\newcommand{\Scal}[2]{ \left\langle #1 | #2 \right\rangle}

\newcommand{\MV}[1]{\langle #1 \rangle}

\newcommand{\braket}[2]{\langle #1 | #2 \rangle}
\newcommand{\ketbra}[2]{| #2 \rangle\langle #1 |}
\newcommand{\ket}[1]{| #1 \rangle}
\newcommand{\bra}[1]{\langle #1 |}
\renewcommand{\equiv}{\coloneqq}

\newcommand{\KEn}{\ket{E_n}}

\DeclareMathOperator{\rhont}{\rho_t^{(n)}}

\DeclareMathOperator{\rde}{\rho_{\mathrm{DE}}}

\usepackage{xparse}

\DeclareSymbolFont{lettersA}{U}{txmia}{m}{it}
\SetSymbolFont{lettersA}{bold}{U}{txmia}{bx}{it}
\DeclareFontSubstitution{U}{txmia}{m}{it}

\makeatletter
\DeclareMathSymbol{\m@thbbch@rA}{\mathord}{lettersA}{129}
\DeclareMathSymbol{\m@thbbch@rB}{\mathord}{lettersA}{130}
\DeclareMathSymbol{\m@thbbch@rC}{\mathord}{lettersA}{131}
\DeclareMathSymbol{\m@thbbch@rD}{\mathord}{lettersA}{132}
\DeclareMathSymbol{\m@thbbch@rE}{\mathord}{lettersA}{133}
\DeclareMathSymbol{\m@thbbch@rF}{\mathord}{lettersA}{134}
\DeclareMathSymbol{\m@thbbch@rG}{\mathord}{lettersA}{135}
\DeclareMathSymbol{\m@thbbch@rH}{\mathord}{lettersA}{136}
\DeclareMathSymbol{\m@thbbch@rI}{\mathord}{lettersA}{137}
\DeclareMathSymbol{\m@thbbch@rJ}{\mathord}{lettersA}{138}
\DeclareMathSymbol{\m@thbbch@rK}{\mathord}{lettersA}{139}
\DeclareMathSymbol{\m@thbbch@rL}{\mathord}{lettersA}{140}
\DeclareMathSymbol{\m@thbbch@rM}{\mathord}{lettersA}{141}
\DeclareMathSymbol{\m@thbbch@rN}{\mathord}{lettersA}{142}
\DeclareMathSymbol{\m@thbbch@rO}{\mathord}{lettersA}{143}
\DeclareMathSymbol{\m@thbbch@rP}{\mathord}{lettersA}{144}
\DeclareMathSymbol{\m@thbbch@rQ}{\mathord}{lettersA}{145}
\DeclareMathSymbol{\m@thbbch@rR}{\mathord}{lettersA}{146}
\DeclareMathSymbol{\m@thbbch@rS}{\mathord}{lettersA}{147}
\DeclareMathSymbol{\m@thbbch@rT}{\mathord}{lettersA}{148}
\DeclareMathSymbol{\m@thbbch@rU}{\mathord}{lettersA}{149}
\DeclareMathSymbol{\m@thbbch@rV}{\mathord}{lettersA}{150}
\DeclareMathSymbol{\m@thbbch@rW}{\mathord}{lettersA}{151}
\DeclareMathSymbol{\m@thbbch@rX}{\mathord}{lettersA}{152}
\DeclareMathSymbol{\m@thbbch@rY}{\mathord}{lettersA}{153}
\DeclareMathSymbol{\m@thbbch@rZ}{\mathord}{lettersA}{154}
\makeatother

\ExplSyntaxOn
\NewDocumentCommand{\varmathbb}{m}
 {
  \tl_map_inline:nn { #1 } { \use:c { m@thbbch@r##1 } }
 }
\ExplSyntaxOff

\title{Pure states statistical mechanics:\\[1ex]     
        On its foundations and \\[1ex] applications to quantum gravity}   

\author{Fabio Anz\`a}             
\college{St Catherine's College}  

\degree{Doctor of Philosophy}     
\degreedate{Trinity 2018}         

\begin{document}

\baselineskip=18pt plus1pt

\setcounter{secnumdepth}{3}
\setcounter{tocdepth}{3}

\maketitle                  
\begin{dedication}
I dedicate this work to my grandparents, \emph{Nonna Pia} and \emph{Nonno Coc\`o}. The path you laid down and your endless support gives me the strength to outlast my mistakes. For that, I will always be extraordinarily  grateful.
\end{dedication}        
\begin{center}
{\huge {\bf Abstract}}
\end{center}
\vspace{1cm}
$\quad \,\,\, $The project concerns the study of the interplay among quantum mechanics, statistical mechanics and thermodynamics, in isolated quantum systems. The goal of this research is to improve our understanding of the concept of thermal equilibrium in quantum systems. 

First, I investigated the role played by observables and measurements in the emergence of thermal behaviour. This led to a new notion of thermal equilibrium which is specific for a given observable, rather than for the whole state of the system. The equilibrium picture that emerges is a generalization of statistical mechanics in which we are not interested in the state of the system but only in the outcome of the measurement process. I investigated how this picture relates to one of the most promising approaches for the emergence of thermal behaviour in quantum systems: the Eigenstate Thermalization Hypothesis. Then, I applied the results to study the equilibrium properties of peculiar quantum systems, which are known to escape thermalization: the many-body localised systems. Despite the localization phenomenon, which prevents thermalization of subsystems, I was able to show that we can still use the predictions of statistical mechanics to describe the equilibrium of some observables. Moreover, the intuition developed in the process led me to propose an experimentally accessible way to unravel the interacting nature of many-body localised systems. 

Then, I exploited the ``Concentration of Measure'' and the related ``Typicality Arguments'' to study the macroscopic properties of the basis states in a tentative theory of quantum gravity: Loop Quantum Gravity. These techniques were previously used to explain why the thermal behaviour in quantum systems is such an ubiquitous phenomenon at the macroscopic scale. I focused on the local properties, their thermodynamic behaviour and interplay with the semiclassical limit. The ultimate goal of this line of research is to give a quantum description of a black hole which is consistent with the expected semiclassical behaviour. This was motivated by the necessity to understand, from a quantum gravity perspective, how and why an horizon exhibits thermal properties.

\begin{acknowledgements}

\hspace{0.5cm} These three years have been a painfully wonderful journey. I have met amazing people and each one of them has given me something I am thankful for: a few moments of their time. There is absolutely no way I will be able to express, in words, how grateful I am. But I am going to try anyway. I would like to start by thanking the people who have had the patience to work with me. 

\hspace{0.5cm} First, my supervisor, Vlatko Vedral, for letting me forge my own research path, without interfering, while being a constant source of encouragement and scientific inspiration. On that note, I am also profoundly grateful to Davide Girolami for always being there when I needed an advice. I will miss, deeply, our chats about physics, your insights and our discussions about Italian politics. It was really fun. 

\hspace{0.5cm} I would also like to thank the other members of the Oxford group for providing a wonderful research environment: Tristan Farrow, Oscar Dahlsten, Benjamin Yadin, Christian Schilling, Cormac Browne, the Felixes (or Felices, according to Ben\cite{Yadin2017}) Binder and Tennie, Nana Liu, Pieter Bogaert, Anu Unnikrishnan, Tian Zhang and Reevu Maity.  Many thanks also to my collaborators for helping, teaching and showing me your way: Goffredo Chirco, Christian Gogolin, Marcus Huber and Francesca Pietracaprina.  

\hspace{0.5cm} During some rough times, I was lucky enough to have someone who was willing to give me a piece of advice. For that, I am particularly indebted to John Goold and Michele Campisi. I want you to know your advices gave me the strength to trust my choice, even when others casted doubt on it. 

\hspace{0.5cm} Many thanks to all the people and institutions who have hosted me during these three years. In particular, I would like to thank Rosario Fazio and the Condensed Matter Group at ICTP for hosting me in Trieste. A special thank is directed to Jim Crutchfield and the Information Theory group at U.C. Davis: For opening your door, giving me asylum and showing me what an amazing environment you created. I am also grateful to the ``Angelo Della Riccia'' foundation and to St. Catherine's College for providing funding for my research.

\hspace{0.5cm} Among the amazing people I have met in these years, there is one group which always made me feel at home: the Oxford University Volleyball Club. I will never be able to thank you enough for everything we have experienced together: The game and the adrenaline, the incredibly late nights, the fun on the sand and even the boring line-judging duties during the ladies' games. We have accomplished much and these marvellous memories will always have a special place in my heart. So, thank you KBar, Alex, Jonas, Gytis, Andy, David, Sanders, Stefan, Rory, Kuba, Christos, Adam, Sven, Nick, Andrea and Megan. I will always hold your friendship in the highest regard.

\hspace{0.5cm} Finally, I would like to thank my parents, my brother and my sister for their endless support. You have given me the strength to make my own mistakes and shown me the wisdom to accept them. I will always love you and I want you to know I could not have done this without you. Last but not least, I would like to thank Carmen for showing up, out of the blue, and for deciding to put up with me every day. It is not easy, but we will have a lot of fun.

\end{acknowledgements}  

\begin{romanpages}          
\tableofcontents            
\listoffigures              
\end{romanpages}            

\addcontentsline{toc}{chapter}{Preface}
\chapter*{Preface}


Quantum Statistical mechanics can be considered as a set of tools which attempt to connect the microscopic description of small quantum systems with the macroscopic behaviour, mostly governed by thermodynamics. With this mindset, the purpose of this thesis if twofold. On the one hand, I aimed at extending the domain of applicability of the tools of statistical mechanics, focussing on observables, rather than on the whole state of the system. On the other hand, one of the popular approaches to connect microscopic details to macroscopic behaviour will be exploited to study the large-scale limit of the states of a tentative theory of quantum gravity: Loop Quantum Gravity. The goal here is to compare the predicted behaviour with the known phenomenology of General Relativity, looking for analogies or discrepancies. 

The manuscript contains this Preface, an Overview, nine chapters and four appendices. The first page of each chapter outlines its content, while the last section summarises its results. Moreover, each chapter has been written so that it is self-contained. I made this choice by considering that most readers are not going to read the manuscript all at once. Rather, they will focus on the information relevant to them at each time and this choice should facilitate the reading. Here we briefly summarise the content. 

\begin{itemize}
\item In the Overview we provide some context, highlighting the mindset of the author, and a summary of the thesis. 

\item Chapter \ref{Ch1} is merely introductory. We review basic tools of quantum mechanics and of quantum information. This is mainly to setup notation and terminology.

\item In Chapter \ref{Ch2} we introduce the basic approaches to the emergence of thermal equilibrium in isolated quantum systems. These have been collected under the name of ``Pure States Statistical Mechanics''. We also included a roadmap of important reviews on this topic. This is done mainly for a reader who is not familiar with the issues of field.

\item Chapter \ref{Ch3} contains a new approach to explain the emergence of thermal equilibrium in quantum systems. We also make a connection with one of the approaches highlighted in Chapter \ref{Ch2}: the Eigenstate Thermalization Hypothesis. The notion of Thermal Observables is given, mirroring the concept of thermal state. Moreover, the definition of Hamiltonian Unbiased Observables is given and their properties are studied.

\item In Chapter \ref{Ch4} we build our intuition on the ideas presented in Chapter \ref{Ch3} and investigate the conditions which make a physically relevant observable to be Hamiltonian Unbiased. We focus on observables which can be concretely measured in a laboratory.

\item In Chapter \ref{Ch5} the proposed picture of thermalization of observables, is tested on a class of systems which are known to avoid thermalization: the many-body localised systems. Moreover, exploiting the ideas developed in the previous chapters we propose a new diagnostic tool to test whether a system is in the Anderson Localised phase or in the Many-Body Localised phase.

\item The second part of the thesis starts with Chapter \ref{Ch6}, which has two purposes. On the one hand, we address the problem of studying the macroscopic behaviour of the states in Loop Quantum Gravity with ideas and techniques borrowed from Information Theory. On the other hand, we provide a technical introduction to the basis states of Loop Quantum Gravity, the spin networks, and their geometric interpretation.

\item In Chapter \ref{Ch7} one of the techniques to understand quantum thermalization is extended and applied to study the local properties of spin networks in the macroscopic regime. This is done in a simple setup where we have a large closed surface and we are interested in the properties of a small patch.

\item In Chapter \ref{Ch8} we push forward the ideas of the two previous chapters. We investigate the boundary properties of a $3D$ region of quantum space. In particular, we analyse the macroscopic properties of the boundary and their interplay with the expected semiclassical phenomenology. 

\item Eventually, in Chapter \ref{Ch9} we provide a quick summary of the results and discuss the picture that emerges. We also outline some future lines of investigation. 
\end{itemize}

\section*{List of Publications}

This manuscript is the result of the author's research, often developed in collaboration
with others. Chapters \ref{Ch3},\ref{Ch4} and \ref{Ch5} deal with the foundations of statistical mechanics
and thermodynamics. They are based on the following papers, \cite{Anza2017,Anza2018,Anza2018b}   
\begin{itemize}
\item F. Anza, V. Vedral, {\it Information-theoretic equilibrium and observable thermalization}, Scientific Reports 7, 44066 (2017)
\item F. Anza, C. Gogolin, M. Huber, {\it Eigenstate Thermalization for degenerate observables},  Phys. Rev. Lett. (2018)
\item F. Anza, F. Pietracaprina, J. Goold, {\it Local aspects of MBL dynamics}, In preparation
\end{itemize}
Chapters \ref{Ch7} and \ref{Ch8} are the result of the work of the author, in collaboration with Dr. G. Chirco, on the problem
of studying the macroscopic regime of spin networks in Loop Quantum Gravity. They are based on the following papers \cite{Anza2016,Anza2017b}: 
\begin{itemize}
\item F. Anza, G. Chirco, {\it Typicality in spin network states of quantum geometry}, Phys. Rev. D 94, 084047 (2016)
\item F. Anza, G. Chirco, {\it Fate of the Hoop Conjecture in quantum gravity}, Phys. Rev. Lett. 119, 231301 (2017)
\end{itemize}
These two sets of work identify the main lines of research followed. Aside them, the following papers \cite{Anza2015a,Aaltonen2015,Anza2015} where also published during D.Phil., but they
are not part of the main lines of research so they have not been included in the thesis:
\begin{itemize}
\item F. Anza, S. Speziale, {\it A note on the secondary simplicity constraints in loop quantum gravity}, Class. Quantum Grav. 32 195015 (2015) 
\item T. Aaltonen {\it et al}., {\it Search for resonances decaying to Top and Bottom quark with the CDF experiment}, Phys. Rev. Lett. 115, 061801 (2015)
\item F. Anza, S. Di Martino, B. D. Militello, A. Messina, {\it Dynamics of a particle confined in a two-dimensional dilating and deforming domain}, Phys. Scr. 90 074062 (2015)
\end{itemize}
Moreover, the following paper was written to explain to the larger public a basic intuition behind thermalization of quantum systems. It has been published in the ``Romulus Magazine'', a journal 
run by the Wolfson College of Oxford:
\begin{itemize}
\item F. Anza, {\it Typically growing entropy}, Romulus Magazin, Wolfson College, University of Oxford (2015)
\end{itemize}

\addcontentsline{toc}{chapter}{Overview}
\chapter*{Overview}\label{Ch:Overview}

Thermodynamics was originally developed as a phenomenological theory of macroscopic systems: A set of empirical results which, during the years, have been progressively formalised, organised and synthesised into a number of ``laws'' which are completely independent on the physical substrate. Because of that, thermodynamics puts severe constraints on the behaviour of macroscopic systems. Most of the technological advances occurred during the industrial revolution have their roots in this simple, and yet powerful, fact. 

In the second half of the XIX century Maxwell, Boltzmann, Gibbs and others tried to connect these macroscopic laws with the multi-faceted nature of the microscopic world. The result of this attempt was Statistical Mechanics: a theory which provides mechanical roots to macroscopic concepts as temperature and heat. According to Statistical Mechanics, such macroscopic behaviour is ascribed to the emergence of a single condition: \emph{thermal equilibrium}. Indeed, the central postulate of the theory is that the state of the system has a specific functional form which we call \emph{thermal state}. The tools of statistical mechanics are then used to study macroscopic properties of complex systems, from this single assumption. Later, in the first half of the XX century, the rise of quantum mechanics forced us to adapt these tools to include quantum effects. The result was Quantum Statistical Mechanics, a notable improvement of the classical theory which expanded its domain of applicability. 
Both classical and quantum statistical mechanics enjoy a marvelous success in predicting and explaining the large-scale behaviour of several systems. 

Despite an undeniable success, it is well known that there are systems which do not thermalize, thus escaping such description. Two examples are \emph{integrable} and \emph{localised} systems. Moreover, the theory is not free from conceptual issues. In particular, a long-standing problem is that the evolution of isolated quantum systems is not compatible with the dynamical emergence of thermal states. We find ourselves in the situation where we have a theory which, despite being based on a seemingly wrong postulate, still works very well in a large number of cases. A possible way out of this problem is to question the practical utility of the concept of isolated systems. This is a legitimate route, based on the successful theory of Open Quantum Systems. However, it is opinion of the author that, from the conceptual point of view, this is not sufficiently satisfactory. Even if our system is interacting with a large environment, in principle we could still include the environment in our description and obtain an isolated system. Here we insist on dealing with isolated quantum systems. This path is undoubtedly more painful as it tackles a broader issue: understanding the interplay between coherent dynamics of complex quantum systems and their equilibrium behaviour. We choose this perspective as, beyond the purpose of the thesis, the outcomes of these investigations are expected to have technological consequences, for example for the task of building a quantum computer.

Moreover, in the last twenty years we witnessed a remarkable progress in the ability to manipulate quantum systems. Thanks to these technological developments, we witnessed a surge of interest for these questions, which now can be experimentally addressed. Modern experiments are able to probe the coherent dynamics of nearly-isolated systems, providing important data about equilibration and thermalization in quantum systems. More recently, the will to understand quantum equilibration and thermalization has met the need to investigate how the rules of thermodynamics are modified at the nanoscale, where quantum effects should not be neglected and we are far away from the macroscopic regime. These joints efforts gave birth to the field of ``Quantum Thermodynamics'', which is addressing the interplay between quantum theory and thermodynamics, both for foundational purposes and with the aim of pushing forward our technology. For all these reasons, to understand the conditions which lead to a thermal behaviour, and how to avoid them, is a topic of both conceptual and technological relevance.\\

With this perspective in mind, in the first part of the thesis we argue that the role played by observables in the emergence of thermal equilibrium is often overlooked. Almost all known approaches to the foundations of statistical mechanics focus on the state of a system. While a statement about the form of the state will always reflect on the observables, the converse is not generally true. Thus, observing thermal equilibrium properties for a few observables does not mean that the whole state of the system is thermal. In practical experiments, and sometimes even in our numerical simulations, our conclusions are based on the behaviour of a few observables. In order to trace back such behaviour to a specific form of the whole state of the system we need to have access to all the matrix elements of the density matrix. While this can be done for relatively small systems, by changing the observable to measure, it is concretely impossible for systems of modest sizes, which still fall in the category of microscopic systems. For example, a pure state of $L$ spin-$1/2$ has $2 \times 2^L -1$ linearly independent real parameters. This number grows exponentially with the size of the system, while the number of observables we can measure in a laboratory is usually very limited and certainly it does not grow exponentially with the size of the system. We must accept that our experimental observations provide only very little information about the state of the system. Thus, the information we can gather, when we have access to one observable, is given by the probability distribution of the eigenvalues of the observable. For this reason, the emergence of thermal equilibrium for an observable $A$ should be ascribed to the behaviour of its eigenvalues probability distribution. This picture presents no contradiction with the unitary dynamics of isolated systems. A more detailed analysis of this issue will be provided in Chapter \ref{Ch3}, where we give a notion of thermal equilibrium which takes into account the fact that we can only measure some observables and not the whole state of the system.\\

Thus, in the first part of this thesis we give such notion of thermal equilibrium and study the properties that it heralds. We are interested in its physical relevance to address the emergence of thermal behaviour in isolated quantum systems. The results are then put into a more general perspective by studying the relation with ``Pure-states quantum statistical mechanics'': a theory which aims at putting solid foundations to statistical mechanics and thermodynamics. The theory is based on three conceptual pillars: the ``Dynamical Equilibration'' approach, the ``Eigenstate Thermalization Hypothesis'' and the ``Typicality Approach''. The first one aims at understanding the mechanisms behind equilibration of quantum systems, irrespectively of the possible thermal nature of the equilibrium. The second one is a hypothesis, based on intuition from quantum chaos, which could concretely explain how thermal equilibrium emerges. The third one is a set of statistical tools which can be used to argue why thermalization seems to be a macroscopically ubiquitous phenomenon. The research presented in this thesis is mainly focused on the latter two approaches and an introduction to the tools of pure-states statistical mechanics is given in Chapter \ref{Ch2}.

The Eigenstate Thermalization Hypothesis is often introduced as a property of the energy eigenstates. Loosely speaking, it ascribes the emergence of thermal equilibrium to the fact that thermal properties emerge already at the level of individual energy eigenstates. Despite being numerically well-corroborated, at this stage the approach is still not really useful for making predictions. The reason is that the Eigenstate Thermalization Hypothesis is a technical condition which depends on the exact form of the energy eigenstates. Without direct access to the energy eigenstates there is no concrete way to assess wether a given observable will satisfy the Eigenstate Thermalization Hypothesis or not. As previously suggested, a problem with this picture is that the role of the observable under scrutiny is often overlooked. We believe that a better way to look at the ETH is to regard it as a relational property of the observable and the energy basis. With this mindset, in the thesis we aimed for a characterisation of the observables which should satisfy the Eigenstate Thermalization Hypothesis. The notion of Hamiltonian Unbiased Observables, presented in the Chapter \ref{Ch3} and studied in Chapters \ref{Ch4} and \ref{Ch5}, should be understood as a first approximation to the solution of this problem. Further work in this direction, to refine such characterization, is currently ongoing. We will comment on this in the conclusive chapter.\\

The typicality approach gives rise to the ``local thermalization paradigm'', summarised in Chapter \ref{Ch2}. According to this picture, even if the whole system is isolated and the state is pure, the reduced state of a small subsystem can still be very close to a thermal state, due its entanglement with the rest of the system. Performing a statistical analysis on the Hilbert space and using the ``concentration of measure phenomenon'' one can argue that this behaviour is expected to be true in the overwhelming majority of cases. The basic intuition behind the statement is the same as the central limit theorem: the fluctuations of smooth functions around their average are suppressed in high-dimensional spaces. This behaviour can be formalised by means of a theorem called ``Levy's Lemma'', which we give in Chapter \ref{Ch2}.

The most striking consequence is that, as long as we are interested in the local properties of a high-dimensional Hilbert space, the maximum entropy assumption for the overall state of the system is a good assumption, in almost all cases. This justifies the use of maximum entropy principle for the study of local properties of macroscopic system. In the standard setup of statistical mechanics, given by weakly interacting quantum systems, this gives local thermalization. However, Levy's Lemma allows for a more general characterization, beyond thermal equilibrium. It provides a tool to study the local properties of macroscopic systems subject to an arbitrary constraint, without necessarily referring to an Hamiltonian. Because of that, it is an approach which is particularly suitable to study systems whose dynamics is not described by a standard Hamiltonian operator. Classical gravity falls in this category, as the solutions of the dynamical problem are given by the solution of a constraints which is called ``Hamiltonian constraints''. This is caused by the fact that general relativity is a theory which is invariant under diffeomorphism. This feature is inherited by Loop Quantum Gravity. This is a tentative theory of quantum gravity which builds on the early attempts by Dirac to quantize General Relativity. Its sates belong to the so-called spin network Hilbert space. It describes the degrees of freedom a quantised geometry and its states have a clean geometrical interpretation as a collection of adjacent polyhedra. 

Building our intuition on that, we exploited the concentration of measure argument to perform a statistical analysis on the spin network Hilbert space. We focus on the local properties, in the macroscopic regime. These are important, especially at the classical level, as they are the ones we are able to experimentally probe. For this reason the ``typical'' results, in the semiclassical regime, are put in comparison with the behaviour expected from general relativity, looking for analogies and discrepancies. The underlying idea is that the ``typical'' properties of an Hilbert space should be preserved in the macroscopic regime. If not exactly, they should provide a decent approximation. In such regime it is then easier to study the emergence of the semiclassical phenomenology. In Chapter \ref{Ch7} we first extended the technique to address gauge-invariant Hilbert spaces and we then applied it to study the behaviour of a small patch of surface belonging to a larger surface. In Chapter \ref{Ch8} we apply the argument to a slightly different setup, where we consider the boundary of a large $3D$ region of quantum space. The consequences of the typicality argument have a striking similarity with the condition for the emergence of a black hole in the classical theory. When there is too much entanglement between the interior and the boundary, the latter exhibits properties similar to the ones of an horizon. Moreover, the threshold for the emergence of typicality, in the classical regime, resembles the threshold for the creation of a black hole.\\

Eventually, Chapter \ref{Ch9} closes the manuscript. There, we give a quick summary and all the results are put into perspective. Moreover, we draw some conclusions and give an outline for future works.


\chapter{Introduction}\label{Ch1}

The purpose of this chapter is to introduce some background material. We focus on the basic tools of quantum mechanics and quantum information, mainly to establish the notation. Our main references for quantum mechanics are: The book ``Quantum Mechanics'', by C. Cohen-Tannoudji, B. Diu and F. Laloe \cite{Cohen-Tannoudji1977} and the book ``Modern Quantum Mechanics'' by J.J. Sakurai\cite{Sakurai1994}. Moreover, our main references for quantum information theory are the book ``Geometry of Quantum States'' by I. Bengtsson and K. Zyczkowski \cite{Bengtsson2008} and the classic ``Quantum Computation and Quantum Information'' by Nielsen and Chuang\cite{Nielsen2010}. 

\section{States}\label{sec:states}

The space of the states of an isolated quantum system has the structure of an Hilbert space $\mathcal{H}$. Throughout the thesis we will deal with Hilbert spaces which have finite dimension $D$, unless otherwise stated. A state of the system is specified by giving an element $\ket{\psi}$ of $\mathcal{H}$ which describes the information that we can have about the state of the quantum system. To every state $\ket{\psi}\in \mathcal{H}$ we associate a unique dual vector $\bra{\psi}$ such that $\left( \ket{\psi} , \ket{\phi}\right) =  \braket{\psi}{\phi}$ defines the scalar product between $\ket{\psi}$ and $\ket{\phi}$. This can always be done thanks to the fact that the space of the states is an Hilbert space. The scalar product can be used to define a norm, which in turn can be used to define a metric over the Hilbert space of the states. The norm function $|| \cdot ||$ is defined as the square root of the scalar product of a vector with itself $||\ket{\psi}||=\sqrt{\braket{\psi}{\psi}}$ and we will consider only vectors which have unit norm: $||\ket{\psi}||=1$. For two arbitrary states $\ket{\psi}$ and $\ket{\phi}$ the quantity $|\braket{\psi}{\phi}|^2$ is called amplitude and it has the follwing probabilistic interpretation (``Born rule''): If we prepare the system in the state $\ket{\phi}$, the probability of seeing $\ket{\psi}$ when we measure the system is the amplitude $|\braket{\psi}{\phi}|^2$. The statement is clearly symmetric with respect to the swap of the two states. We say that two states are orthogonal when $\braket{\psi}{\phi}=0$. This means that if we prepare the system in one of the two states, there is zero probability to find the system in the other one, which means that we have the ability to distinguish them. 

An arbitrary basis $\mathcal{B} \coloneqq \left\{\ket{n}\right\}_{n=1}^D$ for $\mathcal{H}$, where $D$ is the dimension of the Hilbert space, is given by a set of $D$ completely distinguishable states $\braket{n}{m}=\delta_{nm}$, where $\delta_{ab}$ is the Kronecker delta between the two integers $n$ and $m$. Given an arbitrary basis $\mathcal{B}$, we can always decompose a state into its component in $\mathcal{B}$: $\ket{\psi} = \sum_{n=1}^D d_n \ket{n}$, where the coefficients are complex $d_n\in \mathbb{C}$ and sum up to one because of the normalisation of the state: $\braket{\psi}{\psi} = \sum_{n=1}^D$. An analogue expansion exists for the dual vector: $\bra{\psi} = \sum_{n=1}^D d_n^{*} \bra{n}$, where $d_n^{*}$ is the complex conjugate of $d_n$. 

In order to account for classical probabilities a state can be represented by means of a density operator $\rho$. This is a linear operators on $\mathcal{H}$ which satisfies the following three conditions. First, normalization: $\mathrm{Tr} \rho =1$. This is due to the probabilistic interpretation mentioned before. Second, $\rho$ must be self-adjoint: $\rho=\rho^\dagger$, where $\rho^\dagger$ is the adjoint of $\rho$. Third, non-negativity of its eigenvalues $\rho \geq 0$. We say that a state is pure when $\rho_{\mathrm{pure}} = \ket{\psi}\bra{\psi}$ is given by the outer product of $\ket{\psi}$ and $\bra{\psi}$. Otherwise, a state is mixed (or a mixture) when it is a convex sum of several pure states $\ket{\psi_i}$: $\rho_{\mathrm{mix}} = \sum_{i} p_i \ket{\psi_i}  \bra{\psi_i}$.

If we have a composite system made by $N$ subsystems, each described by an Hilbert space $\mathcal{H}_i$, with $i=1,\ldots, N$, the total Hilbert space will be the tensor product of $\mathcal{H} = \otimes_{i=1}^N \mathcal{H}_i$. On this space we can define the ``partial trace operation'' which maps a global state $\ket{\psi} \in \mathcal{H}$ into its marginal state on a smaller portion of the whole system. For example, let's assume  to have $\mathcal{H}= \mathcal{H}_A \otimes \mathcal{H}_{B}$ with $D=D_A D_B$ and $D_{A}$ and $D_B$ are the respective dimensions of the subspaces. Then for any $\rho \in \mathcal{H}$ we have:
\begin{align}
&\rho_A \coloneqq \Tr_{B} \rho = \sum_{i_B=1}^{D_B} \bra{i_B} \rho \ket{i_B} &&\rho_B \coloneqq \Tr_{A} \rho = \sum_{j_A=1}^{D_A} \bra{j_A} \rho \ket{j_A}
\end{align}
where $\mathcal{B}_A \coloneqq \left\{\ket{j_A}\right\}_{j_A=1}^{D_A}$ and $\mathcal{B}_B \coloneqq \left\{\ket{i_B}\right\}_{i_B=1}^{D_B}$ are two arbitrary basis in $\mathcal{H}_A$ and $\mathcal{H}_B$. The result does not depend on the specific basis that we choose to perform the computation. 

\section{Observables}\label{sec:observables}

An observable $A$ in quantum mechanics is represented by a self-adjoint operator acting on the Hilbert space of the states $\mathcal{H}$. When we measure the quantity $A$ on a system the outcome of the measurement process can only be one of the eigenvalues of such operators. From the Born rule, the expected value of an observable $A$ in the state $\rho$ is $\MV{A}_\rho \coloneqq \Tr \rho A$. Using the spectral decomposition of $A$ we can write 
\begin{align}
&A = \sum_{j=1}^{n_A} a_j A_j = \sum_{j=1}^{n_A} \sum_{s_j=1}^{d_j} a_j \ket{a_j ,s_j}\bra{a_j,s_j} \,\, .
\end{align}
Here $n_A$ is the number of distinct eigenvalues of $A$; $A_j$ is an operator which projects any state onto the subspace $\mathcal{H}_j \subset \mathcal{H}$ of the whole Hilbert space where $A$ has a definite eigenvalue $a_j$; the index $s_j$ is present whenever $n_A \le D$ and it accounts for possible degeneracies in the eigenvalues of $A$. Whenever they are present the subspace $\mathcal{H}_j$ has dimension larger than $1$ and the vectors $\ket{a_j,s_j}$ provide a basis for $\mathcal{H}_j$. Whenever $n_A = D$ the eigenvectors of $A$ provide a basis for the whole Hilbert space which we will call the ``eigenbasis of $A$''. The operators $A_j$ belong to a specific class of operators called ``projectors''.  An operator $\Pi$ is a projector when $\Pi^2 = \Pi$ so that it has eigenvalues which are either zero or one. If $n_A=D$ the eigenvectors of $A$ are unique and we can use the eigenvalues to label them. In this case $A_j$ has only one non-zero eigenvalue and we can write $A_j=\ket{a_j}\bra{a_j}$. When the $a_j$ has degeneracy $d_j$ we have $A_j = \sum_{s_j=1}^{d_j} \ket{a_j , s_j} \bra{a_j , s_j}$. For arbitrary degeneracies we have orthogonality $A_i A_j = \delta_{ij} A_j$ and completeness $\sum_{i=1}^{n_A} A_i = \mathbb{I}$, where $\mathbb{I}$ is the identity operator. \\

\section{Measurements}\label{sec:measurements}

What happens to the state $\rho$ of the system after we perform a measurement? If the observable is $A$, the measurement postulate of quantum mechanics dictates that the eigenvalue $a_j$ will occur with probability $p_\rho(a_j) \coloneqq \Tr(\rho A_j)$ and the state after the measurement will be 
\begin{align}
&\rho_j \coloneqq \frac{A_j \rho A_j}{p_\rho(a_j)} \,\, .
\end{align}
If we know that the measurement has been performed but we do not know the actual outcome of the measurement processs the state of the system will be described by the following mixture: 
\begin{align}
&\tilde{\rho} \coloneqq \sum_{i=1}^{n_A} p_\rho(a_j) \rho_j = \sum_{i=1}^{n_A} A_j \rho A_j \,\, .
\end{align}

\section{Quantum bits}\label{sec:qubits}

The first non-trivial Hilbert space is the one which has dimension $2$. This describes the space of the states of a two-states system which we now call ``Qubit''. It can be used, for example, to describe the spin-$1/2$ degrees of freedom of an electron. Given that the total angular momentum is fixed to $1/2$, a basis for the Hilbert space is given by the eigenstates of the angular momentum along the $z$ direction: $\left\{ \ket{\uparrow_z} ,\ket{\downarrow_z} \right\}$. We will often refer to this basis as the ``computational basis'' and use the following notation $\left\{ \ket{\uparrow_z} ,\ket{\downarrow_z} \right\} = \left\{ \ket{0}, \ket{1}\right\}$. The set of Pauli operators 
\begin{align}
& \sigma_x = \left( \begin{array}{ll}
 0 & 1\\
 1 & 0
  \end{array} \right) , &&& \sigma_x = \left( \begin{array}{ll}
 0 & -i\\
 i & 0
  \end{array} \right), &&&& \sigma_z = \left( \begin{array}{ll}
 1 & 0\\
 0 & -1
  \end{array} \right) 
\end{align}
together with the identity operator, provides a basis for the space of the quantum observables acting on the qubit Hilbert space. This means that we can use them to decompose any density matrix as
\begin{align}
&\rho = \frac{\mathbb{I}}{2} + \frac{1}{2}\sum_{\gamma=x,y,z} b^\gamma \sigma_\gamma \,\,\,
\end{align}
The three numbers $(b^x,b^y,b^z)$ are real and they belong to the $3D$ euclidean space. In this way we can visualize the state of a qubit as a vector $\vec{b}$ from the origin to a point inside the $2$-sphere. This is embedded in the three-dimensional euclidean space and has norm $\sum_{\gamma=x,y,z} (b^\gamma)^2 = \vec{b} \cdot \vec{b} \leq 1$. The vector $\vec{b}$ is called Bloch vector and the sphere is the Bloch sphere. The set of pure states can be identified with the surface of the sphere and every point inside represents a mixed state. The center of the sphere is the maximally mixed state $\frac{\mathbb{I}}{2}$. Moreover, if we endow the space $\mathcal{B}(\mathcal{H})$ of bounded operators on $\mathcal{H}$ with the Hilbert-Schmidt scalar product $\left\langle A, B\right\rangle \coloneqq \Tr (AB^\dagger)$ we have, for two arbitrary states $\rho$ and $\sigma$
\begin{align}
&\Tr(\rho \sigma) = \frac{1}{2} + \frac{1}{2} \, \vec{b}_\rho \cdot \vec{b}_\sigma \,\, .
\end{align}
Here $\vec{b}_\rho$ and $\vec{b}_\sigma$ are the Bloch vectors of $\rho$ and $\sigma$, respectively. This means that the scalar product between two states is mapped into the euclidean scalar product between the respective Bloch vectors, plus a $1/2$ constant contribution. \\

The fact that observables are represented by matrices implies that action of the composition of two of them, in principle, depends on the order. For example, $\sigma_x \sigma_y \ket{\psi} \neq \sigma_y \sigma_x \ket{\psi}$. If we define the commutator as $[X,Y] \coloneqq XY - YX$, this means that observables in general will not commute. For example it is known that $[\sigma_i , \sigma_j] = i \epsilon_{ijk} \sigma_k$ with $i,j,k=x,y,z$ and $\epsilon_{ijk}$ being the $3D$ totally antisymmetric Levi-Civita tensor. The non-commutativity between observables is one of the most striking feature of quantum mechanics and it gives rise to a highly  non-classical behaviour of quantum mechanics. For example, if we are in an eigenstate of $\sigma_x$ with eigenvalue $+1$, we have complete certainty that if we measure the spin along the $x$ axis we will find $+1$. However, due to the non-commutative character of the observables, this also implies that we are completely ignorant about the outcome of the measurement along $y$ and $z$. Indeed, if we assume that our state is $\ket{\uparrow_x}$ we have $p_{\uparrow_x}(\uparrow_z) = p_{\uparrow_x}(\downarrow_z) = p_{\uparrow_x}(\uparrow_y) = p_{\uparrow_x}(\downarrow_y) = \frac{1}{2}$, where we used the notation $p_\rho(a_j)$ for the probability of observing $a_j$ when the state is $\rho$.

\section{Time-evolution of isolated quantum systems}\label{sec:time}

Isolated quantum systems evolve in time via the dynamics generated by the Schroedinger equation:
\begin{align}
&i \hbar \frac{d}{dt} \ket{\psi} = \Ham \ket{\psi} \,\, .
\end{align}
Here $\Ham$ is the Hamiltonian of the system and it is called the ``generator'' of the dynamics, borrowing the terminology from Group Theory. Indeed, the most general solution of the Schroedinger equation is given through the action of a time-dependent unitary operator $U(t)$ on the initial state $\ket{\psi(0)}$:
\begin{align}
&\ket{\psi(t)} = U(t) \ket{\psi(0)} && U(t) \coloneqq e^{- \frac{i}{\hbar} \Ham t}
\end{align}
An operator $U$ is unitary when $UU^\dagger = U^\dagger U = \mathbb{I}$. Since $U(t)$ maps the initial state $\ket{\psi(0)}$ into the state at a generic time $\ket{\psi(t)}$ we call it ``propagator''. For the density matrix, this means that the evolution is generated by the commutator with the Hamiltonian
\begin{align}
&i \hbar \frac{d}{dt} \rho = [\Ham , \rho] && \rho(t) = U(t) \rho(0) U^\dagger(t) \,\, .
\end{align}
Generic Hamiltonians may have degenerate eigenvalues. However, for reasons which will become more clear in the next chapter, we focus on the case where $\Ham$ has no degeneracies. This means that the spectral decomposition of $\Ham$ reads
\begin{align}
&\Ham = \sum_{n=1}^D E_n \,\, \mathbb{E}_{nn}  && \mathbb{E}_{nm} \coloneqq \ket{E_n}\bra{E_m}  \,\, .
\end{align} 
Because of the unitary form of the propagator there are always $D$ conserved quantities. As we will see, their physical meaning is the probability distribution of the energy eigenvalues, which does not change in time. If we expand the density matrix at time $t$ in the eigenbasis of $\Ham$ we have
\begin{align}
&\rho(t) =  \sum_{n,m} \rho_{mn} e^{-\frac{i}{\hbar}(E_n-E_m)t} \mathbb{E}_{nm} = \sum_{n} \rho_{nn} \mathbb{E}_{nn} + \sum_{n \neq m} \rho_{mn} e^{-i \omega_{nm}t} \,\, \mathbb{E}_{nm}\, \, ,
\end{align}
with $\omega_{nm} \coloneqq \frac{E_n-E_m}{\hbar}$. The first term, which is constant in time, is called ``Diagonal Ensemble'' $\rho_{\mathrm{DE}}$. It is a mixed state which contains only the information about the state of the system which is preserved through the the unitary evolution. This is given by the diagonal part of $\rho(t)$, in the Hamiltonian eigenbasis, which defines the energy eigenvalues probability distribution $p_t(E_n) \coloneqq p_{\rho(t)}(E_n)$: 
\begin{align}
&p_t(E_n) = \Tr \rho(t) \mathbb{E}_{nn} = \bra{E_n} \rho(t) \ket{E_n} = \bra{E_n} \rho(0) \ket{E_n} = p_0(E_n) = \rho_{nn}
\end{align}
Throughout the thesis we deal with isolated systems with unitary dynamics. For this reason, we can drop the time index in the energy probability distribution $p_t(E_n)=p_0(E_n)=p(E_n)$. Thanks to the no-degeneracy assumption on $\Ham$ we can also see that $\rde$ is equal to the infinite-time-averaged state:
\begin{align}
&\lim_{T\to \infty} \frac{1}{T} \int_{0}^T  \rho(s) ds = \sum_{n=1}^D \rho_{nn} \,\,\mathbb{E}_{nn} = \rde \,\, ,
\end{align}
due to the fact that
\begin{align}
&\lim_{T\to \infty} \frac{1}{T} \int_{0}^T e^{-i \omega_{nm}s} ds = \delta_{nm}
\end{align}
if there are no degeneracies. This is an important property of the dynamical evolution of isolated quantum systems as, among other things, it implies two things. First, the dynamic is periodic. This means that, strictly speaking, the whole state never converges to a stationary state. Second, due to the existence of the set $\left\{p(E_n)\right\}_{n=1}^D$  of $D$ conserved quantities the state will never forget completely about its initial conditions. This entails that, strictly speaking, there can be no thermalization of an isolated quantum system. We see immediately that the phenomenon of ``quantum thermalization'' is structurally different from its classical counterpart. We will expand on this in the next chapter. Now we focus on giving a few tools from quantum information theory which will be used throughout the thesis.

\section{Entropies}\label{sec:entopies}

Suppose a random variable $X$ can assume a certain number $i=1,\ldots,N$ of values with probability distributions $P=\left\{p_i\right\}$. A fundamental quantity in information theory is given by the Shannon entropy $H[P] \coloneqq - \sum_{i=1}^N p_i \log p_i$. This provides a measure of the spreading of $P$ over different values $i=1,\ldots,N$. When $p_i = \delta_{i,i_0}$ its value is zero, which means that we have complete knowledge of $X$. When $p_i=\frac{1}{N}$ the entropy attains its maximum value $H_\mathrm{max}=\log N$ and thus we have complete ignorance about $X$. A generalization of Shannon entropy is given by the ``Renyi entropies'' or ``$\alpha-$entropies'':
\begin{align}
&H_\alpha[P] \coloneqq \frac{1}{1-\alpha} \log \left( \sum_{i=1}^N p_i^\alpha \right) && \alpha \geq 0 \,\, , \,\, \alpha \neq 1
\end{align}
In the limit $\alpha \to 1$ we recover the Shannon entropy: $\lim_{\alpha \to 1} H_\alpha [P] = H[P]$. All these quantities can be generalised to the quantum realm.\\

Given an observable $A$ with $n_A$ distinct eigenvalues, a state $\rho$ defines the probability distribution $p_\rho(a_j)$ over the set of eigenvalues $a_j$ via $p_\rho (a_j) \coloneqq \Tr \rho A_j$. A measure of the spreading of this probability distribution is given by the Shannon entropy 
\begin{align}
&H_A[\rho]\coloneqq - \sum_{j=1}^{n_A} p_\rho(a_j) \log p_\rho(a_j)
\end{align}
We will often refer to this quantity as to ``Shannon Entropy of the Observable $A$'' or ``Observable Entropy''. Among all the possible basis of the Hilbert space, for each state there is one such that the density matrix has diagonal form: $\rho = \sum_{l=1}^D \lambda_l \ket{\lambda_l} \bra{\lambda_l}$.  The eigenvalues $\lambda_l$ have probabilistic interpretation as they give the probability to find the system in $\ket{\lambda_l}$ when the state is $\rho$. Indeed, because of the properties of a density matrix we have $\lambda_l \in [0,1]$ and  $\sum_l \lambda_l =1$. For an arbitrary state $\rho$ the von Neumann entropy $S(\rho)$ is defined as the Shannon entropy the eigenvalues of $\rho$: 
\begin{align}
&S(\rho) \coloneqq  - \sum_{l=1}^{D} \lambda_l \log \lambda_l = - \Tr \rho \log \rho 
\end{align}
A pure state is such that there is only one non-zero eigenvalue as $\rho_{\mathrm{pure}} = \ket{\psi}\bra{\psi}$ for some $\ket{\psi}$. For this reason the von Neumann entropy of a pure state is always zero. Both the Observable and the von Neumann entropy exploits the Shannon entropy functional. They can be generalised by using the Renyi-entropy functional form to obtain the quantum version of the Renyi-entropies:
\begin{align}
&S_\alpha(\rho) \coloneqq  \frac{1}{1-\alpha} \log \sum_{l=1}^{D} (\lambda_l)^\alpha = \frac{1}{1-\alpha} \log \Tr \rho^\alpha
\end{align}

\section{Quantum and classical distinguishability}\label{sec:distinguishability}

In classical information theory, if we have two probability distribution $P=\left\{p_i \right\}_{i=1}^N$ and $Q=\left\{ q_j\right\}_{j=1}^N$ there are quantities which are able to quantify how different are such distributions. Important examples are the ``Total Variation'' $T(P,Q)$, the ``Kullback-Leibler divergence'' (also called Relative Entropy) $H(P|\!|Q)$ and the ``Bhattacharyya coefficient'' $B(P,Q)$. Their definitions are:
\begin{align}
&T(P,Q) \coloneqq \sup_{i} |p_i - q_i| = \frac{1}{2} \sum_{i}|p_i - q_i| = ||P-Q||_1  \in [0,1]\\
& H(P|\!|Q) \coloneqq \sum_i p_i \log \frac{p_i}{q_i} \in [0, + \infty[\\
& B(P,Q) \coloneqq \sum_i \sqrt{p_i q_i} \in [0,1]
\end{align}
All these quantities are non-negative and attain the following extremal values if and only if $P=Q$: $T(P,P)=H(P|\!|P)=0$ and $B(P,P)=1$. The total variation is also a metric and therefore provides a notion of distance between probability distributions. They can all be extended to the quantum realm to address the distinguishability between two quantum states $\rho$ and $\sigma$. Their respective definitions are ``Trace Distance'' $D(\rho,\sigma)$, ``Quantum Relative Entropy'' $H(\rho|\!|\sigma)$ and ``Fidelity'' $F(\rho,\sigma)$:
\begin{align}
&D(\rho,\sigma) \coloneqq \frac{1}{2} || \rho-\sigma||_{\mathrm{Tr}} = \sup_{\Pi} |\Tr(\rho \Pi) - \Tr(\sigma \Pi)| \in [0,1]\\ 
&H(\rho|\!|\sigma) = \Tr \rho \left( \log \rho - \log \sigma \right) \in [0,+\infty]\\
&F(\rho,\sigma) = ||\sqrt{\rho}\sqrt{\sigma}||_{\mathrm{Tr}} = \Tr \sqrt{\sqrt{\rho} \sigma \sqrt{\rho}}
\end{align}
where $||A||_{\mathrm{Tr}} \coloneqq \Tr \sqrt{AA^\dagger}$ is called the ``Trace Norm'' and it is the norm function defined by Hilbert-Schmidt scalar product on the space of the bounded linear operators acting on the Hilbert space $\mathcal{H}$. 

\section{Quantum Entanglement}\label{Sec:Entanglement}

The concept of entanglement is one of the most striking feature of quantum mechanics\cite{Nielsen2010,Bengtsson2008}. It is a form of quantum correlations which pertains the composite nature of quantum systems and the tensor product structure of the Hilbert space. Here we will only need the concept of bipartite entanglement. Assuming that we have a Hilbert space $\mathcal{H}$ that we can partition in the following way $\mathcal{H} = \mathcal{H}_A \otimes \mathcal{H}_B$, with $D_A$ and $D_B$ the respective dimensions, a pure state $\ket{\psi} \in \mathcal{H}$ is said to be separable if it can be written as
\begin{align}
&\ket{\psi} = \sum_{i=1}^{D_A} a_i \ket{i_A}  \otimes \sum_{k=1}^{D_B} b_k \ket{k_B}  \,\, . \label{eq:sep}
\end{align}
If the state is not separable as Eq.(\ref{eq:sep}) we say that $\ket{\psi}$ has bipartite entanglement between $A$ and $B$. In the case of pure states $\ket{\psi}$, the amount of bipartite entanglement between $A$ and $B$ is quantified by the entanglement entropy $E_{AB}(\ket{\psi})$. This is the von Neumann entropy of the reduced state of $A$ or $B$: 
\begin{align}
&E_{AB}(\ket{\psi}) = S_{\mathrm{vN}}(\rho_A(\psi)) = S_{\mathrm{vN}}(\rho_B(\psi))
\end{align}
where $\rho_A(\psi)$ and $\rho_B(\psi)$ are the reduced density matrices obtained via partial trace operation on $\ket{\psi}$. Entanglement allows the entropy of the reduced state to be higher than the entropy of the whole system. Thus, for pure states, we measure of the amount of entanglement between the two subsystems through the entropy of the marginals.

\section{Quantum thermal equilibrium}\label{Sec:EqTh}

Here we briefly review how to describe thermal equilibrium for a quantum system. First, we say that a system is at equilibrium, or in a stationary state, when its state does not change in time. Due to the specific form of the unitary evolution, a stationary state must be diagonal in the Hamiltonian eigenbasis. Moreover, since the energy probability distribution does not change in time, we can always associate a stationary state, the diagonal ensemble $\rde$, to any given initial state $\rho$. This is obtained by applying the ``dephasing map'' \cite{Bengtsson2008,Nielsen2010} to the initial state:
\begin{align}
&\ket{\psi} = \sum_{n}c_n \ket{E_n}  \quad \longrightarrow \quad \sum_{n=1}^D \mathbb{E}_{nn} \, \ket{\psi}\bra{\psi} \, \mathbb{E}_{nn} = \sum_{n=1}^D |c_n|^2 \,\, \mathbb{E}_{nn} = \rde
\end{align}
Second, a quantum system is at \emph{thermal} equilibrium if its state belongs to the class of the so-called ``Gibbs ensembles''. These are stationary states which result from \emph{Jaynes principle} of constrained maximisation of von Neumann entropy\cite{Jaynes1957,Jaynes1957a}. The three most notable examples are the ``microcanonical'' $\rho_{mc}(E,\delta E)$, the ``canonical'' $\rho_\beta$ and the ``grand canonical'' $\rho_{\beta,\mu}$ ensembles. According to quantum statistical mechanics, the microcanonical ensemble is apt to describe the equilibrium behaviour of an isolated quantum system whose energy lies in a small window $I_0 \coloneqq [E_0 - \frac{\delta E}{2}, E_0 + \frac{\delta E}{2}]$:
\begin{align}
&\rho_{\mathrm{mc}}(E_0 , \delta E) = \frac{P_{I_0}}{\Tr P_{I_0}} = \frac{1}{\mathcal{N}(E_0,\delta E)} \sum_{E_n \in I_0} \ket{E_n} \bra{E_n} \,\, ,
\end{align}
where $P_{I_0}$ is the projector onto the subspace defined by the window $I_0$ and $\mathcal{N}(E_0,\delta E) = \Tr P_{I_0}$ is the number of energy eigenstates with eigenvalues in $I_0$. 

The canonical ensemble is characterised by the maximum entropy state at fixed average value of the energy. It is meant to describe the behaviour of a quantum system which exchanges energy, but not matter, with a large environment at temperature $T = \frac{1}{k_B \beta}$, where $k_B$ is the Boltzmann constant:
\begin{align}
&\rho_{\beta} = \frac{e^{- \beta \Ham}}{\mathcal{Z}_\beta} && \mathcal{Z}_{\beta} \coloneqq \Tr e^{- \beta \Ham} = \sum_{n=1}^D e^{- \beta E_n} \,\, .
\end{align}
Eventually, the grand canonical ensemble is characterised by the maximum entropy state at fixed average value of the energy and of the number of particles. It can be used to describe the behaviour of a quantum system which can exchange energy and matter with an environment:
\begin{align}
&\rho_{\beta,\mu} = \frac{e^{- \beta ( \Ham - \mu \hat{N}) }}{\mathcal{Z}_{\beta,\mu}} && \mathcal{Z}_{\beta,\mu} \coloneqq \Tr e^{- \beta (\Ham - \mu \hat{N})} \,\, ,
\end{align}
where $\mu$ is the chemical potential and $\hat{N}$ is the operator which counts the number of particles in the system. The problem with this characterization of thermal equilibrium properties is that it seems to be not compatible with the dynamics generated by the Schroedinger equation. Indeed, these are all mixed states and unitary evolution preserves the von Neumann entropy of the whole system. Therefore, there is no unitary dynamics which can evolve an initially pure state $\ket{\psi_0}$ into a thermal state. Thus, while they are useful tools to describe the macroscopic behaviour of quantum systems, we still need to understand how this description is compatible with the unitary evolution of an isolated system.

\chapter{Pure-states statistical mechanics}\label{Ch2}


``Pure States Statistical Mechanics'' (PSSM) \cite{Lloyd1988} is a set of theoretical results and mathematical tools which aims at explaining how thermodynamics and statistical mechanics emerge, in isolated quantum systems. The theory is not yet a coherent and well understood set of statements but it is founded on three main conceptual pillars: the Eigenstate Thermalisation Hypothesis (ETH), the so-called Typicality Arguments and the Dynamical Equilibration Approach. For the purpose of this thesis only the first two are relevant. 

The number of papers dealing with the topic of thermalization and the microscopic foundations of thermodynamics is, without doubt, overwhelming. Moreover, this topic has been recently the subject of several reviews. For this reason here we do not attempt to give a complete review of the ideas and results involved. Rather, we choose a convenient narrative which aims at giving the reader the concepts which are necessary to understand the original part of this work, while giving an overview of the topic. After a brief introduction, where we reprise the ideas given in the Overview and we expand on the issue of isolated systems, we give a roadmap for the most recent review works dealing with this topic. This is done mainly for the sake of the interested reader who is not familiar with the field. Afterwards, we present two promising approaches to explain the emergence of thermal behaviour in quantum systems: the Eigenstate Thermalization Hypothesis and the Typicality Approach.

\section{Introduction}\label{sec:Ch2Intro}


Quantum Statistical Mechanics (QSM) is based on the assumption that the system under scrutiny is in a thermal equilibrium state, characterised by Jaynes Maximum Entropy Principle \cite{Jaynes1957,Jaynes1957a}. Given the large success of QSM, to assess the validity of its postulates and find criteria which guarantee its applicability is a long-standing issue, which needs to be addressed from the perspective of quantum mechanics. As anticipated in Section \ref{Sec:EqTh} and in the Overview, reconciling the nature of isolated quantum systems with Jaynes Principle is a non-trivial task. An initial state $\ket{\psi_0} \in \mathcal{H}$ will never evolve towards a thermal state under the unitary dynamics given by Schroedinger equation\cite{Kinoshita}. Therefore, either one accepts to be dealing with open quantum systems or one changes perspective and perseveres in the task of investigating the interplay between coherent dynamics and emergence of equilibrium. In this thesis we follow the second path and here we explain our choice.\\

Concretely, isolated systems are often an idealisation which is useful for practical purposes. In real experiments, achieving true isolation of a quantum system is in clear contrast with our ability to perform measurements on it. Despite that, the recent technological progresses in the manipulation of quantum systems now allow to experimentally probe regimes which are very close to such idealisation. As a result, the picture emerging from experimental investigations is in good agreement with the one given by theory: True equilibration and thermalization of the whole system are not possible, due to the unitary dynamics and its conserved quantities. This was undoubtedly clear theoretically, but also proven experimentally\cite{Kinoshita,Greiner2002}. In spite of that, this behaviour seems to be extremely fragile: a rare situation which, very often, we do not experience. The behaviour which is observed in most cases, for physical observables, is a quick equilibration to the prediction given by statistical mechanics \cite{Gring2012,Trotzky2012,Pertot2014}. This behaviour can not be ascribed to an interaction with the environment, for the following reason. 

When thermalization occurs because of decoherence with an environment, this happens on a time-scale $\tau_{\mathrm{dec}} \sim \hbar/\Gamma$, where $\Gamma$ is the strength of interaction between system and  environment. In many experimental setups to investigate isolated quantum systems\cite{Kinoshita,Greiner2002,Gring2012,Trotzky2012,Pertot2014,Bloch2012,Ritter2007} the estimated value of $\Gamma$ is very low and thus $\tau_{\mathrm{dec}}$ is very large. Despite that, in most cases we still observe quick equilibration and emergence of thermal expectation values, on a time-scale which is shorter than $\tau_{\mathrm{dec}}$. This excludes the interaction with an environment as the cause for the emergence of thermal behaviour. Moreover, if we numerically analyse the out-of-equilibrium dynamics of an isolated system, in several cases we observe the same behaviour. This provides sufficient evidence to say that some ``effective equilibration/thermalization'' are part of the phenomenology of isolated quantum systems and we still need to understand their underlying physics. 

The choice to deal with isolated system is therefore motivated by the broader issue of understanding the interplay between coherent dynamics of quantum systems and their equilibrium behaviour. As anticipated in the Overview, this issue is of both conceptual and practical relevance. For example, the emergence of thermalization puts severe constraints on the efficiency of the components in a quantum computer\cite{Lloyd2012}. If we want to avoid them, we should aim at preventing the components to thermalize. While the achievable degree of isolation of quantum systems has remarkably improved in the last twenty years, we still observe a general tendency of isolated quantum systems to local equilibration and thermalization. Only in the last few years, due to the rebirth of such foundational questions, we improved our understanding of these mechanisms. However, we are still far away from having a complete picture of these phenomena and further studies are certainly needed. \\

Due to the large amount of literature on this topic we had to make a choice about which arguments to present. Rather than giving a complete review, we decided to focus on the topics which are relevant for the original work presented in this thesis. To complete the exposition, and provide additional background material, here we present a roadmap for the recent reviews on the topic of equilibration and thermalization in isolated quantum systems.

\section{Roadmap of recent reviews}

The following works review different aspects of the interplay among quantum mechanics, statistical mechanics and thermodynamics, and complement the background material provided in this thesis. It is hope of the author that this can serve as roadmap for a reader who is willing to approach such an interesting research topic. 

\begin{itemize}
\item We believe Christian Gogolin's Master Thesis \cite{Gogolin2010} is the best starting point for the reader who is not familiar with the subject. It covers the basic information necessary to approach the issues of the topic. Moreover, his PhD Thesis \cite{Gogolin2014} and two reviews co-authored with J. Eisert \cite{Gogolin2016} and M. Friesdorf \cite{Eisert2015a} provide a wealth of information about the subject of ``Equilibration and Thermalization in closed Quantum Systems''. These four works summarise most of the results from the approach called ``Dynamical Equilibration''. Its main aim is to understand the dynamical mechanism underlying the equilibration of quantum systems. 

\item The book by J. Gemmer, M. Michel, and G. Mahler ``Quantum Thermodynamics''\cite{Gemmer2010} is also a very good starting point for a reader not familiar with the issues of the field. The second part of the book describes an approach to the quantum foundations of thermodynamics which goes under the name of ``Typicality approach''. This is very close in spirit to the approach described in Section \ref{sec:Ch2Typicality}. However, the technique used are slightly different from the ones presented there. 

\item The review by Yukalov ``Equilibration and thermalization in finite quantum systems'' \cite{Yukalov2011} offers a synthetic summary of some approaches to equilibration, thermalization, and decoherence in finite quantum systems. It provides information about the experimental efforts to investigate the coherent dynamics of quantum systems. Quasi-isolated and open quantum systems are also considered in this review. It is opinion of the author that such review provides a nice birds-eye-view on the topic and it is a good entry point for a reader already familiar with the basic issues of the field. 

\item The work  ``Colloquium: Nonequilibrium dynamics of closed interacting quantum systems''  by A. Polkovnikov, K. Sengupta, A. Silva, and M. Vengalattore\cite{Polkovnikov2011b}. It gives an overview of some theoretical and experimental insights concerning the dynamics of isolated quantum systems. It is mainly focused on the technique called ``quantum quenches'': Sudden changes in the parameters of the Hamiltonian which generate an out-of-equilibrium dynamics. It also covers the basics of Eigenstate Thermalization Hypothesis.

\item The editorial ``Focus on Dynamics and Thermalization in Isolated Quantum Many-Body Systems'' from New Journal of Physics by M. Cazalilla and M. Rigol\cite{Cazalilla2010} and the related focus issue. In the editorial the authors explains the significance of the papers published in the focus issue and offer a more general perspective about understanding of the dynamics of isolated quantum systems. A reader with knowledge of the issues of the field will find the editorial and the related focus issue to be a good starting point to study the promising approaches and some relevant literature. 

\item The review ``From quantum chaos and eigenstate thermalization to statistical mechanics and thermodynamics'' by D'Alessio, Y. Kafric, A. Polkovnikov and M. Rigol\cite{DAlessio2016} gives a pedagogical and detailed exposition of the Eigenstate Thermalization Hypothesis. The authors introduce the topic and highlight the connection with quantum chaos and random matrix theory. They also make the connection with  thermodynamics and provide concepts which are of paramount relevance to understand the emergence of thermal equilibrium in quantum systems.

\item The review ``Quantum Chaos and Thermalization in Isolated Systems of Interacting Particles'' by F. Borgonovi, F.M. Izrailev, L. F. Santos and V. G. Zelevinsky\cite{Borgonovi2016}. In this review the authors are focused on the emergence of chaotic behaviour in many-body quantum systems and the resulting thermal features. The approach stems from early investigations on atomic and nuclear physics and the relevance of chaos to model such systems. The techniques are extended and applied to study the thermalization behaviour of fermions, bosons and spin systems. 

\item The more recent review ``The role of quantum information in thermodynamics'' by J. Goold, M. Huber, A. Riera, L. del Rio and P. Skrzypczyk \cite{Goold2016} focuses on the interplay between quantum information theory and thermodynamics of quantum systems. This review is not focused on the foundations of statistical mechanics. Rather, it offers a broader perspective about the thermodynamic behaviour in quantum systems. It is a good entry point for a reader who is not familiar with the issues of understanding the quantum aspects of thermodynamics. 

\item The work ``On the foundation of statistical mechanics: a brief review'' by N. Singh\cite{Singh2013} provides a very good overview of the topic, where historical aspects are also taken into account. It focuses on  the conceptual aspects rather than giving concrete details on specific approaches. The main aim of the review is to show that the ergodic hypothesis is not necessary for the validity of statistical mechanics. 

\item The paper ``Thermalization and prethermalization in isolated quantum systems: a theoretical overview'' by T. Mori, T. N. Ikeda, E. Kaminishi and M. Ueda \cite{Mori2018}. This is the most recent review, as of March 2018, which has appeared on the topic. It covers modern results on the Eigenstate Thermalization Hypothesis and other intriguing relaxation phenomena as ``pre-thermalization'' and the relaxation dynamics of integrable systems. 

\item The ``Special issue on Quantum Integrability in Out of Equilibrium Systems'' published in the Journal of Statistical Mechanics: Theory and Experiments\cite{Calabrese2016}. This volume collects a large amount of papers on topics relevant to understand the dynamical behaviour of quantum systems. It is focused on integrable quantum systems, which are well-known examples of quantum systems which do not dynamically reach thermal equilibrium. It provides a wealth of information relevant to understand the conditions which lead to thermalization and also a large amount of modern literature on the topic.

\item Eventually, the ``Compendium of the foundations of classical statistical physics'' written by J. Uffink \cite{Uffink2006}. This work provides an historic survey of the classical work by Maxwell, Boltzmann and Gibbs in statistical physics. Moreover, it also reviews more modern approaches such as ergodic theory and non-equilibrium statistical mechanics. While the work is focused on classical systems, it is opinion of the author that several concepts and ideas presented in the review are generally important to understand the equilibration and thermalization phenomena. 

\end{itemize}

This is, by no means, a complete list of the relevant works. Rather, it is a list of review papers where most of the issues of the field are explained and where the promising approaches are presented in details.
In the remaining part of the chapter we will discuss two of these approaches: the Eigenstate Thermalization Hypothesis and the Typicality Approach.

\section{Eigenstate Thermalization}\label{sec:Ch2ETH}

The term \emph{Eigenstate Thermalization Hypothesis} (ETH) was coined by Srednicki\cite{Srednicki} in 1994. According to the ETH, the emergence of thermal equilibrium in quantum systems should be ascribed to the presence of some underlying chaotic behaviour. Due to the fact that the unitary evolution is a deterministic map for the energy eigenstates, the emergence of such chaotic behaviour is not due to the dynamical map, but to the structure of the energy eigenstates. Indeed, ETH is a hypothesis about properties of individual energy eigenstates of quantum many-body systems which was suggested by various results in quantum chaos theory. Among others, Berry Conjecture\cite{Berry1977} and Shnirelman Theorem\cite{Shnirelman1974} certainly played a major role in leading the way toward the formulation of the ETH. 
Their basic intuition is that, at the macroscopic scale, the energy eigenstates of a sufficiently complex quantum system will be so involved that their overlaps with the eigenstates of a physical observable can be effectively described by random variables. 
If the ETH is fulfilled, it guarantees thermalization for all observables that equilibrate.
Depending on how broad one wants the class of initial states that thermalize to be, the fulfillment of the ETH can also be a necessary criterion for thermalization \cite{Gogolin2016,DePalma2015a}.
The ETH has been criticized for its lack of predictive power, as it leaves open at least three important questions: what precisely are ``physical observables'', what makes a system ``sufficiently complex'' to expect that ETH applies, and how long it will take for such observables to reach the thermal expectation values\cite{Short2012,Garcia-Pintos2017,Reimann2016,Wilming2017,DeOliveira2018}.
For this reason, a lot of effort has been focused on numerical investigations that validate the ETH in specific Hamiltonian models and for various observables, often including local ones.
The ETH is generally found to hold in non-integrable systems that are not many-body localized and equilibration towards thermal expectation values usually happens on reasonable times scales \cite{Short2012,Reimann2016,Garcia-Pintos2017,DAlessio2016}. However, a satisfactory analytical explanation for ETH is still missing.


\subsection{Versions of the ETH}

Several different versions of the ETH have appeared in the literature. They are essentially different mathematical statements which aim at formalizing the same physical intuition: thermalization can emerge in quantum systems because, in the thermodynamic limit, the energy eigenstates of reasonable Hamiltonians exhibit thermal properties. Our goal here is to provide a reasonable clusterization of the most used versions of the ETH.
For historical reasons, the one proposed and studied by Srednicki in 1994-99\cite{Srednicki,Srednicki1996,Srednicki1999} has particular relevance and we review it with more details. On a more technical side, all versions of the ETH are statements about properties of large systems. In principle one would hence state the following in terms of families of systems of increasing size/particle number. To not over-complicate things, we do not make this explicit and instead implicitly assume that a limit of large system size exists and makes sense. Whence, the following are meant as statements about asymptotic scaling.

\begin{Srednicki}
It is an ansatz on the matrix elements of an observable when it is written in the Hamiltonian eigenbasis $\left\{ \Ket{E_{n}}\right\}$:
\begin{align}
&A_{m,n} \approx f_A^{(1)}(\overline{E}) \delta_{mn} + e^{-\frac{S(\overline{E})}{2}}f^{(2)}_{A}(\overline{E},\omega) R_{mn} \,\,\, , \label{eq:ETH} 
\end{align}
where $\overline{E} \equiv \frac{E_{n} + E_{m}}{2}$, $\omega \equiv E_{n} - E_{m}$ while $f_A^{(1)}$ and $f_A^{(2)}$ are smooth functions of their arguments. $S(\overline{E})$ is the thermodynamic entropy at energy $\overline{E}$ defined as $e^{S(\overline{E})} \equiv \overline{E} \sum_{n} \delta_{\epsilon}(\overline{E}- E_{n})$, where $\delta_{\epsilon}$ is a smeared version of the Dirac delta distribution. For each $(m,n)$ such that $m \neq n$, $R_{mn}$ is a complex random variable with zero mean and unit variance.\\
\end{Srednicki}
In words: diagonal elements of physically reasonable observables vary smoothly with the energy and off-diagonal matrix elements are exponentially small and randomly distributed. The function $f^{(2)}$ accounts for the fact that there can be some fine-grained structure in the off-diagonal matrix elements. This is what Srednicki argued to be fulfilled in a hard-sphere gas\cite{Srednicki}.

\begin{hypothesis}[Complete ETH] \label{hyp:originaleth}
  The matrix elements $A_{m,n} \coloneqq \bra{E_m}A\ket{E_n}$ of any physically reasonable observable $A$ with respect to the energy eigenstates $\ket{E_m}$ in the bulk of the spectrum of a Hamiltonian of a system with $N$ particles satisfy

  \begin{align}
&- \ln | A_{m+1,m+1} - A_{m,m} | \,\, , \,\, - \ln |A_{m,n}|  \in \mathcal{O}(N) 
  \end{align}

\end{hypothesis}
In words: Off-diagonal elements of physically reasonable observables and the differences between neighboring diagonal elements are exponentially small in the size of the system.
Srednicki's hypothesis belongs to this kind of ETH. Similar variants appeared for example in \cite{Rigol2012,Srednicki1996,Srednicki1999,Khemani2014,DAlessio2016}.

\begin{hypothesis}[Thermal ETH] \label{hyp:thermaleth}
  There exists a function $\beta:\R\to\R_0^+$ such that for any physically reasonable observable $A$ the expectation values $A_{m} \coloneqq \bra{E_m}A\ket{E_m}$ of $A$ with respect to the energy eigenstates $\ket{E_m}$ in the bulk of the spectrum of a Hamiltonian of system are close to thermal in the sense that 
  \begin{equation}
    |A_{m} - \Tr(A\,\e^{-\beta(E_m)\,H})/\Tr(\e^{-\beta(E_m)\,H})| \in \mathcal{O}(1/N) \,.
  \end{equation}
\end{hypothesis}
Whether a $1/N$ scaling should be required or whether one would be content with a weaker decay is debatable. Such formulations of ETH appeared for example in \cite{Rigol2008,Muller2015,DePalma2015a,Kim2014}, along with a rigorous proof of a statement that is closely related but weaker than Hypothesis~\ref{hyp:originaleth} for translation invariant Hamiltonians with finite range interactions.
\begin{hypothesis}[Smoothness ETH] \label{hyp:smoothnesseth}
  For any physically reasonable observable $A$ there exists a function $a:\R \to \R$ that is Lipschitz continuous with a Lipschitz constant $L \in \mathcal{O}(1/N)$ such that the expectation values $A_{m} \coloneqq \bra{E_m}A\ket{E_m}$ of $A$ with respect to the energy eigenstates $\ket{E_m}$ in the bulk of the spectrum of a Hamiltonian of system with $N$ particles satisfy
  \begin{equation}
     - \ln |A_{m} - a(E_m)| \in \mathcal{O}(N) \,.
  \end{equation}
\end{hypothesis}
In words: The expectation values of physically reasonable observables in energy eigenstates approximately vary slowly as a function of energy instead of widely jumping over a broad range of values even in small energy intervals.
The function $a(E)$ is often related to the average of $A$ over a small microcanonical energy window around $E$.
Similar statements of the ETH have been used for example in \cite{Rigol2008,Bartsch2017,Beugeling2014,Beugeling2015,Khodja2015,Polkovnikov2011b,Steinigeweg2013a,Torres-Herrera2014}. \\

Several other versions of the ETH and variations of the statements above can be found in the literature and there is a further level of diversification which needs to be mentioned:
All the statements above are intended to hold for all energy eigenstates in the bulk of the spectrum.
It is also possible to require them to hold only for all but a small fraction of these eigenstates, which somehow goes to zero in the thermodynamical limit.
Similar statements have been dubbed \emph{Weak ETH}\cite{Mori2016}.
Another related concept is the \emph{Eigenstate Randomization Hypothesis}\cite{Ikeda2011}, which states that the diagonal elements of physical observables should behave as random variables.
Together with an assumption on the smoothness of the energy distribution, this allows to derive a bound on the difference between the infinite-time and a suitable microcanonical average.

The main difference among the formulations of the ETH listed above is that the first one is also a statement about the off-diagonal matrix elements $A_{mn}$ while the other two pertain only to diagonal matrix elements $A_{mm}$.
We believe it is important to highlight this aspect because the off-diagonal matrix elements contribute in a non-trivial way to the out-of-equilibrium dynamics of the observable \cite{Short2012,Garcia-Pintos2017,Reimann2016,Wilming2017,DeOliveira2018,Linden,Mondaini2017a}.
This is the reason why we (as others do \cite{Rigol2012}) consider the Complete ETH as more fundamental.
Hereafter, when we refer to ETH we will always refer to the technical statement of Complete ETH or ETH \ref{hyp:originaleth}.

\subsection{Thermalization according to the ETH}

On its own, the ETH does not guarantee the emergence of thermal expectation values. This is a nontrivial point which is often overlooked. For this reason, we review here the argument leading from ETH to the emergence of thermal expectation values. Let $\ket{\psi_0} = \sum_{n} c_n \ket{E_n}$ be a generic initial state of our many-body quantum system and let $\Ham$ be its time-independent Hamiltonian. Let's assume the observable $A$ satisfies the ETH. At a generic time $t \in \mathbb{R}_0^+$ we have
\begin{align}
&\MV{A(t)} \coloneqq \bra{\psi_t} A \ket{\psi_t} = \underbrace{\sum_{n} A_{n,n} |c_n|^2}_{T_1} + \underbrace{\sum_{n \neq m} c_n c_m^{*} e^{- i \frac{E_n - E_m}{\hbar}t} A_{n,m}}_{T_2} \,\,\, .
\end{align}\label{eq:evol}
We see a constant term $T_1$, which depends on the diagonal matrix elements $A_{n,n}$, plus a fluctuation term $T_2$, which depends on the off-diagonal matrix elements $A_{m,n}$. Because of ETH, each element in  the sum of $T_2$ is exponentially small in the size of the system. However, there is an exponentially large number of such elements. If the phases are coherently aligned, they can sum up to a non-negligible value. This accounts for out-of-equilibrium configurations. If this happens at $t=0$, as time passes this condition will not be satisfied, due to the time-evolving phases. After a sufficiently long time the term $T_2$ will become  negligible with respect to $T_1$. This is a dynamical dephasing mechanism which can explain how an observable starting in an out-of-equilibrium configuration can reach an equilibrium value, given by $T_1$. However, equilibration is only part of the thermalization process. The remaining part being independence on the initial conditions. We now show how an observable that satisfies ETH has an equilibrium value which does not depend on the initial conditions.

The state $\ket{\psi_0}$ can not be completely arbitrary, as we know that there are states which never equilibrate or thermalise. For example, the state $\frac{\ket{E_{n}} + \ket{E_m}}{\sqrt{2}}$ will always oscillate between $\ket{E_n}$ and $\ket{E_m}$ at frequency $\omega = \frac{E_n - E_m}{\hbar}$ and dephasing will never occur. However, to prepare such state is an impossible task, for a macroscopic many-body quantum system. When the size of the system increases the energy eigenstates become extremely dense and it is practically impossible, even for the most accurate and careful experimentalist, to single out two arbitrary energy eigenstates. This leads to the following assumption on the energy distribution $p_0(E_n) \coloneqq  |\braket{\psi_0}{E_n}|^2 = |c_n|^2$ of the class of initial states which we are going to consider. We will assume that contributions to $p_0(E_n)$ outside a small energy window, which we call $I_0$, will be negligible. Such window is centred around the average value $E = \sum_n p_0(E_n)E_n$ and its width can be evaluated using the variance $(\Delta E)^2 = \sum_n |c_n|^2(E_n - E)^2$. In this way $I_0 \coloneqq \left[ E - \frac{\Delta E}{2} , E + \frac{\Delta E}{2}\right]$ and 
\begin{align}
&\ket{\psi_0} = \sum_{n} c_n \KEn \approx \sum_{n \in I_0} c_n \KEn \,\,\, .
\end{align}
Therefore, the window $I_0$ is usually assumed to be macroscopically small ($\delta E/E \ll 1$) but sufficiently wide to host a large number of energy eigenstates. Using this assumption in combination with the diagonal part of ETH one can argue that the equilibrium value does not depend on the initial conditions:

\begin{align}
&\Tr \left(A \rho_{\mathrm{DE}} \right) \simeq \sum_{n \in I_0} A_{n,n} p_0(E_n) \simeq A_{n_0,n_0}\sum_{n \in I_0} p_0(E_n) \simeq \nonumber \\
&\qquad \qquad \,\,\, \simeq \frac{1}{N(E,\Delta E)}\sum_{n\in I_0} A_{n,n}= \Tr \left( A \rho_{\mathrm{MC}}\right),
\end{align}
where $N(E,\delta E)$ is the number of energy eigenstates within $I_0$. In the first passage we used the fact that $p_0(E_n)$ has non-negligible contribution only from energies belonging to $I_0$. In the second passage we used the diagonal part of ETH: within the small energy window $I_0$ the $A_{n,n}$ are essentially constant, up to an exponentially small error. So we took the value at the center ($n_0 \, : \min_{n \in I_0}|E_n - E| = |E_{n_0} - E|$) and brought it outside the sum. Using again the assumption discussed before, the sum of $p_0(E_n)$ within $I_0$ is almost $1$, since contributions outside $I_0$ are negligible. In the third passage, using again both the ETH and the small variance assumption, we recognize that this is equal to the expectation value computed on the microcanonical ensemble $\rho_{\mathrm{MC}}$.

\subsection{Summary on the ETH}

The ETH is an ansatz on the matrix elements of an observable, in the energy eigenbasis, which is expected from results in quantum chaos theory. The intuition behind ETH is that the thermal behaviour should emerge at the level of individual energy eigenstates. Concretely, this can be stated in different mathematical ways. Hereafter, where we talk about ETH we will always refer to Srednicki's ETH or ETH 1. The key points of this approach are the following. First, the magnitude of the off-diagonal matrix elements is exponentially small in the size of the system. However, the phases can be coherently aligned with the phases of the initial state, to give an out-of-equilibrium expectation value. If this happens at $t=0$, the time evolution will cause dephasing and equilibration to the predictions of the diagonal ensemble. This will not happen always. However, it is the behaviour that we expect in generic cases. Second, the diagonal matrix element $A^{\mathrm{ETH}}_{nn}$ have a smooth behaviour, as a function of the energy. This means that their variation with the energy can be appreciated only at energy differences which are exponentially large in the size of the system. If the initial state has an energy distribution which is sufficiently narrow, it will sample only values $A^{\mathrm{ETH}}_{nn}$ and $A^{\mathrm{ETH}}_{mm}$ which are exponentially close to each other. Because of that, the expectation value of $A^{\mathrm{ETH}}$ on the diagonal ensemble will be the same as the one on the microcanonical ensemble. 

The ETH has been checked numerically in several cases\cite{Steinigeweg2013,Mondaini2016,Khodja2015,Magan2016,Sorg2014,Steinigeweg2013a,Mondaini2017a,Khemani2014,Rigol2009,Kim2014,Beugeling2014,Ikeda2013,Ikeda2011,Beugeling2015,Bartsch2017,Konstantinidis2015,Dymarsky2018,Yoshizawa2018}, often for local observables. Despite being well-corroborated, we still lack an analytical understanding of this phenomenon. Because of that, the main issue raised against this approach is that it lacks of predictive power, as it leaves open at least three important questions: what precisely are ``physical observables'', what makes a system ``sufficiently complex'' to expect that ETH applies, and how long it will take for such observables to reach thermal expectation values.

\section{Typicality}\label{sec:Ch2Typicality}

Now we turn our attention to a different approach for the emergence of thermal behaviour in isolated quantum systems: the ``Typicality Approach''. The mathematical apparatus which has been used to develop this approach relies heavily on ideas and techniques from Information Theory. The key result is the proof and formalization of the following intuition: Although the pair `subsystem+environment' is isolated and in a pure state, the reduced state of the system can exhibit a thermal behaviour, due to the presence of entanglement between the subsystem and the environment. Different techniques  have been used to obtain similar ``typicality arguments''. Here we do not review all of them, as this is beyond the scope of the thesis. Rather, the goal is to provide the reader with the conceptual and technical tools necessary to understand the second part of this thesis: Applications to Quantum Gravity. For this reason, now we review the approach which exploits the Concentration of Measure Phenomenon \cite{Ledoux2001} and Levy's Lemma, which we give in Appendix \ref{App:Levy}. \\

\begin{figure}[h!]
\begin{center}
\includegraphics[scale=0.2]{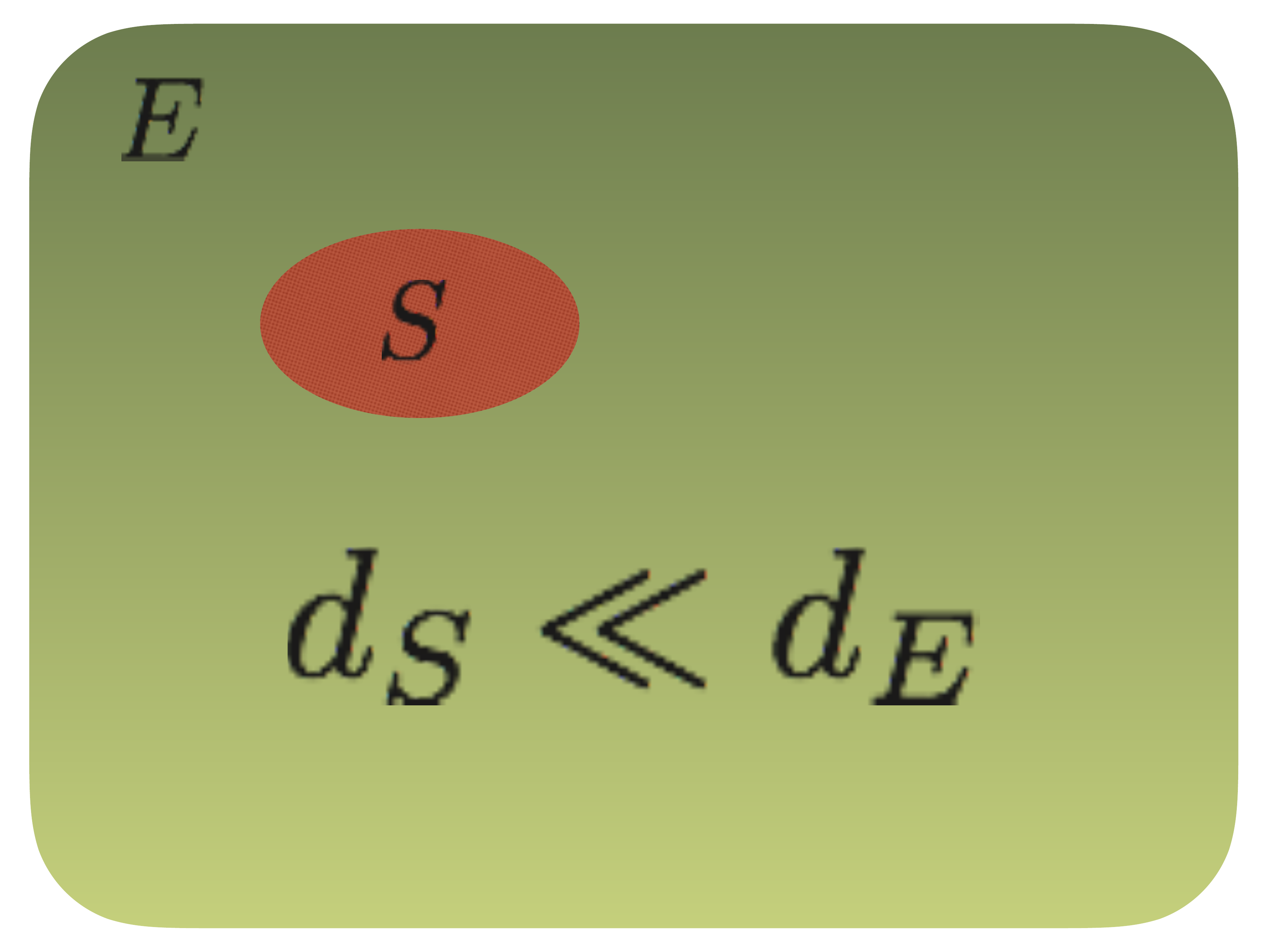}\caption[Setup for the ``Local thermalization'' paradigm]{Setup for the ``Local thermalization'' paradigm. A large quantum system, which we sometimes refer to as the ``universe'' is split in ``small system'' plus ``large environment''. The bipartition is performed in an asymmetric way to ensure that the dimension of the Hilbert space of the system $d_S$ is much smaller than the dimensions of the Hilbert space of the environment $d_E$.}\label{fig:lbit}
\end{center}
\end{figure}

Consider a large isolated quantum system, which we call ``universe'', and a bipartition into ``small system'' $S$ and ``large environment'' $E$. The universe is assumed to be in a pure state. We also assume that it is subject to a completely arbitrary \emph{global constraint} $\mathcal{R}$. In the standard context of statistical mechanics this will be the fixed energy constraint. Such constraint is realised by imposing that the  allowed states must belong to the subspace $\mathcal{H}_{\mathcal{R}}$ of the states of the total Hilbert space $\mathcal{H}_U$ which satisfy the constraint $\mathcal{R}$:
\begin{align}
\Hilb_{\mathcal{R}} \subseteq \Hilb_{U} = \Hilb_E \otimes \Hilb_S \,\,\, .
\end{align}
$\mathcal{H}_S$ and $\mathcal{H}_E$ are the Hilbert spaces of the system and environment, with dimensions $d_S$ and $d_E$, respectively. The bipartition is done in an asymmetric way and this is enforced by assuming that $d_S \ll d_E$. Given this setup, we define the microcanonical state of the universe $\mathcal{I}_{\mathcal{R}}$ as the mixed state which assigns equal probability to all the states which satisfy $\mathcal{R}$ and zero otherwise:
\begin{align}
&\mathcal{I}_{\mathcal{R}} \coloneqq \frac{\mathcal{P}_{\mathcal{R}} }{d_{\mathcal{R}}},
\end{align}
where $\mathcal{P}_\mathcal{R}$ is the projection operator on $\mathcal{H}_{\mathcal{R}}$ and $d_{\mathcal{R}} \coloneqq \mathrm{dim} \,\mathcal{H}_{\mathcal{R}}$. The \emph{canonical state} of the system $\Omega_S$ is defined as the partial trace of $\mathcal{I}_{\mathcal{R}}$ over the degrees of freedom of the environment:
\begin{align}
&\Omega_S  \equiv \mathrm{Tr}_E [\mathcal{I}_{\mathcal{R}}] \,. \label{eq:Canon}
\end{align}
The key result that we are going to exploit in Chapter \ref{Ch6} and \ref{Ch7} was firstly proven in Ref.\cite{Popescu2006}.
\begin{theorem} \label{thm:typicality}
For an arbitrary $\varepsilon > 0$, the volume of states $\mathrm{Vol} \left[  \phi \in \mathcal{H}_{\mathcal{R}} \, \vert \, D(\rho_S(\phi),\Omega_S) \geq \eta \right]$ such that the distance between the reduced density matrix of the system $\rho_S(\phi) \coloneqq \Tr_E \ket{\phi} \ket{\phi}$ and the canonical state $\Omega_S$ is more than $\eta$ can be evaluated as follows
  \begin{equation} 
\frac{\mathrm{Vol} \left[  \phi \in \mathcal{H}_{\mathcal{R}} \, \vert \, D(\rho_S(\phi),\Omega_S) \geq \eta \right] }{\mathrm{Vol} \left[ \phi \in \mathcal{H}_{\mathcal{R}} \right] } \leq 4\,  \mathrm{Exp} \left[ - \frac{2}{9\pi^3} d_{\mathcal{R}} \varepsilon^2 \right] \label{eq:VolStates}
  \end{equation}
\end{theorem}
where $D(\rho,\sigma)$ is the trace-distance on the convex set of the density matrices, while
\begin{align}
& \eta = \varepsilon + \frac{1}{2} \sqrt{\frac{d_S}{d_E^{\mathrm{eff}}}}, 
\end{align}
with the effective dimension of the environment defined as
\begin{align}
 d_E^{\mathrm{eff}} \equiv \frac{1}{\mathrm{Tr}_E \left[ \left( \mathrm{Tr}_S \mathcal{I}_{\mathcal{R}} \right)^2\right] } \geq \frac{d_{\mathcal{R}}}{d_S} \label{deff}.
\end{align}
The bound in Eq.\eqref{eq:VolStates} states that the fraction of states which are far away from the canonical state $\Omega_S$ more than $\eta$ decreases exponentially with the dimension of the ``allowed Hilbert space'' $d_{\mathcal{R}}$ and with $\varepsilon^2 = \left(\eta - \frac{1}{2} \sqrt{\frac{d_S}{d_E^{\mathrm{eff}}}}\right)^2$. This means that, as the dimension of the Hilbert space $d_{\mathcal{R}}$ grows, a huge fraction of states gets concentrated around the canonical state. The proof of the result relies on the \emph{concentration of measure phenomenon} and the key tool to prove it is Levy's lemma\cite{Ledoux2001}, which we summarize in Appendix \ref{App:Levy}.  .\\


To connect this result with statistical mechanics and thermodynamics we need to consider as constraint $\mathcal{R}$ the restriction to a fixed-energy slice. Indeed, we assume here that we have a large quantum system at fixed energy $E_0$. Hence, we restrict our space to the subspace of states which have a given energy $E_0$. Calling $\Ham_U$ the Hamiltonian of the universe,
\begin{align}
&\Ham_U = \Ham_S+\Ham_E+\Ham_{\mathrm{int}},
\end{align}
where $\Ham_S$ and $\Ham_E$ are the Hamiltonian of the system and environment respectively, and $\Ham_{\mathrm{int}}$ is the interaction. In the standard statistical mechanics context the system and the environment are weakly coupled and the density of states of the environment increases exponentially with the energy. With these assumptions, $\Omega_S$ can be computed with well-known techniques\cite{Greiner1995,Schrodinger1989} and the result is the Gibbs canonical ensemble:
\begin{align}
&\Omega_S \simeq \frac{e^{-\beta \Ham_S}}{\mathcal{Z}} &&\mathcal{Z} \coloneqq \Tr e^{-\beta \Ham_S}
\end{align}
As in the classical case, the value of the temperature is set by the value of the fixed energy $E_0$. This allows to turn Theorem \ref{thm:typicality} into a more familiar statement. Given that the total energy of the universe is fixed and approximately $E_0$, assuming weak interaction between system and environment and an exponentially large density of states for the environment, almost every pure state of the universe is such that the reduced state of the system is very close to the Gibbs' state.\\

It is important to stress that, on its own, the theorem does not say anything about the explicit form of $\Omega_S$. Rather, it shows that, within $\mathcal{H}_{\mathcal{R}}$, large fluctuations of $\rho_S(\phi)$ around $\rho_S(\phi)\approx \Omega_S$ are exponentially rare. Given an arbitrary constraint $\mathcal{R}$ and the respective microcanonical state $\mathcal{I}_{\mathcal{R}}$, computing $\Omega_S$ is a highly non-trivial task, which does not pertain Theorem \ref{thm:typicality}. The result that $\Omega_S$ is a thermal state depends on the specific choice of the constraint $\mathcal{R}$ and on other assumptions, which may not be justified in generic cases. Despite that, Theorem \ref{thm:typicality} remain valid for a completely arbitrary constraint $\mathcal{R}$ and does not depend on the specific form of $\Omega_S$. In this sense, it goes beyond the ``local thermalization'' paradigm and it is able to address more general situations, as we will see in Chapter \ref{Ch7} and \ref{Ch8}.

\subsection{Summary on Typicality}

The Typicality Approach formalises the picture of ``local thermalization''. However, for a correct interpretation, it is important to disentangle two aspects. One is the physical reason for the emergence of a local thermal state in weakly interacting quantum systems, which is the entanglement between the subsystem and the rest. This feature is also present in the Eigenstate Thermalization Hypothesis. The other one is the argument that this behaviour is expected in the overwhelming majority of cases. This is the key statement of typicality, due to the concentration of measure phenomenon. It is important to stress here that the statement does not have a dynamical meaning. Rather, it is the result of a kinematical analysis on the Hilbert space and as such we should not assume that a specific dynamics will necessary lead to the emergence of typical properties. Indeed, we know that this does not always happens. Examples are the many-body localised systems. Thus, the correct interpretation of this analysis is that it explains why thermalization is such an ubiquitous phenomenon, in macroscopic systems. The reason is that, at the local level, most of the global states look very similar to a thermal state. Because of that, only very peculiar Hamiltonians will keep the state out of the ``typical situation''. A simple example of such case is given by non-interacting Hamiltonians. Studies aiming at identifying physical criteria to assess whether the dynamics of a given Hamiltonian, for a given observable, will lead to a typical or to an untypical expectation value are currently ongoing\cite{Garcia-Pintos2017,Hamazaki2018}.

\section{Summary}\label{sec:Ch2Discussion}

In this Chapter we introduced the basic aspects of two approaches for the emergence of thermal equilibrium in isolated quantum systems: the Eigenstate Thermalization Hypothesis and the Typicality Approach.  According to the first one, thermalization occurs at the level of individual energy eigenstates. Thus, the emergence of a microcanonical expectation value at equilibrium is due to the energy-eigenstate-expectation values of an observable $A$ being of a thermal form. The main problem with this approach is that there is no analytical explanation for the ETH. It is only a condition that can be checked a posteriori, once we have knowledge of the exact form of the energy eigenstates. Because of that, it lacks of predictive power, as it does not specify ($1$) which systems are expected to satisfy ETH, ($2$) for which observables, and ($3$) how long it will take to reach thermal equilibrium. The relevance of our work for the foundations of ETH will be discussed in Chapters \ref{Ch3}, \ref{Ch4} and \ref{Ch5}.

According to typicality, thermalization of a small subsystem is due to its entanglement with the rest. Performing a statistical analysis on a large-dimensional Hilbert space one can prove that, in a weakly-interacting quantum system, the overwhelming majority of pure states have a reduced state which is thermal. This is the ``local thermalization'' picture which emerges from the Typicality Approach and it can be formalised by means of Levy's Lemma. Theorem \ref{thm:typicality} allows to go beyond the local thermalization paradigm, which is expected to hold only in the weakly interacting case. It argues that a generic pure state will have a reduced state which in most case is indistinguishable from the canonical state of the system Eq.(\ref{eq:Canon}). The number of global states which escape this result is exponentially small in the dimension of the Hilbert space. This result will be used heavily in the second part of the thesis, especially in Chapters \ref{Ch7} and \ref{Ch8}.

\addcontentsline{toc}{part}{On the foundations}
\part*{On the foundations}

\chapter{Information-theoretic equilibrium and Observable Thermalization}\label{Ch3}

A crucial point in statistical mechanics is the definition of the notion of thermal equilibrium, which can be given as the state that maximises the von Neumann entropy, under the validity of some constraints. In this chapter we argue that such a notion of thermal equilibrium is, \emph{de facto}, not experimentally verifiable. To overcome this issue, we propose a weaker notion of thermal equilibrium, specific for a given observable and therefore experimentally testable. We will bring to light the thermal properties that it heralds and understand its relation with Gibbs ensembles. Moreover, we will see how this is relevant to understand the emergence of thermal behaviour in a closed quantum system. In particular, we will show that there is always a class of observables which exhibits thermal equilibrium properties and we give a recipe to explicitly construct them. Eventually,  we will show that such observables have an intimate connection with the ETH. The chapter is based on the author's work, published in Ref. \cite{Anza2017}.

\section{Introduction}

The ordinary way in which thermal equilibrium properties are obtained, in statistical mechanics, is through a complete characterisation of the thermal form of the state of the system. One way of deriving such form is by using \textit{Jaynes principle}\cite{Greiner1995,Uffink2006,Jaynes1957,Jaynes1957a}, which is the constrained maximisation of von Neumann entropy $S_{\mathrm{vN}}$. The main point of Jaynes' work was that statistical mechanics can be seen under the light of probabilistic inference, in which one is forced to give predictions about some macroscopic properties of the system, even though the information that we have on the system is not complete. Adopting such point of view, Jaynes showed that the unique state that maximises $S_{\mathrm{vN}}$ (compatibly with the prior information that we have on the system) is our best guess about the state of the system at the equilibrium. The outcomes of such procedure are the Gibbs ensembles mentioned in Section~\ref{Sec:EqTh}.\\

The starting point of Jaynes' derivation of statistical mechanics is that $S_{\mathrm{vN}}$ is a way of estimating the uncertainty that we have about which pure state the system inhabits. Unfortunately, even though $S_{\mathrm{vN}}$ is undoubtedly an important quantity, we know from quantum information theory that it does not address all kind of ignorance we have about the system. It is not the entropy of an observable (though the state is observable).\\

This issue is intimately related with the way we acquire information about a system, i.e. via quantum measurements. The process of measuring an observable $A$ on a quantum system allows to probe only the diagonal part of the density matrix $\Bra{a_i} \rho \Ket{a_i}$, when this is written in the observable eigenbasis $\left\{\Ket{a_i}\right\}$. For such a reason, from the experimental point of view, it is not possible to assess whether a macroscopic many-body quantum system is at thermal equilibrium (e.g. canonical ensemble $\rho_\beta$): the number of observables needed to probe all the density matrix elements is too big. In any experimentally reasonable situation we have access only to a few (sometimes just one or two) observables. It is therefore natural to imagine situations in which the outcomes of measurements are compatible with the assumption of thermal equilibrium, while the rest of the density matrix of the system is not. \\

The idea can be illustrated with a simple example. Suppose that we are interested in assessing whether a single particle in a box is at thermal equilibrium or not. We assume here that, for experimental reasons, we have access only to the position observable. The only information that we can get about the system is the probability distribution of the position: $p(\vec{x}) \equiv \bra{\vec{x}} \rho  \ket{\vec{x}}$. At thermal equilibrium (i.e. when the state is the canonical ensemble $\rho_\beta$), such distribution is constant. However, the same is true if the particle is in any given eigenstate of the Hamiltonian $\Ket{E}$:
\begin{align}
& \Bra{\vec{x}} \rho_\beta \Ket{\vec{x}} = \frac{1}{V} = \left\vert\frac{e^{-i \vec{k} \cdot \vec{x}}}{\sqrt{V}}\right\vert^2 = |\left\langle \vec{x} | E \right\rangle |^2
\end{align}
where $V$ is the volume of the box. Without other observables it is not possible to probe all the density matrix elements. If our results are, up to experimental uncertainties which are always present, compatible with a constant distribution, we conclude that our observations are compatible with the system being in a thermal state. Adopting a more general perspective, for a many-body quantum system, we will surely encounter a similar issue. This might very well cause the impression that the system is in a thermal state.\\

Despite that, we think that the fact that a distribution is compatible with its thermal counterpart will lead to the emergence of certain thermal properties, concerning the specific observable under scrutiny. Building our intuition on that, we propose a new notion of thermal equilibrium specific for a given observable, experimentally verifiable and which relies on a figure of merit that is not the von Neumann entropy. A good choice for such a figure of merit comes from quantum information theory and it is the Shannon entropy $H_{A}$ of the eigenvalues probability distribution $\left\{p(a_j) \right\}$ of an observable $A$. The operational interpretation of $H_{A}$\cite{Nielsen2010} matches our needs since it addresses the issue of the knowledge of an observable and it provides a measure for the entropy of its probability distribution.\\

Throughout the chapter we will work under the assumption that the Hilbert space of the system has finite dimension and we will refer to the cases in which the Hamiltonian of the system has no local conserved quantities, even though it is possible to address situations where there are several conserved quantities, like integrable quantum systems. Moreover, we will assume that the observable has a pure-point spectrum with the following spectral decomposition $A = \sum_j a_j A_j$, where $A_j$ is the projector onto the eigenspace defined by the eigenvalue $a_j$. $H_A$ is the entropy of its eigenvalues probability distribution $p(a_j) \equiv \mathrm{Tr} \left( \rho A_j\right)$
\begin{align}
&H_{A}[ \hat{\rho} ]  \equiv  -  \sum_{j \, : \, a_j \in \sigma_{A}} p(a_j) \log p(a_j)
\end{align}
where $\sigma_{A}$ is the spectrum of $A$.  We propose to define the notion of thermal equilibrium, for an arbitrary but fixed observable $A$, via a characterisation of the probability distribution of its eigenvalues. We will say that:\\

 $A$ {\it is at  thermal equilibrium when its eigenvalues probability distribution $p(a_j)$ maximises the Shannon entropy $H_{A}$, under arbitrary perturbations with conserved energy}. We call an observable with such a probability distribution: {\it thermal observable}.\\

It is important to note that such a principle characterises only the probability distribution at equilibrium and it does not uniquely identify an equilibrium state. Given an equilibrium distribution $p(a_j) = p^{\mathrm{eq}}(a_j)$, there will be several quantum states $\rho$ which give the same probability distribution for the eigenvalues $a_j$. In this sense this is weaker a notion of equilibrium, with respect to the ordinary one in which the whole state of the system is required to be an equilibrium state.\\

In the remaining part of the chapter we study the main consequences of the proposed notion of thermal equilbrium:  its physical meaning and the relation with Gibbs ensembles. The results provide evidence that the proposed notion of equilibrium is able to address the emergence of thermalisation. This is our first result. \\

Furthermore, we study the proposed notion of equilibrium in a closed quantum system and prove that there is a large class of bases of the Hilbert space which always exhibit thermal behaviour and we give an algorithm to explicitly construct them. We dub them \textit{Hamiltonian Unbiased Bases} (HUBs) and, accordingly, we call an observable which is diagonal in one of these bases \textit{Hamiltonian Unbiased Observable} (HUO). The existence and precise characterisation of observables which always thermalise in a closed quantum system is our second result. Furthermore, we investigate the relation between the notion of thermal observable and the ETH\cite{Berry1977,Shnirelman1974,Rigol2008,Deutsch1991a,Srednicki,Reimann2015a,Rigol2012,Srednicki1996,Srednicki1999,Polkovnikov2011b,Rigol2009}. We find an intimate connection between the concept of HUOs and ETH: the reason why these observables thermalise is precisely because they satisfy the ETH. Hence, with the existence and characterisation of the HUOs we are providing a genuine new prediction about which observables satisfies ETH, for any given Hamiltonian. The existence of this relation between HUOs and ETH is a highly non trivial feature and the fact that we can use it to predict which observables will satisfy ETH is our third result. Indeed, as far as we know, there is no way to predict which observables should satisfy ETH in a generic Hamiltonian system. It is an hypothesis which has to be checked by inspection, case by case.\\

\section{Information-theoretic equilibrium}

The request that the equilibrium distribution must be a maximum for $H_{A}$ can be phrased as a constrained optimisation problem, which can be solved using the Lagrange multiplier technique. The details of the treatment can be found in Appendix~\ref{App:EE}. Two sets of equilibrium equations are obtained and we show in which sense they account for the emergence of thermodynamic behaviour in the observable $A$. We assume that the only knowledge that we have on the system is the normalisation of the state and the mean value of the energy $\MV{\Ham}=E_0$, where $\Ham$ is the Hamiltonian of the system. We also make the following assumptions on the Hamiltonian: that it has spectral decomposition $\Ham = \sum_{\alpha} E_\alpha \mathbb{E}_{\alpha \alpha}$, where $\mathbb{E}_{\alpha \alpha} \equiv \Ket{E_{\alpha}} \Bra{E_{\alpha}}$ and that its eigenvectors $\left\{ \Ket{E_{\alpha}}\right\}$ provide a full basis of the Hilbert space. We call $\Ket{\psi_n}$ the eigenstates of the density operator, $\rho_n$ are the respective projectors and $q_n$ its eigenvalues. To  describe the state of the system we use the following convenient basis: $\left\{ \Ket{j,s}\right\}$ in which the first index $j$ runs over different eigenvalues $a_j$ of $A$ and the second index $s$ accounts for the fact that there might be degeneracies and it labels the states within each subspace defined by $a_j$. We also make use of the projectors $\Pi_{js} \equiv \Ket{j,s} \Bra{j,s}$. Furthermore, we are going to need the quantity $\mathcal{E}_n(j,s)\equiv \Bra{\psi_n} \Pi_{j,s} \Ham \Ket{\psi_n}$. Using the overlaps $D^{(n)}_{j,s} \equiv \left\langle j,s|\psi_n \right\rangle$ the expression of the constraints:
\begin{subequations}
\begin{align}
&\mathcal{C}_N \equiv \mathrm{Tr}(\rho) - 1 =  \sum_{n;j,s} q_n \left\vert D_{js}^{(n)}\right\vert^2 -1 \\
&\mathcal{C}_E \equiv \mathrm{Tr}\left(\rho \Ham \right) - E_0 = \sum_{n;j,s} q_n \mathcal{E}_n(j,s) - E_0 \,\,\,.
\end{align}
\end{subequations}
It is important to remark here that the knowledge of $E_0$ is always subject to uncertainty, which we call $\delta E$. In any physically reasonable situation there will be two conceptually different sources of uncertainty: a purely experimental one and an exquisitely quantum one. In this sense, all the states $\rho : \mathrm{Tr} \rho \Ham  \in I_0 \equiv \left[ E_0 - \frac{\delta E}{2}, E_0 + \frac{\delta E}{2}\right]$ will be solutions of the constraint equation, up to uncertainty $\delta E$. Even though we do not make any assumption on such a quantity, we note that $\delta E$ is usually assumed to be small on a macroscopic scale but still big enough to host a large number of eigenvalues of the Hamiltonian. \\

Exploiting Lagrange's multipliers technique we obtain four sets of equations. Derivatives with respect to the multipliers 
enforce the validity of the constraints while the derivatives with respect to the field variables give two independent set of equations. 
Using two particular linear combinations of them we obtain the following equilibrium equations (EEs):
\begin{subequations}
\begin{align}
&\overline{\mathcal{E}}_n(j,s) = \mathcal{E}_n(j,s) \label{eq:EE1} \\
&\hspace{-0.5cm} - |D_{js}^{(n)}|^2 \log \left(\sum_{m,s'} q_m  |D_{js'}^{(m)}|^2\right) =  \left( 1-\lambda_N \right)|D_{js}^{(n)}|^2   -  \lambda_E \mathcal{E}_n(j,s) \,\,\, , \label{eq:EE2}
\end{align}
\end{subequations}
where $\lambda_E$ and $\lambda_N$ are the Lagrange multipliers associated to $\mathcal{C}_E$ and $\mathcal{C}_N$, respectively.\\

Our first result is the understanding of the physical consequences of these equations. Together, they account for the emergence of thermal behaviour in $A$. Eq.\eqref{eq:EE1} gives the stability
under the flow generated by the Hamiltonian and it implies that the equilibrium distribution $p_{\mathrm{eq}}(a_j)$ has to be be invariant under the unitary dynamics. Indeed, writing the time evolution equation for $p_t(a_j)$, using the von Neumann equation, we obtain
\begin{align}
&i\hbar \frac{\partial }{\partial t} p_t(a_j) =\sum_{n,s} q_n \left(\overline{\mathcal{E}}_n(j,s)- \mathcal{E}_n(j,s)\right) \stackrel{eq}{=} 0 \,\,\, ,
\end{align}
where the superscript ``$eq$'' stands for ``at the equilibrium'', i.e. after plugging in the EEs. Eq.\eqref{eq:EE2} fixes the functional form of the distribution with respect to the Hamiltonian and to the Lagrange multipliers. It can be shown that it is responsible for the emergence of a thermodynamical relation between $H_A$ and the mean value of the energy. Integrating \eqref{eq:EE2} over the whole spectrum of $A$ we obtain
\begin{align}
&H_{A}^{eq} = (1-\lambda_N) - \lambda_E \, E_0 .
\end{align}
There is a linear contribution in the mean value of the energy, plus a ``zero-point'' term $H^{(0)}_{A}=(1-\lambda_N)$. This relation
brings to light the thermodynamical meaning of Shannon entropy $H_{A}^{eq}$ since such linear dependence on the average 
energy is a distinguishing feature of thermodynamic equilibrium. We note a strong analogy with the properties of von Neumann entropy, 
which acquires thermodynamical relevance once the state of the system is the Gibbs state $\rho_G = \frac{e^{-\beta \Ham}}{\mathcal{Z}}$, where 
$\mathcal{Z} = \mathrm{Tr} e^{-\beta \Ham}$ is the partition function
\begin{align}
& S_{\mathrm{vN}}(\rho_G) = \log \mathcal{Z} + \beta E_0 \qquad \longleftrightarrow \qquad H_{A}^{eq} = (1-\lambda_N) - \lambda_E \, E_0\,\,. \label{eq:von}
\end{align}

\section{Relation with statistical mechanics}

We now come to the issue of the relation with the ordinary notion of thermal equilibrium and, therefore, with the Gibbs ensemble. With respect to ours, this is a much more stringent condition because it is a characterisation of the full state of the system. For such a reason, in order to understand if our notion of thermal observable is compatible with the standard characterisation of thermal equilibrium, we need to find the condition under which our criterion gives a complete characterisation of the state of the system. Considering that we are using a maximum-entropy principle, a plausible auxiliary condition is the maximisation of the smallest among all the Shannon entropies. Such a request fully characterises the state of the system because the lowest Shannon entropy of the state is unique and it is the one in which the density matrix of the state is diagonal. Indeed, using the Schur-concavity of the Shannon entropy and the Schur-Horn theorem on eigenvalues of an Hermitian matrix \cite{Nielsen2010} it is easy to prove that
\begin{align}
& \min_{A \in \mathcal{A}} H_{A} (\rho) = S_{\mathrm{vN}} (\rho) \equiv - \mathrm{Tr} \left( \rho \log \rho \right) \,\,\, ,
\end{align}
where $\mathcal{A}$ is the algebra of the observables\cite{Strocchi2008a}. Our minimalist request to maximise the lowest Shannon entropy is translated in the condition of maximising von Neumann entropy. The state which maximises von Neumann's entropy, with the constraint of fixed total energy, is precisely the Gibbs state. It is therefore clear that our proposal constitutes an observable-wise generalisation of the ordinary notion of thermal equilibrium.\\

In the next section we will unravel the existence of a profound relation between the notion of thermal observable and the Eigenstate Thermalization Hypothesis.

\section{Isolated Quantum Systems - Relation to ETH}

ETH, in its original formulation by Srednicki \cite{Rigol2008,Deutsch1991a,Srednicki,Reimann2015a,Rigol2012,Srednicki1996,Srednicki1999,Polkovnikov2011b}, is an ansatz on the matrix elements of an observable when it is written in the Hamiltonian eigenbasis $\left\{ \Ket{E_{\alpha}}\right\}$:
\begin{align}
&A_{\alpha \beta}^{\mathrm{ETH}} \approx f_O^{(1)}(\overline{E}) \delta_{\alpha \beta} + e^{-\frac{S(\overline{E})}{2}}f^{(2)}_{A}(\overline{E},\omega) R_{\alpha \beta} \,\,\, , \label{eq:ETH} 
\end{align}
where $\overline{E} \equiv \frac{E_{\alpha} + E_{\beta}}{2}$, $\omega \equiv E_{\alpha} - E_{\beta}$ while $f_O^{(1)}$ and $f_O^{(2)}$ are smooth functions of their arguments. $S(\overline{E})$ is the thermodynamic entropy at energy $\overline{E}$ defined as $e^{S(\overline{E})} \equiv E \sum_{\alpha} \delta_{\epsilon}(\overline{E}- E_{\alpha})$, where $\delta_{\epsilon}$ is a smeared version of the Dirac delta distribution. $R_{\alpha \beta}$ is a complex random variable with zero mean and unit variance. Furthermore, it is important to remember that ETH by itself does not guarantee thermalisation, we need to impose that the initial state has a small dispersion in the energy eigenbasis\cite{Reimann2015a}. When this is true one says that $A^{\mathrm{ETH}}$ thermalises in the sense that its dynamically evolving expectation value is close to the microcanonical expectation value defined by the conditions on the average energy
\begin{align}
&\Bra{\psi(t)} A^{\mathrm{ETH}}\Ket{\psi(t)} \simeq \mathrm{Tr} \left( A^{\mathrm{ETH}} \rho_{mc} \right)\label{eq:AverTherm}
\end{align}
 where $\rho_{mc} = \rho_{mc}(E_0, \delta E)$ is the microcanonical state defined by the condition on the average value of the energy $\mathrm{Tr} \rho \Ham  \in I_0$.

The reason why Eq.(\ref{eq:AverTherm}) becomes true if ETH holds is purely technical, in the following sense. The small dispersion assumption on the initial state is telling us that when we expand the initial state in the energy eigenbasis, we will only have non-zero contributions from a certain energy window:
\begin{align}
&\Ket{\psi_0} = \sum_{\alpha} c_\alpha \Ket{E_\alpha} \simeq \sum_{\alpha : E_{\alpha} \in I_0} c_{\alpha} \Ket{E_{\alpha}} \label{eq:pure} \,\, .
\end{align}
If $\delta E$ is small enough, the smooth function $f_{O}^{(1)}$ will not vary too much in such an energy window. Once ETH has been checked by inspection, one can conclude that there is thermalisation of the average value in the sense of Eq.(\ref{eq:AverTherm}). It is important to note that, from the conceptual point of view, the ``small energy-dispersion assumption'' is a key element of the emergence of thermal equilibrium but it has nothing to do with ETH which, by itself, is only the ansatz in Eq.(\ref{eq:ETH}). Nevertheless, this assumption is expected to hold in real experiments because, when working with a many-body quantum system, it is almost impossible to prepare coherent superpositions of states with macroscopically different energies \cite{Reimann,Reimann2010}. \\

The main intuition behind the emergence of ETH for an observable $A$ (``Berry conjecture'' \cite{Berry1977,Shnirelman1974}) is that for a many-body quantum system the Hamiltonian eigenstates should be so complicated that, when they are written in the eigenbasis of $A$, their coefficients can be effectively described in term of randomly chosen coefficients and this would lead to the emergence of thermalisation. So far, nothing has been said about which observables we should look at, however it is generally expected that such a statement should be true for local observables. We explicitly show that such mechanism is revealed by the maximisation of Shannon entropy and this allows us to give a prediction about which observables should satisfy ETH. In the next section we will also clarify why we expect local observables to satisfy ETH and we will also provide some general conditions on the Hamiltonian whose validity guarantees that local observables will satisfy ETH.\\

Before we continue we need to define the following short-hand notation: $\sum_{\alpha}' \equiv \sum_{\alpha \in I_0}$, where $\alpha \in I_0$ stands for $\alpha : E_\alpha \in I_0$. 

\subsection{Hamiltonian Unbiased Observables and ETH}\label{Sub:HUOandETH}

To study the relation with ETH we need to change our perspective. The point of view that we are adopting is the following. One of the key-points behind the ETH is that, in many real cases, the expectation values computed onto the Hamiltonian eigenvectors can be very close to the thermal expectation values. Moreover, when one wants to argue that thermalisation in a closed quantum system arises because of ETH, a main assumption is that the initial pure state of the system $\Ket{\psi_0}$ has a very small energy uncertainty $\Delta E$, with respect to the average energy $E_0$: $\frac{\Delta E}{E_0} \ll 1$. Following these two insights, we take the extreme limit in which $\Delta E = 0$. With such a choice, we are left with an Hamiltonian eigenstate and a constraint equation given by $\mathrm{Tr} \rho \Ham = E_\alpha$. We note that Eq.(\ref{eq:EE1}) is trivially satisfied for an Hamiltonian eigenstate. Hence we assume that $\Ket{\psi} = \Ket{E_\alpha} : E_{\alpha} \in I_0$ and use the solvability of Eq.(\ref{eq:EE2}) as a criterion to look for observables which can be thermal. While this is a very specific choice, we will show that it is enough to unravel interesting features regarding the ETH. With this assumption, the second equilibrium equation becomes
\begin{align}
&\hspace{-0.5cm} - |D_{j}|^2 \log  |D_{j}|^2  =  \left( 1-\lambda_N \right)|D_{j}|^2    -  \lambda_E |D_{j}|^2 E_{\alpha}
\end{align}
Dividing by $|D_{j}|^2$ on each side we obtain a right-hand side which does not depend on the label $j$.
Hence, keeping in mind that $x\log x \to 0$ for $x \to 0$, the most general solution of this equation is a 
constant distribution with support on some subset ($\mathcal{I}_{\alpha} \subset \sigma_{A}$ 
depending on $E_{\alpha}$) of the spectrum and zero on its complementary:
\begin{equation}
p(a_j)  = \left\{ \begin{array}{ll}
  \frac{1}{d_{\alpha}} & \forall \,\, a_j \in \mathcal{I}_{\alpha} \subset \sigma_{A} \\
 & \\
 0 & \forall \,\, a_j \in \sigma_{A} / \mathcal{I}_{\alpha}  \label{eq:const} 
  \end{array} \right. 
\end{equation}
where $d_{\alpha} = \mathrm{dim} \, \mathcal{I}_{\alpha}$ is the number of orthogonal states on which the distribution has non-zero value. The distribution of eigenvalues of $A$ given by Eq.(\ref{eq:const}) fully agrees with the prediction from the microcanonical ensemble $\rho_{mc}^{(\alpha)}$, defined by the condition in Eq.(\ref{eq:const}). This is true, in particular, for the expectation value
\begin{align}
& \MV{A} = \frac{1}{d_{\alpha}} \sum_{j \, : a_j \in \mathcal{I}_{\alpha}} a_j = \mathrm{Tr} \left(A \, \rho_{mc}^{(\alpha)}\right)
\end{align}
It has to be understood here that Eq.(\ref{eq:const}) is a highly non-trivial condition on $A$, which is not going to be fulfilled by every observable since it imposes a very specific relation between its eigenstates and the energy eigenvectors.\\

Within the context of quantum information theory, and especially in quantum cryptography, a similar relation has already been studied. Two bases of a $\mathcal{D}-$dimensional Hilbert space ($\mathcal{B}_v \equiv \left\{ \Ket{v_i} \right\}$ and $\mathcal{B}_w \equiv \left\{ \Ket{w_j} \right\}$) are called {\it mutually unbiased bases} (MUB) \cite{Wehner2010,Bengtsson2007,Bandyopadhyay2002,Durt2010,Lawrence2002} when 
\begin{align}
&|\left\langle v_i | w_j \right\rangle|^2 = \frac{1}{\mathcal{D}} \qquad \forall \,\, i,j=1,\ldots, \mathcal{D}
\end{align}
Such a concept is a generalisation, expressed in term of vector bases, of canonically conjugated operators. In other words, each vector of $\mathcal{B}_v$ is completely delocalised in the basis $\mathcal{B}_w$ and viceversa. Here we mention the result about MUBs which matters most for our purposes: given the $2^N$-dimensional Hilbert space of $N$ qubits, there are $2^N+1$ MUBs and we have an algorithm to explicitly find all of them\cite{Lawrence2002}. Therefore, if we are in an energy eigenstate, an observable unbiased with respect to the Hamiltonian basis will always exhibit a microcanonical distribution. This is also true if our state is not exactly an energy eigenstate, it is enough to have a state that has a sufficiently narrowed energy distribution. We now provide a simple argument to prove such a statement. We also note that such condition is closely related to the small dispersion condition briefly discussed before and that it is necessary to guarantee thermalisation, according to the ETH.\\ 

By assumption, the pure state in Eq.(\ref{eq:pure}) has an energy distribution $\left\{p_{\alpha} \equiv |c_{\alpha}|^2 \right\}_{\alpha = 1}^{\mathcal{D}}$ with small dispersion. This implies that its Shannon entropy $H_{\Ham}$ has a small value, because the profile of the distribution will be peaked around a certain value. For such a reason $H_{\Ham}(\Ket{\psi_0})$ will be much smaller than its maximum value
\begin{align}
&H_{\Ham}(\Ket{\psi_0}) \ll \log \mathcal{D} \label{eq:SmallH}
\end{align}
Moreover, it can be proven\cite{Wehner2010} that between any pair of MUB, like the Hamiltonian eigenbasis and an HUB, there exists the following entropic uncertainty relation, involving their Shannon entropies:
\begin{align}
&H_{\Ham} (\Ket{\phi}) + H_{HUB}(\Ket{\phi}) \geq \log \mathcal{D} \qquad \forall \Ket{\phi} \in \mathcal{H} \label{eq:EntUnc}
\end{align}
Putting together Eq.(\ref{eq:SmallH}) and Eq.(\ref{eq:EntUnc}) we obtain that, for all the states with a small energy dispersion, the Shannon entropy of an HUB will always have a value close to the maximum:
\begin{align}
&H_{HUB}(\Ket{\psi_0})  \simeq \log \mathcal{D}
\end{align}
This in turn implies that the distribution of all the HUO will be approximately the same as the one computed on the microcanonical state. For such a reason, we now study the properties of a HUO:
\begin{align}
&A = \sum_{j,s} a_j \Ket{j,s} \Bra{j,s} &&\qquad \braket{j,s}{E_{\alpha}} =  \frac{e^{i\theta_{js,\alpha}}}{\sqrt{\mathcal{D}}} \label{eq:Const}
\end{align}

\subsubsection*{HUOs and ETH: diagonal matrix elements}

In order to investigate the relation with ETH we need to study the matrix elements of a HUO in the energy basis:
\begin{align}
&A_{\alpha \beta}= \frac{1}{\mathcal{D}} \sum_{j,s} a_j e^{i \omega_{js}^{\alpha \beta}}  \,\,\, ,  \,\, \mathrm{with} \,\,\, \omega_{js}^{\alpha \beta} =  (\theta_{js,\beta}- \theta_{js,\alpha})\,\,. \label{eq:off}
\end{align}
It is straightforward to conclude that its diagonal matrix elements are constant in such a basis and therefore the so-called (Hamiltonian) Eigenstate Expectation values reproduce the 
microcanonical expectation values: 
\begin{align}
&A_{\alpha \alpha} = \frac{1}{\mathcal{D}} \sum_{j,s} a_j = \frac{1}{\mathcal{D}}\mathrm{Tr} \, A = \mathrm{Tr} \left(A \, \rho_{mc}\right) = \MV{A}_{mc} \label{eq:EE}
\end{align}
This is the first part of our third result and as it is, it can already be used to explain the emergence of thermalisation in closed quantum system, for some observables. In \cite{Reimann,Reimann2010} Reimann proved an important theorem about equilibration of closed quantum systems. He was able to show that under certain experimentally realistic conditions, the mean value of an observable is not much different from its value computed on the time-averaged density matrix, or Diagonal Ensemble (DE):
\begin{align}
&\rho_{DE} \equiv \sum_{\alpha} \left|c_{\alpha}\right|^2 \Ket{E_{\alpha}} \Bra{E_{\alpha}} \label{eq:DE}
\end{align}
where $c_{\alpha} \equiv \left\langle \psi_0 | E_{\alpha}\right\rangle$ and $\Ket{\psi_0}$ is the initial pure state of the isolated system. Roughly speaking, the two main assumptions made by Reimann are the following: first, that in any experimentally realistic condition the state of the system will occupy a huge number of energy eigenstates, even if the average energy is known up to a small experimental uncertainty; second, that the observable under study has a finite range of average values, due to the fact that we wish to measure it. For a clear and synthetic discussion on this topic we suggest \cite{Yukalov2011} and we send the reader to the original references \cite{Reimann,Reimann2010}. We note that the first assumption does not contradict the small energy dispersion assumption. Indeed, as argued by Reimann, in a many-body quantum system, even if the energy is known up to a macroscopically small scale $\delta E$, there will be a huge number of eigenstates within the range $E_{\alpha} \in I_0$. This is equivalent to having a high density of energy eigenstates. To conclude, given a HUB it is always possible to obtain a HUO which satisfies the finite-range assumption. We can therefore apply Reimann's theorem to HUOs.\\

It is important to note that Reimann's theorem explains equilibration around the DE but this does not necessarily entail thermalisation. The DE still retains information about the initial state while, on 
the contrary, thermalisation is defined (also) by the independence on the initial state. This is the point where our result is able to take a step forward and explain the emergence of thermal equilibrium in the HUOs. We can use Eq.(\ref{eq:EE}) to prove that all the HUOs exhibit complete independence from the initial conditions:
\begin{align}
& \mathrm{Tr} \left( A \rho_{DE}\right) = \sum_{\alpha} |\psi_{\alpha}^{0}|^2 A_{\alpha \alpha} = \mathrm{Tr} \left( A \rho_{mc} \right) =\MV{A}_{mc} 
\end{align}

\subsubsection*{HUOs and ETH: off-diagonal matrix elements}

In order to prove that a HUO satisfies ETH we need to study also its off-diagonal matrix elements. By using Eq.(\ref{eq:Const}), we can investigate how the phases $\omega_{js}^{\alpha \beta} \equiv \left( \theta_{js,\beta} - \theta_{js,\alpha} \right)$ are distributed. This can be done numerically exploiting the available algorithms to generate MUBs\cite{Lawrence2002,Durt2010}. The numerical investigation of the distribution of $\omega_{js}^{\alpha \beta}$ is reported in the supplementary material. Here we simply state the result: for a fixed value of the energy quantum numbers, the observed distributions of $\cos \omega_{js}^{\alpha \beta}$, $\sin \omega_{js}^{\alpha \beta}$ are well described by the assumption that $\omega_{js}^{\alpha \beta}$ are independent and randomly distributed in $\left[ -\pi,\pi \right]$, with a constant probability distribution.\\

There are different ways in which this result can be used. The general argument is the following and it reflects the spirit of ETH: in Eq.(\ref{eq:off}), the phases $\omega_{js}^{\alpha \beta}$ will have a randomising action on the eigenvalues $a_j$  and this will make the value of the off-diagonal matrix elements severely smaller than the value of the diagonal ones:
\begin{align}
&\mathrm{Re} A_{\alpha \beta}, \mathrm{Im} A_{\alpha \beta} \ll \frac{1}{\mathcal{D}}\mathrm{Tr} A =  A_{\alpha \alpha}
\end{align}
The randomness of the coefficients involved in the evaluation of the off-diagonal matrix elements has been recently proposed as the basic mechanism to explain the $\frac{1}{\sqrt{\mathcal{D}}}$ scaling behaviour which was observed to occur in some models \cite{Khatami2013,Konstantinidis2015,Beugeling2014,Beugeling2015,Khemani2014}. The maximisation of Shannon entropy is therefore giving us a recipe to find the observables for which this is true.\\

\subsection{HUOs and ETH: two important examples}

In order to understand how this works in practice we need to say something specific about the eigenvalues. We highlight two important cases in which our result is helpful: an observable which is highly degenerate and an observable whose eigenvalues distribution is not correlated with the phases $\omega_{js}^{\alpha \beta}$.\\

{\it Highly degenerate observable} - Assuming that $j=1,\ldots, n_A$ while $s=1,\ldots,d_j$, with $d_j \gg n_A$, the sum in Eq. (\ref{eq:off}) splits into $n_A$ of sums and each one of them is a sum of $d_j \gg 1$ identically distributed random variables. We can apply the central limit theorem to the real and imaginary part of Eq. (\ref{eq:off}) and obtain the following expression for the off-diagonal matrix elements 
\begin{align}
A^{\mathrm{HUO}}_{\alpha \beta} &\approx \sum_{j=1}^{n_A}  X_{\alpha \beta}^{(j)} & X_{\alpha \beta}^{(j)} &\sim \mathcal{N} \left[0,\left(\frac{\lambda_j \sqrt{d_j}}{D}\right)^2\right] \label{eq:HUO}\,,
\end{align}
where $X^{(j)}_{\alpha \beta} \sim \mathcal{N}[\mu,\sigma^2]$ means that $X^{(j)}_{\alpha \beta}$ is a complex random variable, normally distributed, with mean $\mu$ and variance $\sigma^2$.
Under the additional assumption that the $X^{(j)}_{\alpha \beta}$ are independent one finds that, because Eq.~\eqref{eq:HUO} is a finite sum of normally distributed random variables, we have 
\begin{equation}
A_{\alpha \beta}^{\mathrm{HUO}} \sim \mathcal{N} \left[ 0, \sigma^2_{n_A}\right] \qquad \mathrm{with} \qquad  \sigma_{n_A}^{2} \equiv \sum_{j=1}^{n_A} \left(\frac{\lambda_j \sqrt{d_j}}{D}\right)^2 = \frac{1}{D} \MV{\left(A^{\mathrm{HUO}}\right)^2}_{\mathrm{mc}}\,
\end{equation}
Eventually we get
\begin{align}
A^{\mathrm{HUO}}_{\alpha \beta} \approx \sqrt{\frac{1}{D} \MV{\left(A^{\mathrm{HUO}}\right)^2}_{\mathrm{mc}} }  \,\,\,\mathcal{R}_{\alpha \beta} \,. \label{eq:HUOandETH}
\end{align}
For a binary observable, i.e with two distinct eigenvalues ($\pm 1$), this means that for large $d_j$ 
\begin{align}
A^{\mathrm{HUO}}_{\alpha \beta} &\approx \frac{1}{\sqrt{D}} \mathcal{R}_{\alpha \beta} & \mathcal{R}_{\alpha \beta} &\sim \mathcal{N}[0,1] \,,
\end{align}
which means that $A^{\mathrm{HUO}}$ satisfies Hypothesis~\ref{hyp:originaleth}. The diagonal matrix elements reproduces the microcanonical expectation values and the off-diagonal matrix elements are well described by a random variable with zero mean and  $\sim \frac{1}{\sqrt{\mathcal{D}}}$ variance. This is precisely what we expect from a local observable, since its eigenvalues have a huge number of degeneracies, which grows exponentially with the size of the system, and it is in full agreement with the randomness conjecture made in \cite{Beugeling2015,Khemani2014}.\\

{\it Uncorrelated distribution} - If the eigenvalues distributions $\left\{ a_j \right\}_{j,s}$ and the phases $\left\{e^{i \omega_{js}^{\alpha \beta}}\right\}_{j,s}$ are not correlated Eq.(\ref{eq:off}) becomes 

\begin{align}
&A_{\alpha \beta}^{unc} \approx   \left(\sum_{j,s} \lambda_{j} \right) \left(\frac{1}{\mathcal{D}}\sum_{j,s} e^{i \omega_{js}^{\alpha \beta}}  \right) = \frac{\mathrm{Tr} A}{\mathcal{D}} \,  \delta_{\alpha \beta}
\end{align}

Where we used the fact that the two sequences are uncorrelated to approximate the sum as the product of two sums and in the second identity we used the fact that $\left\{\Ket{j,s}\right\}$ is a complete basis. We conclude that the off-diagonal matrix elements of such an observable are much smaller than the diagonal ones and therefore we can neglect them when we compute its dynamical expectation value:
\begin{align}
&\Bra{\psi(t)} A^{unc} \Ket{\psi(t)} = \sum_{\alpha,\beta} \psi_{\alpha}^{0} \bar{\psi}_{\beta}^{0} e^{- \frac{i}{\hbar} (E_{\alpha} - E_{\beta})t} A_{\alpha \beta}^{unc}  \approx \sum_{\alpha} |\psi_{\alpha}^{0}|^2 A_{\alpha \alpha}^{unc} = \MV{A^{unc}}_{mc}
\end{align}
 The assumption that the sequences $\left\{ a_j \right\}_{j,s}$ and $\left\{e^{i \omega_{js}^{\alpha \beta}}\right\}_{j,s}$ are not correlated will not be always fulfilled, nevertheless this is what we intuitively expect from the following argument. We know that $\omega_{js}^{\alpha \beta}$ can be described by a random variable with a constant probability distribution, this suggests that there is no correlation between $\left\{\omega_{js}^{\alpha \beta}\right\}_{j,s}$ and its labeling $\left\{ j,s\right\}_{j,s}$. On the contrary, the eigenvalues $a_j$ are highly correlated to the labels and therefore we do not expect them to be correlated to the $\left\{\omega_{js}^{\alpha \beta}\right\}_{j,s}$. \\

These results follows from the study of the equilibrium equations, under the assumption that the state is an Hamiltonian eigenstate $\Ket{\psi} = \Ket{E_{\alpha}}$ belonging to the energy shell $I_0$. It is straightforward to see that the same results hold when the state is the microcanonical state $\rho_{mc}(E_0,\delta E)$ involved in Eq.(\ref{eq:AverTherm}) and defined by the condition $\mathrm{Tr} \rho \Ham \in I_0$.

\section{Discussion}

We proposed a new notion of thermal equilibrium for an observable: we say that $A$ is a \emph{thermal observable} when its eigenvalues probability distribution maximises the Shannon entropy $H_{A}$, compatibly with the validity of some constrains. Setting up a constrained optimisation problem we derived two equilibrium equations and studied their physical implications. Eq.(\ref{eq:EE1}) enforces the stability of the distribution with respect to the dynamics generated by the Hamiltonian while Eq.(\ref{eq:EE2}) fixes the functional form of the distribution. Integrating the second equation we obtained a linear relation between Shannon entropy and the mean value of the energy, which is a strong indication that $H_{A}$ at equilibrium has some thermodynamic properties. The physical meaning of the two equilibrium equations provides evidence that the maximisation of Shannon entropy is able to address the emergence of thermal behaviour. We also studied the relation of the proposed notion of equilibrium with quantum statistical mechanics. We showed that it is possible to interpret the ordinary notion of thermal equilibrium through the maximisation of Shannon entropy. The request to maximise the lowest among all the possible Shannon entropies lead to Gibbs ensembles and therefore to the quantum statistical mechanics characterisation of thermal equilibrium.\\

Together, the physical meaning of the equilibrium equations and the proven relation between maximisation of Shannon entropy and Gibbs ensemble, provide strong evidence that the proposed notion of thermal observable is physically relevant for the purpose of investigating the concept of thermal equilibrium. This is our first result. \\

The second result is the relevance of the proposed notion of equilibrium to address the emergence of thermal behaviour in a closed quantum system and especially its relation with the ETH. Using maximisation of $H_{A}$ we were able to find a large class of observables which always thermalise and provide an algorithm to explicitly construct them. We call them Hamiltonian Unbiased Observables (HUOs). Using this result, we analysed the matrix elements of an HUO in the Hamiltonian eigenbasis to understand if they satisfy the ETH. The study of the diagonal matrix elements reveals that they always satisfy ETH, being constant in such basis. This has been used to prove thermalisation of the average value of HUOs, in connection with Reimann's theorem about equilibration of observables.\\

The study of the off-diagonal matrix elements of a HUO revealed that their value is exponentially suppressed in the dimension of the Hilbert space. Concretely, this has been proven for two classes of observables: highly degenerate and uncorrelated observables. This complete the proof that such HUOs satisfy the ETH. The relevance of this result for pure-states statistical mechanics is related to the two main objections usually raised against ETH: the lack of predictive power for what concern which observables satisfy ETH. The proposed notion of thermal equilibrium is revealing its predictive power since it gives us a way of finding observables which always satisfy ETH, in a closed quantum system.\\

We would like to conclude by putting this set of results about ETH in a more general perspective. ETH is one of the main paradigms to justify the applicability of statistical mechanics to closed many-body quantum systems. However, it is just a working hypothesis, it is not derived from a conceptually clear theoretical framework. For such a reason, one of the major open issues is its lack of predictability. Despite that, there has been a huge effort to investigate whether the ETH can be invoked to explain thermalisation in concrete Hamiltonian models \cite{Khatami2013,Konstantinidis2015,Beugeling2014,Beugeling2015,Khemani2014,Steinigeweg2013,Sorg2014,Zangara2013,Steinigeweg2013a,Kim2014,Muller2015,Mondaini2016,Magan2016,Khodja2015,Ikeda2013} and its use is nowadays ubiquitous. Therefore, we think it is very important to put the ETH under a conceptually clear framework. In this sense, the relevance of our work resides in the fact that we obtain the ETH ansatz as a prediction, by using a natural starting point, the maximisation of Shannon entropy. Furthermore, using the proposed notion of thermal equilibrium is already giving concrete benefits. First, we now have a way of computing observables which satisfy the ETH and this prediction can be tested both numerically and experimentally; second, by studying the conditions under which local observables satisfy our equilibrium equations we are able to give two conditions which are necessary to guarantee the validity of ETH for all local observables. We believe that our investigation suggests that maximisation of $H_{A}$ is able to grasp the main intuition behind ``thermalisation according to ETH'' and it can be the physical principle behind the appearance of ETH.\\ 

Further investigation in this direction is certainly needed, but we would like to conclude by suggesting a way in which this new tool can be used to address the long-standing issue of the thermalisation times. From our investigation one can infer that Shannon entropy is a good figure of merit to study the dynamical onset of thermalisation in a closed quantum system. Within this picture, the time-scale at which thermalisation should occur for $A$ is therefore given by the time-scale at which $H_{A}$ reaches its maximum value. A prediction about the time-scale at which  $H_{A}$ is maximised will translate straightforwardly in a prediction about the thermalisation time.

\chapter{Eigenstate Thermalization for Degenerate Observables}\label{Ch4}



In the previous chapter we argued for a new, observable-wise, notion of thermal equilibrium. We have investigated how such notion of thermal equilibrium might be relevant to address thermal equilibrium in closed quantum systems. As a result, we found a large class of bases and observables which exhibit thermal properties: the HUBs and the HUOs. We have shown that they admit an algorithmic construction and, investigating their matrix elements in the Hamiltonian eigenbasis, we have argued that highly degenerate HUOs satisfy ETH. Unfortunately this still leaves open the question of when concrete, physically relevant, observables satisfy the HUO condition. In this chapter we make progress in this direction: we present a theorem which can be used as a tool to investigate the emergence of the HUO condition. Building on the connection between HUOs and ETH, we aim at elucidating the emergence of ETH for observables that can realistically be measured due to their high degeneracy, such as local, extensive or macroscopic observables. We bisect this problem into two parts, a condition on the relative overlaps and one on the relative phases between the eigenbases of the observable and Hamiltonian. We show that the relative overlaps are unbiased for highly degenerate observables and demonstrate that unless relative phases conspire to cumulative effects, this makes such observables verify ETH. Through this we elucidate potential pathways towards proofs of thermalization. This chapter is based on the work published by the author, in collaboration with Dr. C. Gogolin and Dr. M. Huber, in Ref. \cite{Anza2018}.\\

\section{Physical observables}

Another issue left open by the definitions of the ETH given in Chapter \ref{Ch2} is the identification of physical observables for which ETH is supposed to hold.
In this work we show that highly degenerate observables are good candidates.
Those are natural in at least three scenarios:
First, local observables only have a small number of distinct eigenvalues, as they act non-trivially only on a low dimensional space, and each such level is exponentially degenerate in the size of the system on which they do not act.
Second, averages of local observables, like for example the total magnetization, are, for combinatorial reasons, highly degenerate around the center of their spectrum.
Third, \emph{macro observables} as introduced by von~Neumann \cite{Neumann1929,VonNeumann2010} and studied in \cite{Goldstein2006,Goldstein2010,Goldstein2010a} that are degenerate through the notion of macroscopicity.
Here the idea is that on macroscopically large systems one can only ever measure a rather small number of observables and these observables can take only a number of values that is much smaller than the enormous dimension of the Hilbert space and they commute either exactly or are very close to commuting observables. An example are the classical position and momentum of a macroscopic system.
While they are ultimately a coarse grained version of the sum of the microscopic positions and momenta of all the constituents, they can both be measured without disturbing the other in any noticeable way.
Such classical observables hence partition, in a natural way, the Hilbert space of a quantum system in a direct sum of subspaces, each corresponding to a vector of assignments of outcomes for all the macro observables.
Even by measuring all the available macro observables one can only identify which subspace a quantum system is in, but never learn its precise quantum state.
Dynamically, to get the impression that a system equilibrates or thermalizes, it is hence sufficient that the overlap of the true quantum state with each of the subspaces from the partition is roughly constant in time and the average agrees with the suitable thermodynamical prediction. 
One would thus expect ETH to hold for such observables.
As in any realistic situation, the number of observables times the maximum number of outcomes per observable (and hence the number of different subspaces) is vastly smaller than the dimension of the Hilbert space, one is again dealing with highly degenerate observables.

\subsection{Hamiltonian Unbiased Observables}
Before we proceed with the main result of the chapter, we quickly summarize the conceptual results of the previous one, originally derived by the author in \cite{Anza2017}. Suppose $A\coloneqq\sum_i a_i A_i$ is an observable with eigenvalues $a_i$ and respective projectors $A_i$. We say that $A$ is a \textit{thermal observable} with respect to the state $\rho$ if its measurement statistics $p(a_i) \coloneqq \Tr\left( \rho\, A_i\right)$ maximizes the Shannon entropy $S_A \coloneqq - \sum_i p(a_i) \log p(a_i)$ under two constraints: normalization of the state $\Tr(\rho) = 1$ and fixed average energy $\Tr\left(\rho \Ham \right)$.

In \cite{Anza2017} it was proven that this is a generalization of the standard notion of thermal equilibrium, in the following sense: What we usually mean by thermal equilibrium is that the state of the system $\rho$ is  close to the Gibbs state $\rho_G$, in the sense given by some distance defined on the convex set of density matrices. A well-known way to characterize $\rho_G$ is via the constrained maximization of von Neumann entropy $S_{\mathrm{vN}}\coloneqq-\Tr( \rho\, \log \rho)$.
Now, for any state $\rho$, the minimum Shannon entropy $S_A$ (among all the observables $A$) is the von Neumann entropy
\begin{align}
\min_{A} S_A = S_{\mathrm{vN}}\,.
\end{align}
Therefore, the Gibbs ensemble is the state that maximizes the lowest among all the Shannon entropies $S_A$.
Hence the maximization of the Shannon entropy $S_A$ is an observable dependent generalization of the ordinary notion of thermal equilibrium. 

One can use the Lagrange multiplier technique to solve the constrained optimization problem and two equilibrium equations emerge. They implicitly define the equilibrium distribution $p_{\mathrm{eq}}(a_i)$ as their solution. Using such equations to investigate the emergence of thermal observables in a closed quantum system, it can be proven that for any given Hamiltonian there is a huge amount of observables that satisfy ETH: the Hamiltonian Unbiased Observables (HUO).

The name originates from the following notion:
A set of normalized vectors $\{\Ket{u_j}\}_j$ is mutually unbiased with respect to another set of vectors $\{\ket{v_k}\}_k$ if the inner product between any pair satisfies |$\braket{u_i}{v_k}| = 1/\sqrt{D}$, where $D$ is the dimension of the Hilbert space. A basis is called Hamiltonian Unbiased Basis (HUB) if it is unbiased with respect to the Hamiltonian basis.
Accordingly, a HUO is an observable which is diagonal in a HUB.
The concept of mutually unbiased basis (MUBs) has been studied in depth in quantum information theory \cite{Wehner2010,Bengtsson2007,Bandyopadhyay2002,Lawrence2002,Durt2010,Bengtsson2007a}.
For our purposes, the most important result is the following:
Given a Hilbert space $\Hilb = \otimes_{j=1}^N \Hilb_j$ with $\dim(\Hilb_j) = p$ for some prime number $p$ and some fixed orthonormal basis in $\Hilb$ there is a total of $p^N+1$ orthonormal bases, including the fixed basis, that are all pairwise mutually unbiased \cite{Bengtsson2007,Bandyopadhyay2002}.
Moreover, there is an algorithm to explicitly construct all of them \cite{Bengtsson2007,Bandyopadhyay2002}.
Applying this result to the Hamiltonian basis we conclude that there are $p^N$ HUBs. By studying the matrix elements of a HUO in the Hamiltonian basis, in subsection \ref{Sub:HUOandETH} we argued that sufficiently degenerate HUOs should satisfy ETH, under some mild additional conditions. Before we proceed, we would like to expand on the mechanism behind the emergence of ETH for a highly-degenerate HUO.
Eq.~\eqref{eq:HUOandETH} will hold whenever we can apply the central limit theorem within each subspace at fixed eigenvalue.
As was argued in \cite{Anza2017}, for a fixed pair of indices $(m,n)$, the phases $\gamma_{js}^{mn}$ behave as if they were pseudo-random variables and their number is exponentially large in the system size.
The labels $(j,s)$ provide a partition of these $D$ phases into $n_A$ groups, each made of $d_j$ elements.
In the overwhelming majority of cases each group of $d_j$ phases will exhibit the same statistical behavior as the whole set.
In this case, Eq.~\eqref{eq:HUOandETH} will behave as a sum of independent random variables and it will give the exponential decay of the off-diagonal matrix elements.
It may happen that the index $j$, labeling different eigenvalues, samples the phases in a biased way and prevents some of the off-diagonal matrix elements from being exponentially small.
This, even though it seems unlikely, is possible and it would induce a coherent dynamics on the observable which can prevent its thermalization.
This can happen for example in integrable quantum system for observables which are close to being conserved quantities.

The point can also be seen from the perspective of random matrix theory. Given the Hamiltonian eigenbasis, if we perform several random unitary transformations and study the distribution of the outcome basis, it can be shown that, in the overwhelming majority of cases we will end up with a basis that is almost HUB\cite{Durt2010,Bengtsson2007a}, up to corrections which are exponentially small in the system size.
Hence for large system sizes, if we pick a basis at random, most likely it will be almost a HUB \cite{Durt2010,Bengtsson2007a}.

We now present the main result of the paper: a theorem that can be used to study under which conditions highly degenerate observables are HUO.
\begin{theorem} \label{thm:mainresult}
  Let $\{\ket{\psi_m}\}_{m=1}^M \subset \Hilb$ be a set of orthonormal vectors in a Hilbert space $\Hilb$ of dimension $D$. 
  Let $A = \sum_{j=1}^{n_A} a_j \Pi_j$ be an operator on $\Hilb$ with $n_A \leq D$ distinct eigenvalues $a_j$ and corresponding eigen-projectors $\Pi_j$.
  Decompose $\Hilb = \bigoplus_{j=1}^{n_A} \Hilb_j$ into a direct sum such that each $\Hilb_j$ is the image of the corresponding $\Pi_j$ with dimension $D_j$.
  For each $j$ for which $D_j(D_j-1) \geq  M+1$ there exists an orthonormal basis $\{\ket{j,k}\}_{k=1}^{D_j} \subset \Hilb_j$ such that for all $k,m$
  \begin{equation} \label{eq:THEO}
    | \braket{\psi_m}{j,k} |^2 = \bra{\psi_m} \Pi_j \ket{\psi_m} / D_j \,.  
  \end{equation}
\end{theorem}
A detailed proof is provided in Appendix~\ref{App:Proof}.
If the condition $D_j(D_j-1) \geq M+1$ is fulfilled for all $j$, then the set of all $\{\ket{j,k}\}_{j,k}$ obviously is an orthonormal basis of $\Hilb$ and $A$ is diagonal in that basis.
So, as long as the degeneracies $D_j$ of $A$ are all high enough with respect to $M$,  $A$ has an eigenbasis whose overlaps with the states $\ket{\psi_m}$ are given exactly by the right-hand side of Eq.\eqref{eq:THEO}. 

A particularly relevant case is when $A$ is a local observable acting non-trivially only on some small subsystem $S$ of dimension $D_S$ of a larger $N$-partite spin system of dimension $D = d^N$, i.e., $A \coloneqq \sum_{j=1}^{D_S} a_j \ketbra{a_j}{a_j} \otimes \1_{\overline{S}}$ and $\{\ket{\psi_m}\}_{m=1}^M$ is taken to be an eigenbasis $\{\ket{E_m}\}_{m=1}^D$ of the Hamiltonian $H$ of the full system.
We summarize some non-essential further details in Appendix \ref{AppB}.
In this case the degeneracies are all at least $D_j \geq D/D_S = d^{N-|S|}$, so that the above results guarantees that for all observables on up to $|S| < N/2$ sites there exists a tensor product basis $\{\ket{a_j, k}\}_{j,k}$ for $\Hilb$ which diagonalizes $A$ and with the property that
\begin{equation}
  | \braket{E_m}{a_j,k} |^2 = \frac{1}{d^{N-|S|}} \, \bra{a_j} \Tr_{\overline{S}} \ketbra{E_m}{E_m} \, \ket{a_j} \,.
\end{equation}
For subsystems with support on a small part of the whole system $|S| \ll N-|S|$, it is well known that the reduced states of highly entangled states are (almost) maximally mixed \cite{Popescu2006}, i.e. proportional to the identity. Moreover, based on the data available in the literature \cite{Eisert2010,CastilhoAlcaraz,Palmai2014,Parker2017,Wong2013,Bhattacharya,Alba2009,Ares2014,Molter2014}, there is agreement on the fact that, away from integrability, the energy eigenstates in the bulk of the spectrum have a large amount of entanglement. Thus, if the eigenstates $\ket{E_m}$ are all highly entangled $ \Tr_{\overline{S}} \ketbra{E_m}{E_m} \approx \1_{S} / d^{|S|}$ and we have
\begin{equation}
  | \braket{E_m}{a_j,k} |^2 \approx 1/d^N \,.
\end{equation}

This way of arguing shows how entanglement in the energy basis can lead to the emergence of the ETH in a local observable.
While this result was expected for the diagonal part of ETH, we would like to stress that it is a non-trivial statement about the off-diagonal matrix elements. Since the magnitude of the off-diagonal matrix elements controls the magnitude of fluctuations around the equilibrium values, their suppression in increasing system size is of paramount importance for the emergence of thermal equilibrium.
If one assumes high-entanglement in the energy eigenstates, it is trivial to see that  $A_{mm} \approx \Tr{A}/D$.
Moreover, thanks to the HUO construction and Theorem~\ref{thm:mainresult} we can also make non-trivial statements (as Eq.~\eqref{eq:HUOandETH}) about the off-diagonal matrix elements.

The physical picture that emerges is the following: Entanglement in the energy eigenstates is the feature which makes a local observable satisfy the statement of the ETH. If the energy eigenstates are highly entangled in a certain energy window $I_0 = \left[E_{a} , E_{b} \right]$, as it is expected to happen in a non-integrable model, the ETH will be true for local observables, in the same energy window. 

We now turn our attention to the study of extensive observables and assume that we are interested in a certain energy window $\left[ E_a,E_b\right] $ which contains $M \leq D$ energy eigenstates.
The details of the computations can be found in Appendix \ref{AppB}.
The paradigmatic case that we study is the global magnetization $M_z \coloneqq \sum_{i=1}^{N} \sigma_i^z$.
Writing its spectral decomposition we have $M_z = \sum_{j=-N}^N j \Pi_j$, where the degeneracy $\Tr \Pi_j = D_j$ of each eigenvalue $j$ can be easily computed to be $D_j = {N \choose \frac{N-j}{2}}$.
Again, we call $\mathcal{H}_j \subset \mathcal{H}$ the image of the projector $\Pi_j$.
The inequality  $D_j(D_j -1) \geq M$ selects a subset $j\in [-j_{*}(M),j_{*}(M)]$ of spaces $\mathcal{H}_j$ for which the conditions of our theorem are satisfied.
Small $M$ will guarantee that the hypotheses of the theorem are satisfied in a larger set of subspaces $\mathcal{H}_j$.
If we are interested in the whole energy spectrum $M=D$, a rough estimation, supported by numerical calculations, shows that $j_{*}(D)$ scales linearly with system size: $j_{*}(D) \simeq 0.78 N$.
The physical intuition that we obtain is the following:
Subspaces with ``macroscopic magnetization'', i.e. around the edges of the spectrum of $M_z$, have very small degeneracy and the theorem does not yield anything meaningful for them.
However, in the bulk of the spectrum there is a large window $j\in [-j_{*}(D),j_{*}(D)]$ where the respective subspaces $\mathcal{H}_j$ meet the conditions for the applicability of the theorem.
Therefore $ \forall j \in \mathbb{Z} \cap [-j_{*}(D),j_{*}(D)] $ we have 
\begin{align}
 |\braket{E_m}{j,s}|^2 = \frac{\bra{E_m} \Pi_j \ket{E_m}}{D_j} \,.
\end{align}
If, for some physical reasons, one is not interested in the whole set energy spectrum but only in a small subset, the window $[-j_{*}(M),j_{*}(M)]$ will increase accordingly. Thanks to our theorem we can extract a physical criterion under which the global magnetization will satisfy ETH. Assuming that we can use Stirling's approximation, the  $M_z$ is a HUO iff 
\begin{equation}
  \bra{E_m} \Pi_j \ket{E_m} \approx 2^{-N H_2(p(j)|\!| p_{\mathrm{mix}})} \label{eq:ldt}\,,
\end{equation}
where $p(j) \coloneqq \left( \frac{1}{2} + \frac{j}{2N}, \frac{1}{2} - \frac{j}{2N}\right)$, $p_{\mathrm{mix}}\coloneqq p(0)$ and we used the binary relative entropy $H_2(p|\!|q) \coloneqq \sum_{k=1,2} p_k \log \frac{p_k}{q_k}$.
This relation has a natural interpretation in terms of large-deviation theory. Indeed, such a relation is a statement about the statistics induced by the energy eigenstates on the observable $M_z$. If such statistics satisfies large-deviation theory, as in Eq.~\eqref{eq:ldt}, the observable will satisfy ETH. A complete understanding of how this concretely happens goes beyond the purpose of the present work and it is left for future investigation.

We note that the hypothesis of the theorem do not hold for the whole spectrum of $M_z$. Moreover, the proven connection between HUOs and ETH relies on the applicability of the central limit theorem in the degeneracy space $\mathcal{H}_j$. Hence the picture that emerges is the following.
For extensive observables, ETH will hold if the statistics induced by the energy eigenstates satisfies a large deviation theory. If this is true, we do not expect it to hold through the whole spectrum but only in the subsectors with sufficiently high degeneracy.
Both statements fully agree with the intuition that, in the thermodynamic limit, macroscopically large values of an extensive sum of local observables should be highly unlikely. In the recent work by Biroli \emph{et al.} \cite{Biroli}, it was argued that in a chain of interacting harmonic oscillators, the measurement statistics of the average of the nearest-neighbour interactions, given by the diagonal ensemble, satisfies a large-deviation statistics. This allows for the presence of rare, non-thermal, eigenstates which can account for the absence of thermalization in some integrable systems. Our results goes along with such intuition. Indeed, if it is possible to show that a large-deviation bound emerges at the level of each energy eigenstate, for all of them, this would amount to a proof of ETH, as discussed before.

We now come to the last application of our theorem: the macro-observables originally proposed by von Neumann.
As for the two previous applications, more details can be found in the Appendix \ref{AppB}
As explained before, macro-observables induce a partition of the Hilbert space into subspaces in which such classical-like observables have all well defined eigenvalues. In this sense a \emph{macrostate} is an assignment of the eigenvalues of all these observables and the index $j$ runs over different macrostates.
By construction, each macrostate $j=1,\ldots, n$ corresponds to a subspace $\mathcal{H}_j$ of the whole Hilbert space which is highly degenerate and to which we can apply our theorem.
According to the result by von Neumann \cite{Neumann1929} and Goldstein et al. \cite{Goldstein2006} it can be proven that the following relation holds for a given partition, \emph{for most Hamiltonians, in the sense of the Haar measure}: $\bra{E_m} P_j \ket{E_m} = \frac{D_j}{D}$.
The $P_j$'s are the projectors onto the subspaces $\mathcal{H}_j$.
Our theorem tells us that there exists a basis $\left\{ \ket{j,s}\right\}$ which diagonalises all the macro-observables such that $\bra{E_m} P_j \ket{E_m} = D_j |\braket{E_m}{j,s}|^2 $.
Using it in synergy with the previously mentioned result we find: 
\begin{align}
&|\braket{E_m}{j,s}|^2 = \frac{1}{D}\,.
\end{align}
This means that for most Hamiltonians, those macro-observables have a common basis that is a HUB.
Given the huge degeneracy of the spaces $\mathcal{H}_j$ this in turn allows us to formulate the following statement: \emph{for most Hamiltonians, in the sense of Haar, the macro-observables are degenerate HUOs and therefore satisfy ETH \ref{hyp:originaleth}}.

\section{Conclusions}
The ETH captures the wide-spread and numerically very well corroborated intuition that the eigenstates of sufficiently complicated quantum many-body system have thermal properties.
Its importance stems from the fact that together with the results that constitute the framework of pure state quantum statistical mechanics, a proof of the ETH would yield a very general argument for the emergence of not just equilibration, but thermalization towards the prediction of equilibrium statistical mechanics from quantum mechanics alone.
Such a rigorous proof is, however, still missing, despite the progresses in recent years that have significantly improved our understanding of the ETH by means of proofs of related statements and counterexamples.
Here we contribute to this program by bisecting the problem of proving ETH in two sub-problems related to the relative phases and the the overlaps between the eigenstates of the Hamiltonian and an observable.
We argue that the ETH can fail because of the former only through conspiratorial correlations in the phases.
Our main result concerns the second half of the problem.
Here we prove a rigorous result that shows when highly degenerate observables satisfy this part of the ETH and become Hamiltonian unbiased observables.
We illustrate our results with three types of physical observables, local, extensive, and macroscopic observables and collect and compare different versions of the ETH.
Our approach allows us in particular to make statements about the off-diagonal elements that are prominent in the original version of the ETH.

\chapter{Thermal Observables in a Many-Body Localized System}\label{Ch5}


In nature, not all systems reach thermal equilibrium. Well-known examples are integrable systems and systems which exhibit Anderson localization (AL) or Many-Body localization (MBL). The latter are disordered models whose energy eigenstates, when expanded in a local basis, have a very narrow profile: they are localised. They constitute a counterexample to the ``local thermalization'' paradigm, summarised in Section \ref{sec:Ch2Typicality}, and therefore we should not trust statistical mechanics to predict their equilibrium behaviour. For this reason, such kind of systems offer a good platform to test the predictive power of the machinery developed so far. In this chapter, we study the emergence of thermal observables (see Chapter \ref{Ch3}) in systems with a localised phase. Moreover, building our intuition on the notion of HUO, we show how they can be exploited to unravel the dynamical phenomenology of a MBL system. The chapter is organised as follows. First, we give a quick summary of the phenomenology of localised systems. Then, we proceed with the study of thermal observables. We will briefly look at their equilibrium behaviour and exploit the results to propose an experimentally accessible way to quantify the logarithmic growth of entanglement in time. This chapter is based on the work done by the author, partially in collaboration with Dr. F. Pietracaprina and Prof. J. Goold. 


\section{Localised Systems}

The study of transport properties of quantum systems is a topic of paramount relevance in condensed matter physics. To such purpose, a crucial aspect is the presence of disorder, usually due to the existence of defects and irregularities within the material under study. In a celebrated work \cite{Anderson1958}, Anderson showed how the presence of strong disorder in a lattice was able to completely suppress transport, in a free-electrons model. More recently, Basko et. al. \cite{Basko2006} argued that such phenomenon was stable with respect to the addition of interactions between particles. Such seminal work effectively showed the existence of a new dynamical phase of matter which we now call \emph{Many-Body Localized} (MBL).  In the last ten years there has been a large amount of effort devoted to understand the defining properties and the phenomenology of MBL systems. Such efforts recently culminated in a series of interesting reviews \cite{Abanin2017,Alet2017a,Scardicchio2017,Nandkishore2015,Altman2015,Imbrie2017} describing the behaviour of quantum systems in the MBL phase. Now we introduce Anderson Localization and, in the next subsection, summarize the most important results about the phenomenology of MBL systems.

\subsection{Anderson Localization}

The core of the localization phenomenology can be described using the original model studied by Anderson in his seminal paper in 1958\cite{Anderson1958}:
\begin{equation}
\Ham_{\mathrm{AL}} = t \sum_{i=1}^L c_i^\dagger c_{i+1} + c_{i+1}^\dagger c_{i} + \sum_{i=1}^L U_i c_i^\dagger c_i \,\,. \label{eq:Anderson}
\end{equation}
Here $c^\dagger_i$ and $c_i$ are the creation and annihilation operators on the site $i$ and $U_i$ represents a random on-site potential where $U_i \in u([-W,W])$ and $u([-W,W])$ is the uniform distribution on the interval $[-W,W]$. For any value of the disorder $W$ all eigenstates $\Ket{E_\alpha^{\mathrm{AL}}}$ of $\Ham_{\mathrm{AL}}$ are localized, with exponentially suppressed coefficients $|\braket{x}{E_\alpha^{\mathrm{AL}}}| \sim e^{-|x-x_{\alpha}|/\xi_\alpha}$. The localization length $\xi_\alpha$ depends on the disorder strength $W$. Each state $\ket{E_\alpha^{\mathrm{AL}}}$ is localized around some value $x_\alpha$ and all many-particle energy eigenstates are products of single-particle eigenstates. Although this model provides the best-known example of single-particle localization, the phenomenon is more general and it can be also be observed in spin systems. The canonical model used for numerical simulation is the XX model with transverse field
\begin{align}
&\Ham_{XX} = J \sum_{i=1}^L  \sigma^x_i  \sigma^x_{i+1} + \sigma^y_i \sigma^y_{i+1} + \sum_{i=1}^L B_z^i \sigma^z_i \,\,, \qquad B_z^i \in u([-W,W]) \,\, .
\end{align}
It can be mapped into the Anderson free-electron models in Eq.(\ref{eq:Anderson}) by means of the Jordan-Wigner transform\cite{Jordan1928}. Thanks to presence of such high degree of localization it can be shown that the Schroedinger dynamics does not bring the system to thermal equilibrium. If we start with a state which has a localised configuration, the system will indefinitely remain very close to such configuration and we will not observe the spreading of the wave-packet. This can be seen as a consequence of the fact that eigenstates are exponentially localized and thus violate the ETH. This phenomenology has to be compared with the standard case of a particle hopping on a clean lattice, where an initially localised wave-packet will rapidly spread over the whole lattice. We now turn our attention to the interacting version of this phenomenon: Many-Body Localized systems.

\subsection{Many-Body Localization}



The Anderson model is a free-fermions model and the idea behind MBL is that the localization phenomenon should survive even in presence of interactions. The stability of the localised phase with respect to the addition of weak interactions was initially argued for with a perturbative treatment\cite{Altshuler1997a,Gornyi2005,Basko2006,Aleiner2010}. Since then, a large body of numerical investigations, mostly focused on fermionic and spin models, has revealed that localization properties can survive in the strongly interacting regime\cite{Kjall2014,Serbyn,Znidarica,DeLuca2013,Andraschko2014,Wei2018,Ponte2015,Bardarson2012}. Thanks to the marvelous improvement in techniques to manipulate and isolate quantum systems, there are nowadays experimental protocols which are able to realize systems in the MBL phase and measure some of their properties\cite{Wei2018,Smith2016,Choi2016,Schreiber2015,Kondov2015}. 

Many aspects of the MBL phenomenology have been (theoretically) explored so far, with the aim of giving a complete characterization of the phenomenon. Many diagnostic tools were used to address different properties of an MBL system. Thanks to such large body of literature we can list here a set of physical properties of the MBL phase which have been, more or less rigorously, understood. This list was compiled by J. Z. Imbrie, V. Ros and A. Scardicchio in Ref.\cite{Imbrie2017}, where more details were given. Here we briefly summarise the overall picture and send the interested reader to the original reference.
\begin{itemize}
\item[(1)] {\bf Absence of DC transport}. In the Anderson case it is possible to give rigorous arguments for the vanishing of the diffusion coefficient. Similar arguments can be given for the MBL phase. They mostly rely on the exponential decay of correlations of the local density operator on the many-body energy eigenstates. The absence of diffusion was confirmed numerically by studying the dynamical conductivity, the spin-spin or the density-density correlation functions in the infinite-time limit. 

\item[(2)] {\bf Anderson localization in configuration space}. Anderson localization can be seen as localization on an abstract graph where each site corresponds to the Fock states of the Hamiltonian. A similar picture is available in the MBL case, where each site is a common eigenstate of all the local operators $\sigma_i^z$. From the point of view of a single energy eigenstate this means that they are weak deformation of the non-interacting eigenstates of all the $\sigma_i^z$. This has been confirmed with numerous numerical investigations.  

\item[(3)] {\bf Area-law entanglement in highly excited states}. In the bulk of the spectrum, away from the edges, thermalising systems are expected to have a large amount of entanglement, growing extensively with the size of the system. This feature is numerically investigated with the use of the half-chain entropy: the von Neumann entropy of the reduced state of half of the system. Systems in the MBL phase exhibit a clear violation of this behaviour as their half-chain entanglement entropy has a scaling with system size $L$ which is sub-extensive: it satisfies an area-law. This behaviour is typical of the low-lying energy eigenstates of a gapped Hamiltonian but in the MBL phase it extends to the whole spectrum.

\item[(4)] {\bf Violation of Eigenstate Thermalization}. The localization properties of both the MBL and AL phases are structurally non-compatible with the ETH. This is mostly due to the lack of a large amount of entanglement in the energy eigenstates. In the previous chapter we argued that a large amount of entanglement is sufficient to guarantee that local observables satisfy the ETH. For this reason, the lack of an extensive amount of entanglement can be seen as evidence that, among all local observables, there is at least one which violates the ETH. We will come back to this point later.

\item[(5)] {\bf Absence of level repulsion in the energy spectrum}. In quantum chaos theory an important tool to diagnose the transition from chaotic (ergodic-like) to integrable behaviour is given by the probability distribution of the level spacing $s=E_n - E_{n-1}$. The same idea can be applied to the Ergodic-MBL transition and the level spacing distribution can be used as a diagnostic to study the transition. On the ergodic side this is expected to exhibit a typical Wigner form, with its level repulsion. On the MBL side this is expected to behave as for integrable systems and exhibit a Poisson distribution shape.

\item[(6)] {\bf Slow growth of entanglement and slow relaxation}. Despite the fact that the structure of energy eigenstates is only weakly modified by the presence of interactions, these have a strong impact on the out-of-equilibrium dynamics. Indeed, if we look at the time-dependent behaviour of the half-chain entropy we observe a slow growth of entanglement which saturates to an extensively large value. The dynamical profile of the half-chain entropy seems to have a logarithmic shape. Several independent numerical investigations supports this behaviour. So far, this is the only diagnostic tool that we have which is concretely able to discern between the Anderson and the Many-Body Localised phase. We will come back to this point later.
\end{itemize}


Such an interesting plethora of phenomena can be effectively understood if we assume that there is an extensive number of almost-local conserved quantities: the \emph{quasi-local integrals of motion} (Q-LIOMs). This gives rise to a unifying framework where to understand the localization phenomenon, via the presence of the Q-LIOMs. 

\subsection{The Q-LIOMs picture}

It is largely believed that the lack of ergodicity in a system should be traced back to the existence of some, possibly local, conserved quantities. Based on a similar intuition, in a recent set of papers\cite{Serbyn2013,Huse2014,Serbyn2013a,Bardarson2012} it was argued that the dynamical behaviour of MBL systems can be understood through an Hamiltonian $\Ham_{MBL}$ which is a nonlinear functional of a complete set of conserved quantities $\tau_i^z$, which have been called $l$-bits. If we restrict the discussion to a chain of spin-$1/2$ systems this means 
\begin{align}
&\Ham_{MBL} = E_0 + \sum_i \lambda_i^{(1)} \tau_i^z + \sum_{i,j} \lambda_{ij}^{(2)} \tau_i^z \tau_j^z + \sum_{n=3}^{\infty} \lambda_{j_1,\ldots,j_n}^{(n)} \tau_{j_1}^z \tau_{j_2}^z \ldots \tau_{j_n}^z \,\, . \label{eq:QLIOMs}
\end{align}
Here the $\left\{\tau_i^z \right\}$ are a extensive and commuting set of conserved quantities which satisfy the $SU(2)$ algebra $[\tau_i^a,\tau_j^b] = \delta_{ij} i \epsilon_{abc} \tau_i^c$. The typical magnitude of the interaction coefficients (the $\lambda$s) falls off exponentially with distance. The $l$-bit $\tau_j^z$ has a support on the chain which is exponentially localised around the $j$-th site of the chain. This notion of localization extends the familiar localization properties of the energy eigenstates in the AL case. If we write a decomposition of the $l$-bits into a basis of local operators $\left\{\mathcal{O}_m\right\}$ we have $\tau_j^z = \sum_{m} C^{j}_m \mathcal{O}_m$ and
\begin{align}
&|C^j_m| \sim  e^{- \frac{d[j,S(m)]}{\xi}}\,\,,
\end{align}
where $\xi$ is the localization length, $S(m)$ is the support of the operator $\mathcal{O}_m$ and $d[j,S(m)]$ stands for the distance between the site $j$ and the furthest site on which the operator $\mathcal{O}_m$ has support on. This means that the $j-$th $l$-bit has a support on the chain which is exponentially suppressed with the distance from the site $j$ and it can be obtained by applying a quasi-local rotation $\Omega$  to the local spins $\sigma_i^z$: $\tau_i^z = e^\Omega \sigma_i^z e^{-\Omega} $ A cartoon to visualize this behaviour can be found in Fig.(\ref{fig:lbit}).
\begin{figure}[h!]
\begin{center}
\includegraphics[scale=0.2]{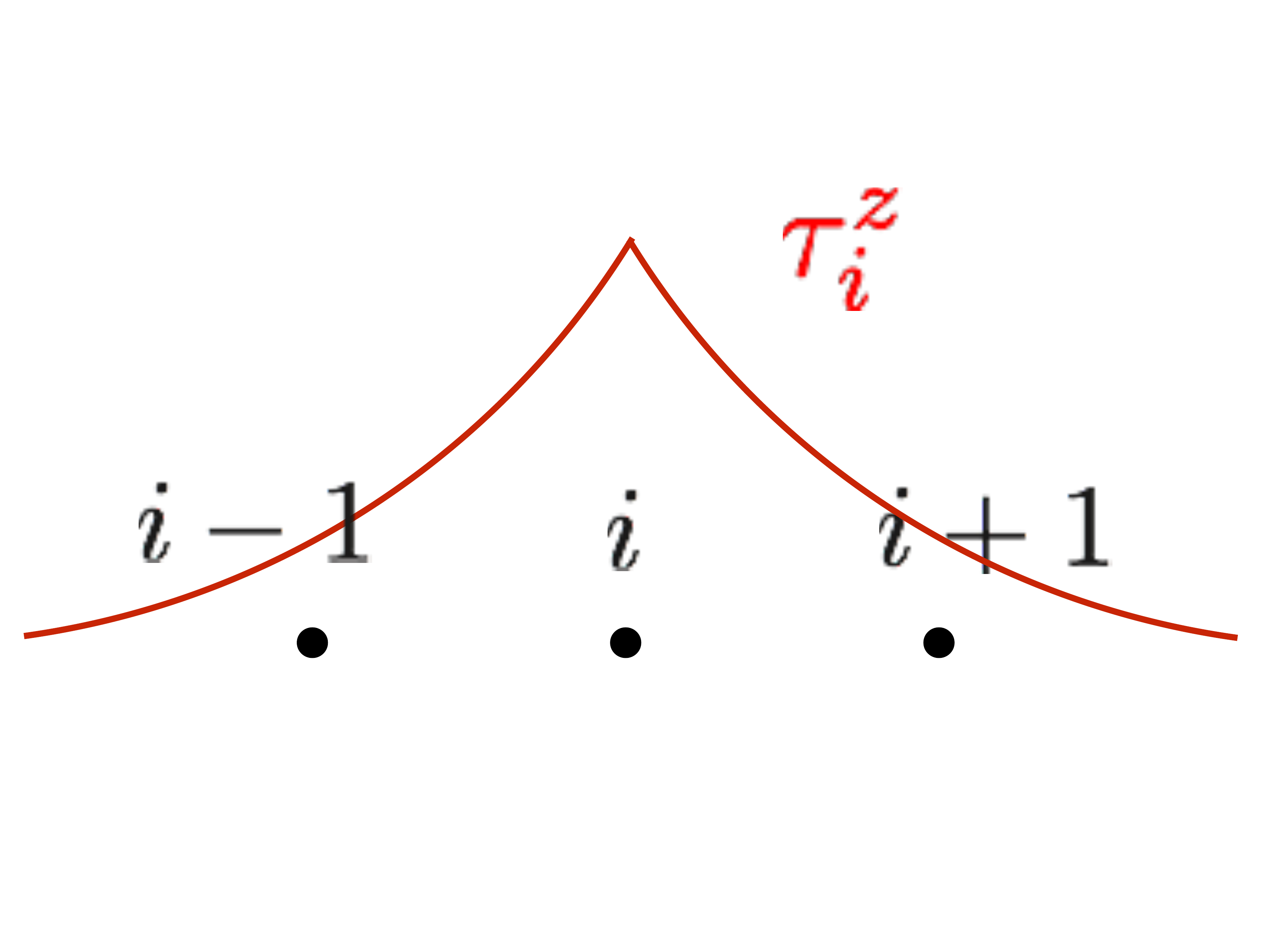}\caption{Exponentially suppressed support of an $l$-bit.}\label{fig:lbit}
\end{center}
\end{figure}

All the properties mentioned above can be accounted for by assuming the existence of the Q-LIOMs in an MBL system and a functional form of the Hamiltonian as in Eq.(\ref{eq:QLIOMs}). As in Ref. \cite{Nandkishore2015}, we summarize them in Table \ref{tab:MBL}. Among all of them, an important role is played by the logarithmic growth of entanglement in time. This is an important feature of the MBL phase which allows us to distinguish it from the AL phase. While such dynamical behavior was discovered with numeric analysis and confirmed in independent studies, it still eludes experimental confirmation. This is mainly due to the highly non-local character of the half-chain entanglement entropy, which is a hardly measurable quantity. Such difficulty gave rise to a plethora of works which aim at providing alternative ways to witness such dynamical behavior\cite{Campbell2017,Iemini2016,Serbyn2014a,DeTomasi2017,Bera2016,Goold2015,Serbyn2014}. Despite that, the logarithmic growth of information is still a dynamical feature which escapes the currently available measurement strategies. In this chapter we build on the results of the previous chapters and study the emergence of thermal behaviour in MBL systems. We will see that the ``Observable's Thermalization'' picture developed so far allows to make predictions about the emergence of thermalization in local observables, despite the system being in a localized phase. Such predictions have been tested against numerical investigations of MBL systems and show full agreement with the theoretical picture. Moreover, building on the intuition developed so far about the HUOs, we devise an efficient strategy to measure the logarithmic growth of entanglement. By focusing on realistically measurable observables, as the single-site ones, we show that the logarithmic spread of entanglement is encoded in the behavior of the single-site observables. These are  accessible with the available techniques of local quantum tomography.



\begin{table}[h!]
\begin{adjustwidth}{-0.8in}{-.5in}  
\begin{center}
    \caption{Comparative table for Thermal, AL and MBL phases}
    \label{tab:MBL}
    \begin{tabular}{c|c|c}
      \toprule 
      \rowcolor{red}
     \textbf{Thermal phase} & \textbf{Anderson Localized} & \textbf{Many-Body Localized}\\
\midrule
 Memory of initial conditions  & Local observables retain memory & Local observables retain memory  \\
      in non-local observables. & of initial conditions, at any time. & of initial conditions, at any time. \\
      \hline
      ETH holds for all  & ETH violated for some  & ETH violated for some \\ 
      local observables & local observables & local observables \\
      \hline
      Transport of energy allowed & No transport of energy & No transport of energy \\
            \hline
      Eigenstates with   & Eigenstates with  & Eigenstates with\\
      volume-law entanglement & area-law entanglement & area-law entanglement \\
            \hline
      Power-law spreading  & No spreading  & Logarithmic spreading \\
      of entanglement & of entanglement & of entanglement \\
	\hline
      Dephasing and dissipation &   No Dephasing and no dissipation & Dephasing but no dissipation \\
      \hline
     Level repulsion &   No Level repulsion & No Level repulsion \\
      \bottomrule 
    \end{tabular}
    \end{center}
\end{adjustwidth}
\end{table}

\section{Thermal observables in MBL}\label{sec:ThermObsMBL}

The existence of a large number of MUBs can be guaranteed if one assumes that the total dimension of the Hilbert space is the power of a prime number\cite{Bengtsson2007,Bandyopadhyay2002}. For this reason, in all fermionic/spin-$1/2$ systems with $L$ particles there are $2^L+1$ MUBs. If we choose the eigenstates of the Hamiltonian as our reference basis, this means that there are $2^L$ HUBs.  It is therefore natural to wonder what is the physical meaning of all these bases. Since they all admit an algorithmic construction, in principle one could diagonalize the Hamiltonian and build them to study their physical behaviour. Unfortunately this is highly unpractical, for two reasons. First, their physical meaning strongly depends on the concrete Hamiltonian model we are considering. Second, their number grows exponentially with the size of the system so the growth of computational complexity with the size of the system is rather unfavourable. One possibility could be to focus on a few of them, but we do not have a general criterion to select a few ``physically relevant'' HUBs. However, for MBL systems we can exploit the intuition provided by the existence of the Q-LIOMs.  

Starting from Eq.(\ref{eq:QLIOMs}) we know that the set $\mathcal{B}^\tau_z \coloneqq \left\{ \tau_i^z\right\}_{i=1}^L$ is a complete set of commuting observables which defines the eigenbasis of the Hamiltonian. Together with the sets $\mathcal{B}^\tau_x \coloneqq \left\{ \tau_i^x\right\}_{i=1}^{L}$ and $\mathcal{B}^\tau_y \coloneqq \left\{ \tau_i^y\right\}_{i=1}^L$ they provide $L$ copies of the quasi-local $SU(2)$ algebra. For this reason it can be easily shown that $\mathcal{B}^\tau_x$ and $\mathcal{B}^\tau_y$ are HUBs and all observables which are diagonal in these basis are HUOs. In the strong-disorder regime the disordered term of the Hamiltonian will dominate the interaction. Hence, it is reasonable to expect that, deep in the MBL phase, the eigenstates of the effective model in Eq.(\ref{eq:QLIOMs}) will be close to being eigenstates of the disordered term. In the case of the disordered XXZ model
\begin{align}
&\Ham_{XXZ} = \sum_{i=1}^L  J \sigma^x_i  \sigma^x_{i+1} + J \sigma^y_i \sigma^y_{i+1} + \Delta \sigma^z_i \sigma^z_{i+1} + \sum_{i=1}^L B_z^i \sigma^z_i \,\,, \label{eq:XXZDis}
\end{align}
where $B_z^i \in u([-W,W])$, this means that the Hamiltonian eigenstates will be close to the eigenstates of the local magnetization along the $z$ direction: $\mathcal{B}_z=\left\{\sigma_i^z \right\}$. We conclude that, deep in the MBL phase, $\mathcal{B}_x$ and $\mathcal{B}_y$ are very close to being HUB and all observables which are diagonal in such basis will be HUO. Following the intuition developed in the previous chapters this should lead to emergence of thermal equilibrium for these observables. We investigated the validity of such claims by means of numerical simulations. Using exact diagonalization we studied the out-of-equilibrium behaviour of the disordered, isotropic ($\Delta=1$) XXZ model in Eq.(\ref{eq:XXZDis}). This exhibits a transition from the thermal $(W < W_c)$ to the MBL $(W>W_c)$ regime at disorder strength $W_c \approx 3.7$ \cite{Luitz2015}. 

\subsection{Dynamical study of observable thermalization}

Before we proceed with the numerics we would like to discuss what we believe is a proper setup to address the dynamical emergence of thermalization, for a fixed, but otherwise generic, observable $A$. The formalism developed in the previous chapters is based on the intuition that, when we talk about thermalization we need to specify the observable we are interested in. This is due to the fact that, if we have access only to the observable $A=\sum_{j}a_j A_j$ we can probe only the diagonal matrix elements of $\rho$ in the eigenbasis of $A$. This defines the probability distribution $p_\rho(a_j)$ for $A$: $p_\rho(a_j) = \Tr{\rho A_j}$. The notion of thermalization we developed in the first chapter is based on the maximisation of the Shannon Entropy of such probability distribution. Given that the HUO satisfy the equilibrium equations and the ETH, if our predictions are correct we should see the dynamical emergence of the maximum entropy principle for an HUO. In other words, by simulating the exact unitary dynamics we should see that, after a transient, the time-dependent probability distribution $p_t(a_j)$ should equilibrate to the prediction given by the maximum entropy principle. For this reason, when we address the problem of thermalization of an observable $A$ we will use as initial state any eigenstate $\ket{a_j}$ of $A$. The reason behind this choice is that such state has a probability distribution for $A$ which as far away as possible from equilibrium, having zero Shannon entropy. Indeed, if $\ket{\psi_0}=\ket{a_{x}}$ we have $p_{0}(a_j) = |\braket{\psi_0}{a_j}|^2 = \delta_{jx}$, whose entropy is clearly zero. Because of that, we are studying how the out-of-equilibrium unitary dynamics can justify the use of the maximum of Shannon entropy principle stated in Chapter \ref{Ch3}. 

Now we proceed with the investigation of this problem for a specific class of observables: the local magnetizations along the three directions $x,y$ and $z$. In particular, we are interested in studying the dynamics of the local HUOs $\sigma_x^i$ and $\sigma_y^i$. If our predictions are correct, if we start from an out-of-equilibrium state for a HUO, we should observe thermalization of the whole probability distribution. Since we are dealing with local magnetization, which is a binary observable for the spin-$1/2$ case, the only linearly independent quantity is the average value, given the normalization of the probability distribution.

\paragraph{Parameters of the numerics} We start by showing the behaviour of the disordered XXZ model below and above the transition at $W=W_c$. Here we show the plots for two representative values of the disorder strength: $W=1$ for the thermal phase and $W=10$ for the MBL phase. The model was also investigated for other values of the disorder strength: $W=0.5, 1, 2 ,5,7,10$. Within the thermal and the MBL phase the results are qualitatively the same as the one showed, respectively, for the $W=1$ and $W=10$ cases. Moreover, we studied the behaviour for chains of increasing lengths $L=6,8,10,12$. The results presented are qualitatively independent on the size of the system. Two notable points are (i) that time fluctuations appear to be less violent for increasing system size (ii) already for chains of very small sizes we observe a clear time-dependent profile, despite the fact that we are far away from the macroscopic regime. Here show the plots for the $L=10$ case. 
\begin{figure}[h!]
\centering
\includegraphics[scale=0.5]{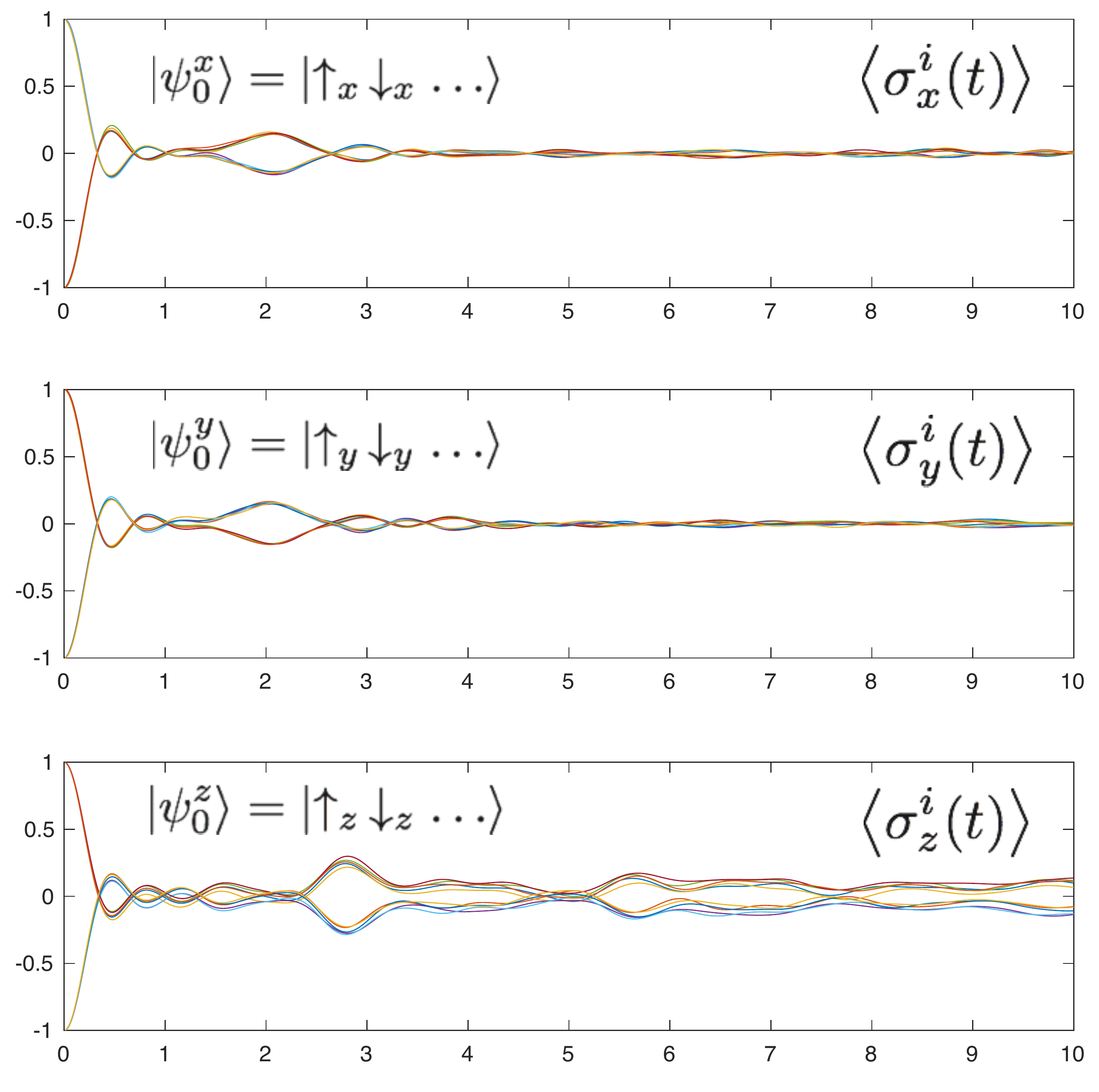}\caption[Local magnetization in the thermal phase]{Here we show the behaviour of the local magnetization along the three directions $x,y,z$. All observables are initialised in an out-of-equilibrium probability distribution given by the antiferromagnetic order and evolved dynamically with a unitary propagator generated by the Hamiltonian of the XXZ model. The parameters of the system are: size $L=10$, interaction $\Delta=1$ and disorder strength $W=1$, which means that the system is in the thermal phase. We average over $N_{dis}=100$ disorder realizations. We observe that in this phase all local observables dynamically thermalize to the expectation value given by the microcanonical ensemble.}\label{fig:LocalThermal}
\end{figure}
Now we show the behavior of the local magnetization along the three directions $x,y,z$, when the initial state of the evolution is the respective Neel state $\ket{\psi_0^\alpha} = \ket{\uparrow_\alpha \, \downarrow_\alpha \ldots}$ with antiferromagnetic order along the $\alpha=x,y,z$ direction. The fact that our dynamical simulations begin with a Neel state does not influence the qualitative picture and our conclusions. Other states belonging to the three basis $\mathcal{B}_x,\mathcal{B}_y,\mathcal{B}_z$ can be used as initial state for the respective observable and the qualitative results seems to be independent on the specific choice.
\begin{figure}[h!]
\centering
\includegraphics[scale=0.5]{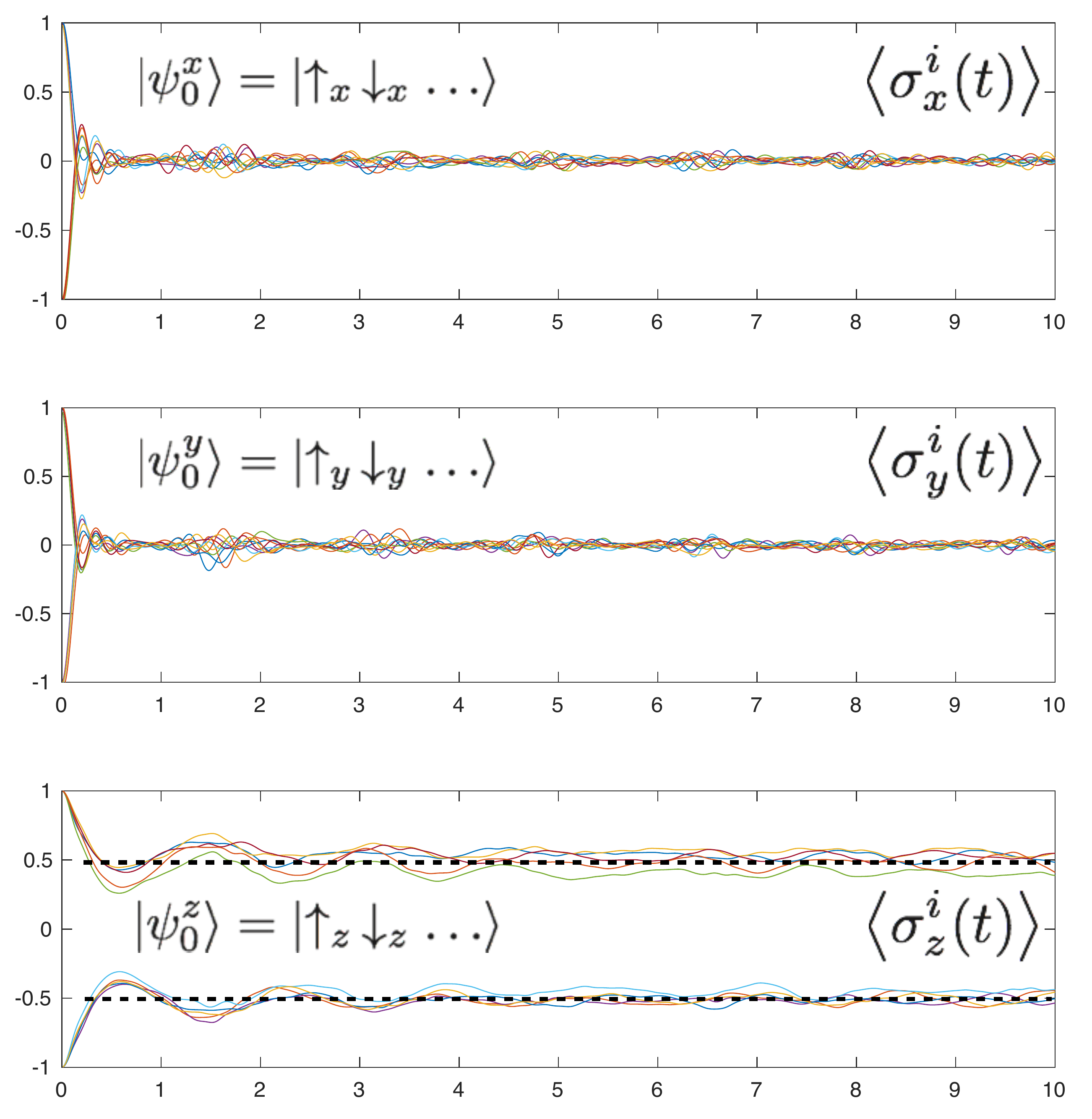}\caption[Local magnetization in the localized phase]{Here we plot the behaviour of the local magnetization along the three directions $x,y,z$. All observables are initialised in an out-of-equilibrium probability distribution given by the antiferromagnetic order and evolved dynamically with a unitary propagator generated by the Hamiltonian of the XXZ model. The parameters of the system are: size $L=10$, interaction $\Delta=1$ and disorder strength $W=10$, which means that the system is in the MBL phase. We average over $N_{dis}=100$ disorder realizations. We observe that the magnetization along the $z$ axis fails to relax to the maximum entropy distribution. This signal the fact that we are in a localised phase. Despite that, we can see that local HUOs, as the local magnetization along the $x$ and the $y$ axis, still dynamically relax to the prediction of the maximum entropy principle.}\label{fig:LocalMBL}
\end{figure}
On the one hand in Fig.~\ref{fig:LocalThermal} we see that in the thermal phase all local observables, when initialised in their antiferromagnetic order eigenstate, dynamically evolve toward a thermal expectation value given by the maximum of Shannon entropy principle. On the other hand, in Fig.~\ref{fig:LocalMBL} we can see that all the $\sigma_z^i$ fails to dynamically reach a thermal expectation value, when the system is in the MBL phase. Despite that, for the local HUOs as the local magnetization along the $x$ and $y$ direction we still observe thermalization, despite being in the MBL phase. This results provide evidence in favour of the picture developed so far about thermalization of observables and HUOs. In this sense this provide a straightforward example of the fact that the theory developed so far constitutes a generalization of statistical mechanics. Indeed, as we can see from Fig. \ref{fig:LocalMBL} the local magnetization along $z$ fails to thermalise, due to localization. This also means that the reduced density matrix $\rho_i(t)$ of a site $i$ fails to relax to a thermal ensemble. Despite that, if we are interested in the magnetization along the $x$ and $y$ direction we can still use statistical mechanics to predict their equilibrium behaviour.

\section{Logarithmic growth of entanglement}

The logarithmic growth of entanglement is the dynamical property of the MBL phase which allows us to distinguish it from the AL phase. Indeed it is a consequence of the existence of the exponentially small tails of the Q-LIOMs and of their interaction in the effective model. While such dynamical behavior was confirmed by several numeric investigations, it still eludes experimental confirmation. This is mainly due to the highly non-local character of the half-chain entanglement entropy, which is a hardly measurable quantity. 

Here we overcome such obstacle and devise an efficient strategy to measure the logarithmic growth of entanglement. By focusing on realistically measurable observables, as the single-site ones, we show that the logarithmic spread of entanglement is encoded in the behavior of the single-site observables. These are accessible with the available techniques of local quantum tomography. If one has access to the reduced density matrices of all sites, it is also possible to compute the \emph{Total Correlations}: a quantity which, at equilibrium, has already been proven to signal the transition from ergodic to the MBL phase\cite{Goold2015,Pietracaprina2017}. Using the total correlation instead of the half-chain entropy brings also a computational advantage. For a spin$-1/2$ chain, to compute the half-chain entropy we need to diagonalize a matrix of dimension $2^{L/2}$. This requires $\mathcal{O}\left(2^{3L/2}\right)$ elementary operations while for the total correlations we only need $\mathcal{O}(L)$. This gives an exponential speedup in the computation of the quantity of interest.

\subsection{Local entropy and total correlations}
Throughout the paper we assume to deal with a one-dimensional isolated quantum system, made by $L$ spin-$1/2$. The dynamics is generated by a Hamiltonian $\Ham$, made by the sum of two terms: a clean one $\Ham_0$ and a disordered one $\Ham_1 (W)$, where $W$ is the strength of the disorder. If $\ket{\psi_t}$ is the state of the whole system at time $t$, the reduced state of the $n$th site is obtained by tracing out the complement $L/n$: $\rhont \coloneqq \Tr_{L / n} \ket{\psi_t} \bra{\psi_t}$. The bipartite entanglement between the $n$th site and the rest is quantified by the von Neumann entropy of the reduced state:
\begin{equation}
S_n(t) \coloneqq - \Tr \rhont \log \rhont 
\end{equation}
This quantity can be measured by means of local tomography. Moreover, if it is possible to do it over all sites $n=1,\ldots,L$ one can also evaluate the average entropy over all sites:
\begin{equation}
S(t) = \frac{1}{L}\sum_{n=1}^L S_n(t) \, .
\end{equation}
This quantity has the following operational meaning. Let $\mathcal{P} \subset \mathcal{H}$ be the set of all tensor product states of an $L$-partite quantum system. The total correlations $T(\rho)$ of a state $\rho$ is defined as
\begin{equation}
T(\rho) = \min_{\pi \in \mathcal{P}}S(\rho |\!| \pi) \, ,
\end{equation}
where $S(\rho |\!| \sigma)\coloneqq - \Tr \rho \log \sigma - S(\rho)$ is the relative entropy. 
$T(\rho)$ measures the relative entropy between $\rho$ and the closest product state $\pi_\rho \in \mathcal{P}$. It turns out that, for each $\rho$ such state is unique. It is the product state of the reduced density matrices obtained from $\rho$: $\pi_\rho = \rho_1 \otimes \ldots \otimes \rho_L$ where $\rho_i \coloneqq \Tr_{L/i}\rho$. In our case it is easy to see that the total correlations $T_t\coloneqq T(\ket{\psi_t}\bra{\psi_t})$ are simply a rescaling of $S(t)$
\begin{equation}
T_t = \sum_{n=1}^L S(\rhont) - S(\ket{\psi_t}\bra{\psi_t}) = L S(t) \label{eq:Tot2}
\end{equation}
By studying the dynamical behavior of the set of all $\{ S_n(t) \}$ and showing that they exhibit a logarithmic modulation in time here we achieve two goals. On the one hand we show that, in order to observe the logarithmic spread of entanglement in time it is sufficient to have access to the density matrix of a single site. On the other hand, we unravel the relevance of the total correlations as a tool to detect the dynamical growth of entanglement in a MBL phase. In \cite{Goold2015} the authors proposed to use the total correlations, at equilibrium, as a tool to detect the breaking of ergodicity. Our result goes along with such intuition. By showing that the time-dependence of $T_t$ is modulated on a logarithmic scale we provide a dynamical picture which complements the intuition developed in \cite{Goold2015}. We now give a theoretical argument which illustrate why each one of the $S_n(t)$ should approach the saturation value with a logarithmic law.

\subsection{Logarithmic spread of entanglement}\label{sec:logspread}
In \cite{Serbyn2013} it was proposed a theoretical picture, based on perturbation theory, to understand this phenomenon. Here we review such argument and explore its implications for the growth of the local entropy. The main intuition is based on the presence of the exponentially suppressed tails of the Q-LIOMs, the decay constant being the localization length $\xi$. Call $V$ the coupling constant of the interaction term, which gives its energy scale. If there are no interactions ($V=0$) all energy eigenstates are single-particle excitations, since there is Anderson Localization. In presence of interactions we have the Q-LIOMs. In this case, if two particles are placed at a distance $x_{ij}$ they would have an interaction energy which is exponentially suppressed because of the exponentially small tails $V_{ij} \sim V e^{-x_{ij}/\xi}$. The dephasing time between them is therefore $t_{ij}\sim \hbar/V_{ij}=\hbar e^{x_{ij}/\xi}/V$. It was therefore argued and numerically confirmed that the dynamical behaviour of the half-chain entropy $S_{L/2}(t)$ exhibits a logarithmic profile in time.

We now look at the implications of this argument for the bipartite entanglement between a single site and the rest. The degrees of freedom on a lattice site $n$ will become entangled with the degrees of freedom living on the $n+k$ site on a time-scale $t_k \sim t_{\mathrm{min}}e^{k a/\xi}$ where $a$ is the distance between any two nearest neighbor   sites. As time goes by the $n$th site will become entangled with an increasing number of sites. Therefore, the higher the number of sites which have entanglement with the $n$-th one, the higher $S_n(t)$. Because of what has been said before, such bipartite correlations will pile-up on a logarithmic time scale. From this we expect a logarithmic growth of $S_n(t)$. We would like to stress here that the argument holds for each one of the sites, for all of them. This is important, from an experimental point of view, because it means that it is not necessary to perform an extensive number of local measurements but one can focus on a single site.

Despite that, if it is possible to reconstruct all the local density matrices $\rhont$, the quantity $S(t)$ can be used as a single quantifier. Its meaning is clear: it quantifies the average growth of entanglement between one site and the rest of the chain. Moreover, as we showed in Eq.(\ref{eq:Tot2}) it is a rescaling of the total correlations, which in turn provide an upper bound to the amount of multipartite entanglement present in the system.

\subsection{About initial states}
When performing numeric simulations of isolated quantum systems the choice of the initial state is of paramount relevance. This is usually driven by a criterion of experimental feasibility and a well-known example is the Neel state. It exhibits an anti-ferromagnetic order and can be prepared as the ground state of a local Hamiltonian. From the quantum information point of view such state is part of the so-called computational basis. The tensor-product basis of the $z$ components of the local spins: $\mathcal{B}_z \coloneqq \left\{ \ket{m^z_1} \ket{m^z_2}\ldots \ket{m^z_L}\right\}$, where $\ket{m^z_i} \in \left\{ \ket{\uparrow_z},\ket{\downarrow_z}\right\}$. However, as we argue now, experimental feasibility is not the only criterion which should be taken into account.

Our goal here is to probe a dynamical feature of a quantum system. Therefore, our initial state should not be too close to a single energy eigenstate. Indeed, if this would be the case, the dynamics would always occur in the proximity of the initial state. In this sense, the time-evolution of observables would not carry much information about the properties of the dynamics. On the contrary, if we start with a state which is as far away from being a single energy eigenstate as possible, we would unravel the non-trivial effects of the dynamics on physically relevant observables. In this case, measuring the time-evolution of physically relevant observables would yield a more informative answer about the properties of the dynamics. 

For our purpose, the MBL phase is characterized by the presence of an extensive number of Q-LIOMs. This means that, deep in the MBL phase, the energy eigenstates are almost tensor product states. Hence, when approaching the high-disorder limit, an initial product-state can get very close to an energy eigenstate. This is what happens, for example, in the XXZ model for the Neel state. The disorder is in the nearest-neighbor term, along the $z$ direction. Therefore, deep in the MBL phase the Neel state is almost an energy eigenstate. For this reason, its time evolution is not ideal for the purpose of probing a dynamical feature of an MBL system.

A set of initial states which can be used to avoid such problem is given by the elements of an Hamiltonian Unbiased Basis\cite{Anza2017,Anza2018} (HUB). As argued in subsection \ref{sec:ThermObsMBL},   deep in the MBL phase the tensor-product states polarized along the $x$ or $y$ direction will be very close to being HUBs:
\begin{align}
& \mathcal{B}_x \coloneqq \left\{ \ket{m_1^x}\ldots \ket{m_L^x} \right\}  &&  \mathcal{B}_y \coloneqq \left\{ \ket{m_1^y}\ldots \ket{m_L^y} \right\} \, , \label{eq:HUBs}
\end{align}
where $\ket{m_i^x}\in \left\{ \ket{\uparrow_x},  \ket{\downarrow_x}\right\}$ and $\ket{m_i^y}\in \left\{ \ket{\uparrow_y},  \ket{\downarrow_y}\right\}$ for all $i=1,\ldots,L$.
An observable which is diagonal in a HUB is called Hamiltonian Unbiased Observable (HUOs). Defined in \cite{Anza2017} and studied in \cite{Anza2018}, highly degenerate HUOs have a phenomenology quite similar to the one of Random Matrices\cite{Torres-Herrera2016,Gubin2012,Torres-Herrera2017,Borgonovi2016,Balz2017,Reimann2016} and they are expected to satisfy the Eigenstate Thermalization Hypothesis\cite{Anza2018}. As such, if our initial state is an eigenstate of a HUO, its expectation value will quickly equilibrate to the predictions of RMT. The evolution starts with a Neel state polarized along the $x$ direction and we follow the expectation values of all $\sigma_i^x$. The dynamics is generated by a disordered XXZ model in the MBL phase. We see a fast equilibration dynamics which leads to thermalization. Such behavior is not peculiar to the specific choice of the Neel order for the initial state. Any initial state which belongs to the basis in Eq.(\ref{eq:HUBs}), and more generally which is part of a HUB, will exhibit the same behavior. With these initial states the dynamics will explore a large portion of the Hilbert space and this should facilitate the study of dynamical properties. We will now provide numeric evidence that such intuition can be used to unravel the logarithmic spread of entanglement in time.
\begin{figure}[h!]
\centering
\includegraphics[scale=0.33]{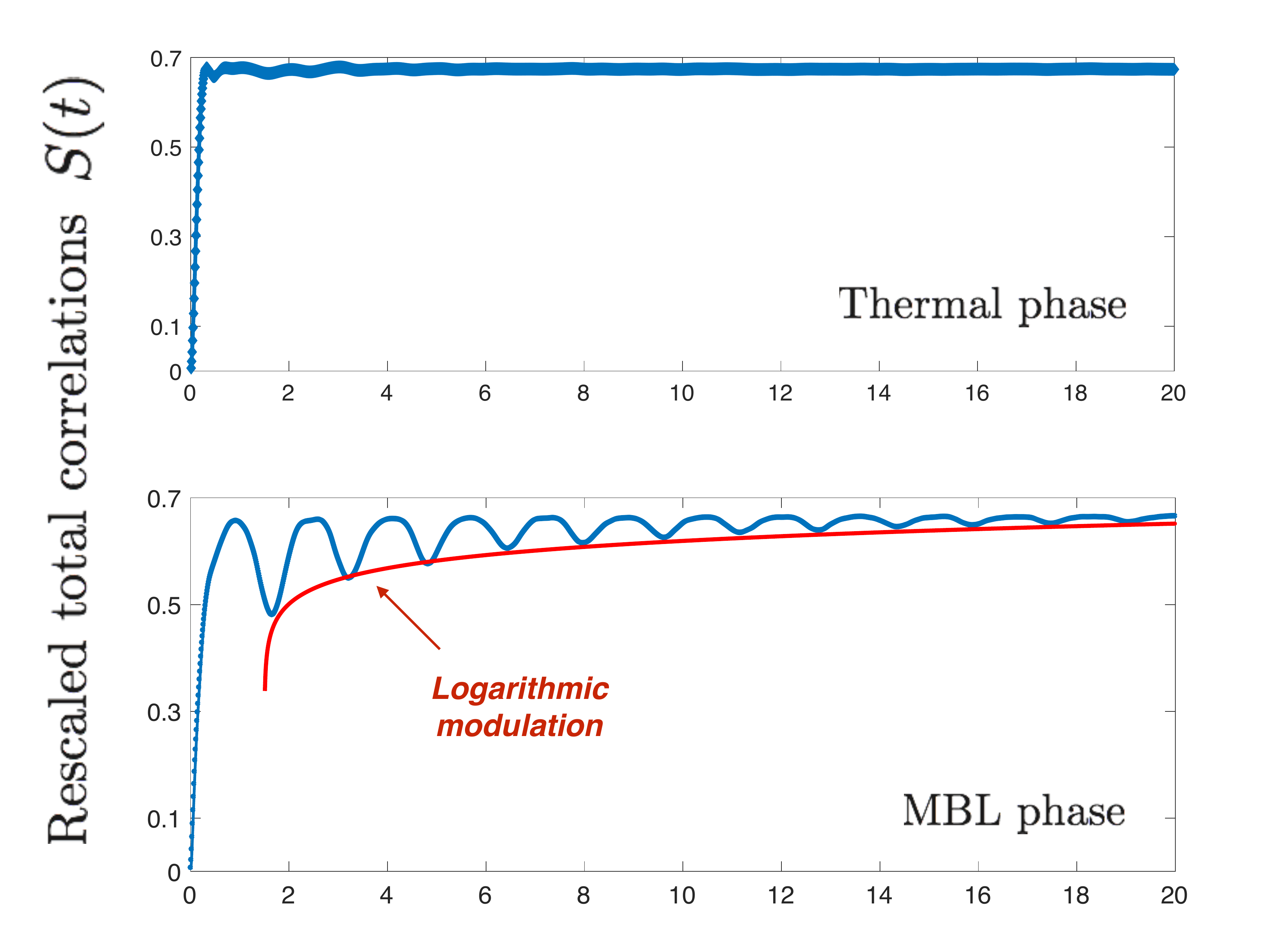}
\caption[Time-depedent profile of the rescaled total correlations]{Time-depedent profile of the rescaled total correlations $S(t)$ in the thermal and MBL phase for the disordered XXZ model with the following parameters. The length of the chain is $L=12$; the model is isotropic so $\Delta=1$ and the disorder strength in the thermal phase (upper figure) is $W=1$ while in the MBL phase (lower figure) is $W=10$. In the MBL phase we can clearly see fast oscillations which are modulated on a logarithmic time-scale. The red line going through the local minima is the fit for the logarithmic modulation.}\label{fig:LogTot}
\end{figure}
\paragraph{Parameters of the numerics.}
For our numeric investigation we focus on the disordered XXZ model of Eq.\ref{eq:XXZDis} at different sizes: $L=8,\ldots,14$. We averaged over $N_{\mathrm{dis}}=200$ disorder realizations. At $\Delta=1$ this model exhibits an ergodic phase for $W<W_c$ and a MBL phase at $W>W_c$ where $W_c \approx 3.7$. More information about the phase diagram can be found in \cite{Alet2017a}. As previously mentioned, since the disorder is along the $z$ direction, deep in the MBL phase the Q-LIOM will be almost diagonal in $\mathcal{B}_z$. For this reason we avoid choosing initial states which are part of $\mathcal{B}_z$ and we focus on states of two HUBs: $\mathcal{B}_x$ and $\mathcal{B}_y$. All plots show the data obtained by using the Neel state in the $x$ direction as initial state: $\ket{\psi_0}=\ket{\uparrow_x \downarrow_x\ldots }$. The qualitative behavior of different initial states, belonging to $\mathcal{B}_x$ or $\mathcal{B}_y$, is the same.
\begin{figure}[h!]
\centering
\includegraphics[scale=0.33]{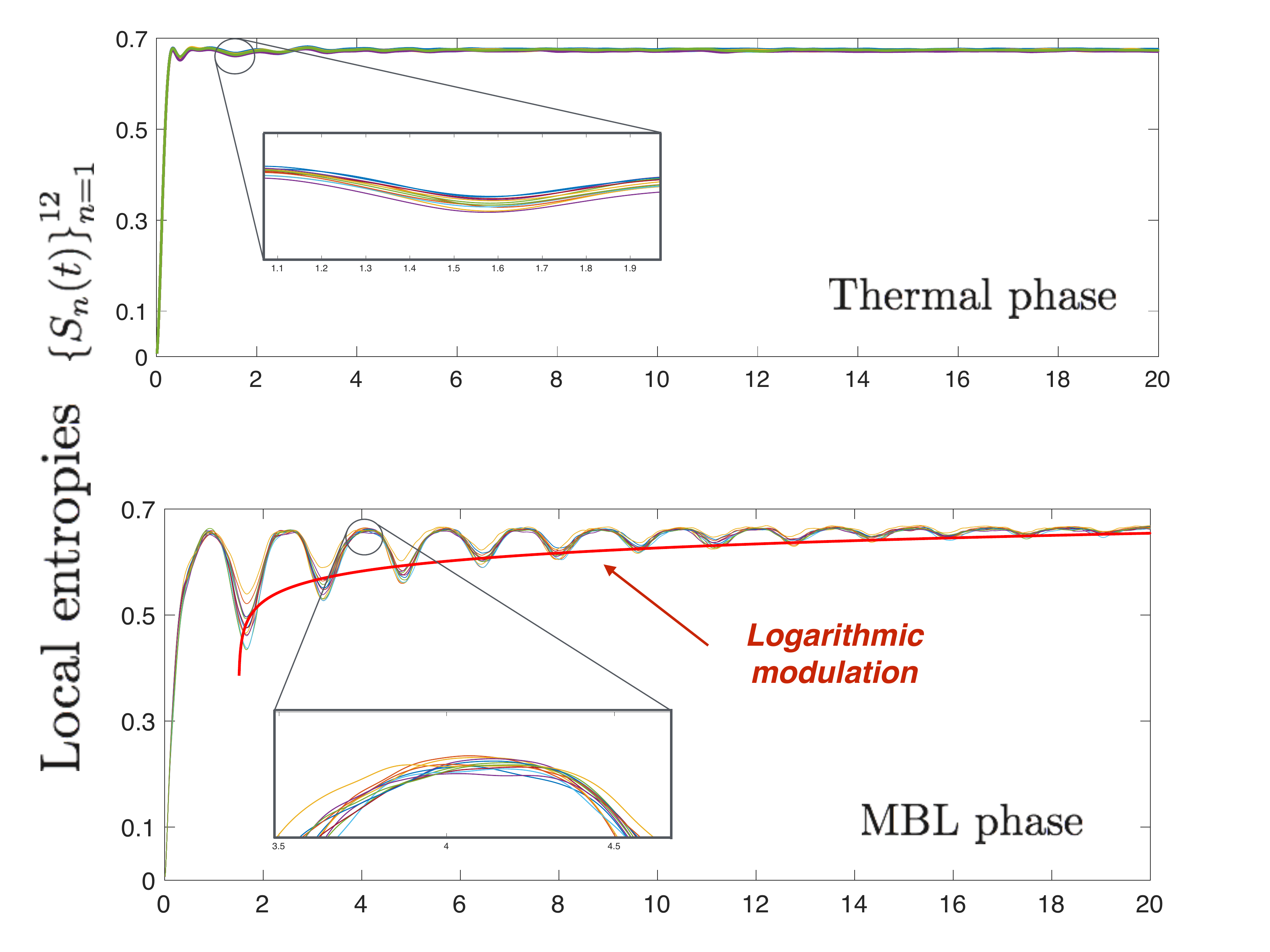}
\caption[Time-depedent profile of the local entropies]{Here we show the time-depedent profile of all the local entropies $\left\{S_n(t)\right\}_{n=1}^L$, in the thermal and MBL phase, for the disordered XXZ model with the following parameters. The length of the chain is $L=12$; the model is isotropic so $\Delta=1$ and the disorder strength in the thermal phase (upper figure) is $W=1$ while in the MBL phase (lower figure) is $W=10$. We see that, in the MBL phase, the behaviour of each $S_n(t)$ is very similar to their average. The red line going through the local minima is the fit for the logarithmic modulation.}\label{fig:LogEnt}
\end{figure}
In figure \ref{fig:LogTot} we compare the behaviour of $S(t)$ for $L=12$ in the thermal ($W=1$) and in the MBL ($W=10$) phase. In the thermal phase we clearly see a fast thermalization to the predictions of the maximum von Neumann entropy, plus small fluctuations around the equilibrium value. The behaviour in the MBL phase is more complex. We observe fast oscillations in time which are modulated on a logarithmic time-scale. This behaviour is also observed at the level of each local entropy $\left\{S_n(t)\right\}_{n=1}^L$, for all of them. In figure \ref{fig:LogEnt} we see that both the slow logarithmic modulation and the fast oscillations are synchronised, across the chain. This confirms the following theoretical picture: In a system which is in an MBL phase its elementary constituents, spin-$1/2$ in this case, share information on a logarithmic time-scale. This is a more fine-grained statement as it pertains the behaviour of all microscopic constituents. In this sense, our result agrees and go beyond the picture given by the half-chain entanglement entropy, confirming the theoretical picture summarised in subsection \ref{sec:logspread}. 

Now we propose the following interpretation of these results for out-of-equilibrium dynamics of the disordered XXZ model in the MBL phase. All our dynamical simulation started with a state which is as far away from an energy eigenstate as possible: the Neel state polarised along the $x$ direction. This is a product state which has no entanglement so all the local entropies starts at a null value. As time goes by it appears that the pure state of the system is oscillating between two classes of states. The first does not change in time and it is clearly a highly entangled one, which we could identify with the local maxima in the behaviour of $S(t)$. The second one is identified with the local minima of $S(t)$ and its nature changes. At first this is almost a tensor product state, having a small entropy. However, as time goes by its tensor-product structure is progressively lost. This happens on a logarithmic time-scale and it should be ascribed to the exponentially small tails of the Q-LIOMs.   

Thus, the observed numerical data can be explained by separating two relevant dynamical processes. We refer to the cartoon in Fig.\ref{fig:MBLConjecture} to visualize this behaviour. The fast dynamics is given by the oscillations of the local entropy between two classes of states: a high-entanglement one and a low-entanglement one. The slower dynamics, attributed to the existence of the Q-LIOMs, is given by the slow movement of the second class of states, whose entanglement increases on a logarithmic time-scale.
\begin{figure}[h!]
\centering
\includegraphics[scale=0.3]{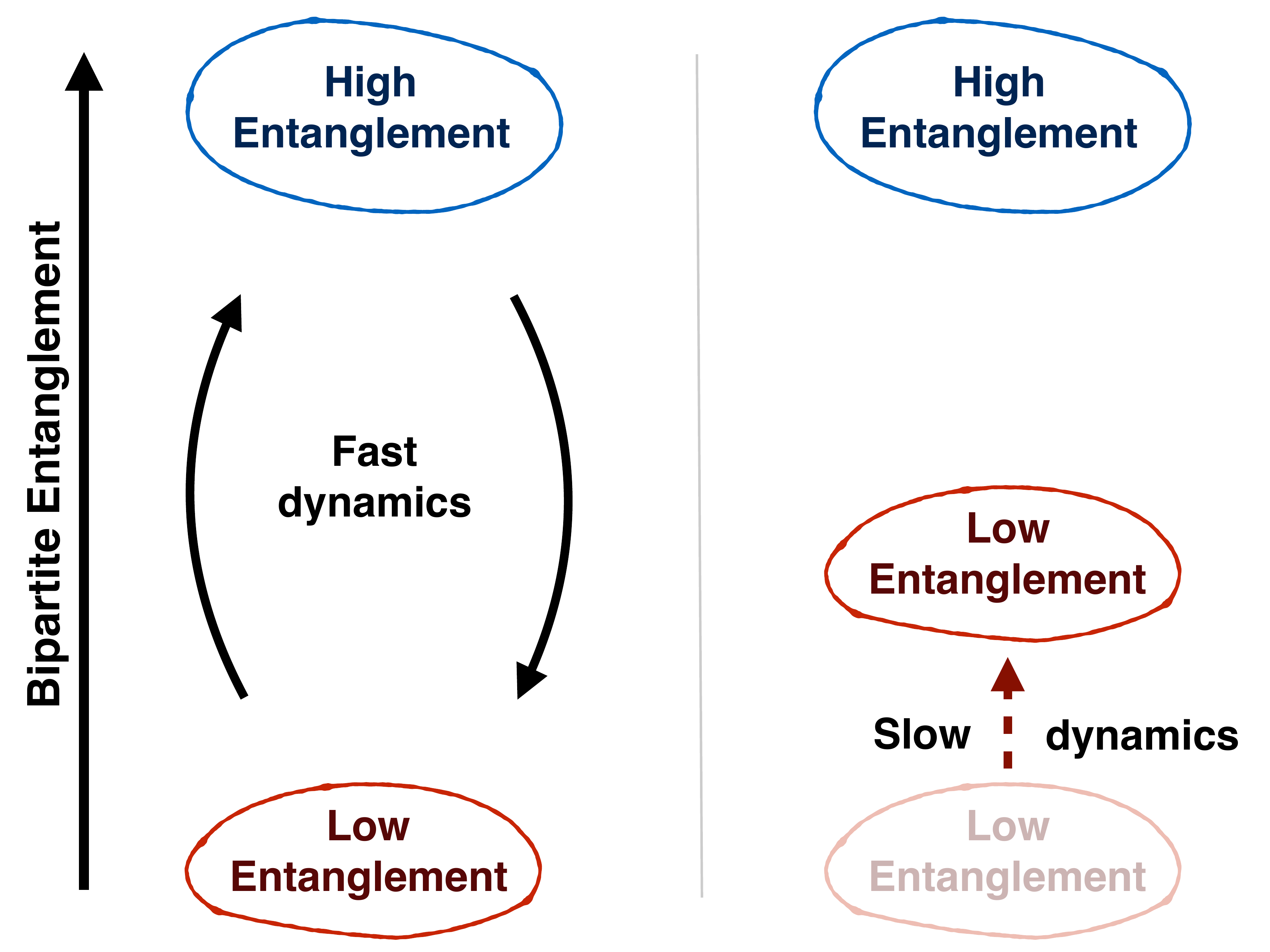}
\caption[Bipartite entanglement dynamics in the localised phase]{In this cartoon we show how to visualize the dynamics of bipartite entanglement in the disordered XXZ model in the MBL phase. The fast dynamics is given by regular oscillations between two regions of the Hilbert space which have high and low entanglement. The slow dynamics is given by a progressive increase of the entanglement in the red region, which approaches the blue region on a logarithmic time scale.}\label{fig:MBLConjecture}
\end{figure}

\subsection{Summary}
We exploited the intuition developed in the previous Chapters to propose a new way to unravel the logarithmic growth of entanglement in MBL systems. By studying the behaviour of the local entropies $\left\{ S_n(t)\right\}$ and of the rescaled total correlation $S(t)$ we provided further numerical evidence that the microscopic constituents of an MBL system share information on a logarithmic time-scale. Our conclusions are based on the numerical simulation of the XXZ model, in the MBL phase, and on the intuition given by the Hamiltonian Unbiased Basis and their dynamical relevance.

\section{Summary and Conclusion}

In Chapter \ref{Ch3} we proposed a new notion of thermal equilibrium, specific for a given observable rather than for the whole state of the system. The equilibrium picture that emerges is that we are dealing with an observable-wise generalization of statistical mechanics. In this framework, the notion of thermal observable plays the same role of the notion of thermal state in statistical mechanics. A peculiar subset of such observables is given by the Hamiltonian Unbiased Observables. On the one side, in this Chapter we test the framework against systems which are known to escape the description of statistical mechanics. We investigated the presence of thermal observables in MBL systems, providing evidence that the proposed notion of thermal equilibrium is able to go beyond the domain of applicability of statistical mechanics. Indeed, we are able to predict the emergence of thermal equilibrium in local observables, even though the reduced state of the system is not at thermal equilibrium. 

On the other side, it has been repeatedly argued that the most important dynamical feature of a system in the MBL phase is the logarithmic spread of entanglement. Despite the amount of theoretical effort aiming at proposing diverse quantities to witness such phenomenon, an experimental confirmation is still lacking. This is partly due to the notorious difficulty in measuring entanglement in a quantum system. Here we make significant progress in this direction by arguing that such behavior is encoded in the eigenvalues of all local density matrices. These are experimentally accessible with local quantum tomography. If it is possible to reconstruct all local density matrices one has also access to the total correlations: a quantity which, in a previous work, was shown to detect, at equilibrium, the transition from the ergodic to the MBL phase. Building our intuition on the notion of Hamiltonian Unbiased Basis we suggest a criterion for the choice of initial states which are both experimentally available and convenient for studying dynamical features of a system in the MBL phase. Eventually, we corroborated our arguments with numeric evidence, built on the study of a spin-$1/2$ chain in the disordered XXZ model. While the numerical data presented here is sufficient to argue for an exponential growth of local entropies, further studies to improve our understanding of this phenomenon are currently ongoing.

\addcontentsline{toc}{part}{Application to Quantum Gravity}
\part*{Application to Quantum Gravity}

\chapter{Macroscopic aspects of Spin Networks}\label{Ch6}


Quantum gravity is considered, by many, as the holy grail of theoretical physics. For several years people have investigated different approaches, each with their own limitations and advantages. Unfortunately, the quest for a coherent quantum theory of gravity, with a semiclassical regime in agreement with General Relativity, is still ongoing. Among all approaches, String Theory and Loop Quantum Gravity (LQG) are certainly the most advanced ones. Here we will focus on LQG and its underlying formalism: the spin networks. A full introduction to the topic of LQG is beyond the scope of the thesis. Rather, the purpose of this chapter is to give the reader the basic tools necessary to understand the original results presented in the next two chapters. For this reason, we decided to split the chapter in two pieces, which are conceptually different. In the first one, 
we address the question at the core of this investigation and present the point of view of the author. In the remaining part we will present the formalism, with a constructive exposition, and show in which sense the spin networks provide a well defined notion of \emph{quantum geometry}. 


\section{Macroscopic behaviour of spin networks}

Information Theory was originally developed as a theory of communication. Since its inception, an increasing number of physicists has used its tools and concepts to study the behaviour of physical systems. The interaction between physics and information theory has been especially enlightening in the context of quantum mechanics. Indeed, the inherently statistical nature of quantum theory makes Information Theory particularly suitable to investigate the behaviour of quantum systems. The large body of results developed in the last 50 years gave rise to Quantum Information Theory (QI): a well developed research field which, at its core, tries to describe how quantum systems share or store information and ``what they do with it''\cite{Lloyd1988}. For this reason, it is opinion of the author that QI provides a good set of ideas and techniques to study the behaviour of quantum systems, beyond the standard tools of analysis inherited from classical mechanics. Nowadays, ideas from QI, such as the entanglement entropy, are part of the standard tools of analysis in various fields: from condensed matter systems to the more fundamental topics such as understanding the emergence of thermal equilibrium in quantum mechanics. \\

In more recent years, techniques and tools from QI have played an increasingly central role also in quantum gravity. Such interplay has proved particularly insightful both in the context of the holographic duality in AdS/CFT \cite{Maldacena1998a,Ryu2006b,Raamsdonk2010,Swingle2012,Czech2012a}, as well as for the current background independent approaches to quantum gravity, including loop quantum gravity (LQG)  \cite{Thiemann2008,Rovelli}, the related spin-foam formulation \cite{Perez2013}, and group field theory \cite{Oriti2016}. Interestingly, many background-independent approaches today share a microscopic description of space-time geometry given in terms of discrete, pre-geometric degrees of freedom of combinatorial and algebraic nature, based on spin-network Hilbert spaces \cite{Rovelli1995a,Baez1996,Major1999}. In this context, the behaviour of quantum correlations provides a new tool to characterise the quantum texture of space-time in terms of the structure of microscopic correlations of the spin networks states. For example, several recent works have considered the possibility to use specific features of the short-range entanglement (area law, thermal behaviour) to select states which may eventually lead to smooth spacetime geometry classically \cite{Bianchi2014,Chirco2015,Chirco2014,Pranzetti}. These analyses usually focus on states with few degrees of freedom, leaving open the question of whether a \emph{statistical} characterisation may reveal new structural properties. Our goal here is to tackle the problem from a slightly different perspective. We are interested in understanding how some elementary degrees of freedom give rise to the smooth structure that we perceive as spacetime. To achieve this goal, we propose the use of the information-theoretical notion of \emph{quantum typicality} as a tool to investigate and characterise universal local features of quantum geometry, going beyond the physics of states with few degrees of freedom. We will now explain in more details the meaning of this statement, whose intuition is built on the mathematical concept of the ``concentration of measure'' phenomenon. The typicality approach to thermalization is an application of this concept and its main tool is Levy's Lemma\cite{Ledoux2001}, which we give in Appendix \ref{App:Levy}.\\

Canonical typicality states that almost every pure state of a large quantum system subject to an arbitrary constraint $\mathcal{R}$ is such that its reduced density matrix over a sufficiently small subsystem $\rho_S(\phi)=\Tr_{\overline{S}}\ket{\phi}\bra{\phi}$ is approximately in the canonical state $\Omega_S=\Tr_{\overline{S}}\mathcal{I}_{\mathcal{R}}$ defined as the partial trace, over the complement $\overline{S}$ of $S$, of the ``equiprobable state'' $\mathcal{I}_{\mathcal{R}}$ defined on the constrained space $\mathcal{H}_{\mathcal{R}}$. A full summary of this statement was given in Section \ref{sec:Ch2Typicality}. A different way of phrasing it is through its relation with the ``a priori equal probabilities'' principle, which leads to the maximally mixed state $\mathcal{E}_{\mathcal{R}}$: For almost every pure state of the universe, the reduced state of a small subsystem is approximately the same as if the universe were in the equiprobable state $\mathcal{I}_{\mathcal{R}}$. This justifies the use of the maximum entropy principle, at the global level, as long as we are interested in the behaviour of small subsystems. 

As previously showed (see Section \ref{sec:Ch2Typicality}), the concentration of measure phenomenon can be exploited to argue for the ubiquity of thermal behaviour in macroscopic systems, but its use goes beyond that. If we apply the same logic to the Hilbert space of a tentative theory of quantum gravity, this implies that, in the macroscopic regime, the statistical fluctuation around the typical state $\Omega_S$ of a small subsystem will be incredibly suppressed. In other words, the local properties of the theory under scrutiny will be heavily constrained by the presence of global constraints. We interpret this as proof that, at the macroscopic scale, there are certain local properties of quantum geometry which are ``universal''. Such set of properties is heralded by $\Omega_S$ and it can be unraveled by studying the explicit form of $\Omega_S$.

When we use a typicality argument, we are implicitly performing a statistical analysis over the Hilbert space, to understand what is the ``average local behaviour''. Because of the concentration of measure phenomenon, such local behaviour is extremely peaked around the properties of $\Omega_S$. It is therefore reasonable to expect, from the correct quantum theory of gravity, that the typical properties will be the ones which we observe at the macroscopic scale. Thus, the typicality approach offers a nice and rigorous way to connect the local properties of a quantum theory to their counterpart in the macroscopic regime. For this reason, it is opinion of the author that this logic can be exploited to provide arguments to support, or rule out, tentative theories of quantum gravity. Indeed, the local properties of the classical theory are of paramount relevance and   a good candidate for a quantum theory of gravity should be able to reproduce them, in the macroscopic regime. The underlying intuition is that \emph{a smooth geometry is an emergent feature which results from the structure of correlations at the quantum level.} The concentration of measure phenomenon, and the resulting typicality arguments, can then be used to unravel the connection between ``microscopic structure'' and  ``macroscopic emergence''.

To summarise, the typicality arguments provide a rigorous way to connect the local properties of a quantum system to their macroscopic counterpart. We use this technique to investigate if the spin network formalism  can, at the macroscopic level, reproduce the ``smooth geometry'' behaviour which we expect from General Relativity.




\section{Spin Network states of quantum geometry}

Now we briefly introduce the formalism of spin-networks and their geometric interpretation, developed in LQG. This is, by no means, a complete introduction to spin networks or to
LQG. The sole purpose of this section is to give the reader the technical tools necessary to understand the original results presented in the next chapters. We will present the formalism, with a constructive exposition, and show in which sense the spin networks provide a well defined notion of \emph{quantum geometry}. 

\subsection{Spin Networks}

In several background-independent approaches to quantum gravity, the spin network states provide a \emph{kinematical} description of quantum geometry, in terms of superpositions of graphs $\Gamma$ labelled by group or Lie algebra elements representing holonomies of the gravitational connection and their conjugate triad \cite{Thiemann2008,Rovelli}. These states are constructed as follows (for a thorough introduction to spin-networks we refer to \cite{Rovelli1995a,Baez1996,Major1999}). A graph $\Gamma$ is a set of vertices (or nodes) which are connected by links (or edges). Call $V$ the number of vertices and $N$ the number of links. To each link $l \in \Gamma$ one associates an $SU(2)$ irreducible representation (irrep) labelled by a half-integer $j_l \in \mathbb{N}/2$ called spin. The representation (Hilbert) space is denoted $V^{j_l}$ and has dimension $d_{j_l} = 2 j_l +1$. To each link $l$ one associates the tensor product of the space of the irreducible representation $V_{j_l}$ times its dual $\overline{V}_{j_l}$. In other words, for each link $l$ with spin $j_l$ there is an Hilbert space $V_{j_l} \otimes \overline{V}_{j_{l}}$. A spin network is therefore a graph $\Gamma$, together with the assignment of a set of $N$ irreducible representations of $SU(2)$ labeled by their unique quantum number $j_l$, with $l=1,\ldots,N$. Its Hilbert space is simply the sum over all possible representations
\begin{align}
&\mathcal{H}_U=\bigoplus_{j_l \in \mathbb{N}/2} \bigotimes_{l=1}^N V_{j_l} \otimes \overline{V}_{j_l}.
\end{align}
\begin{figure}[h!]
\centering
\includegraphics[scale=0.2]{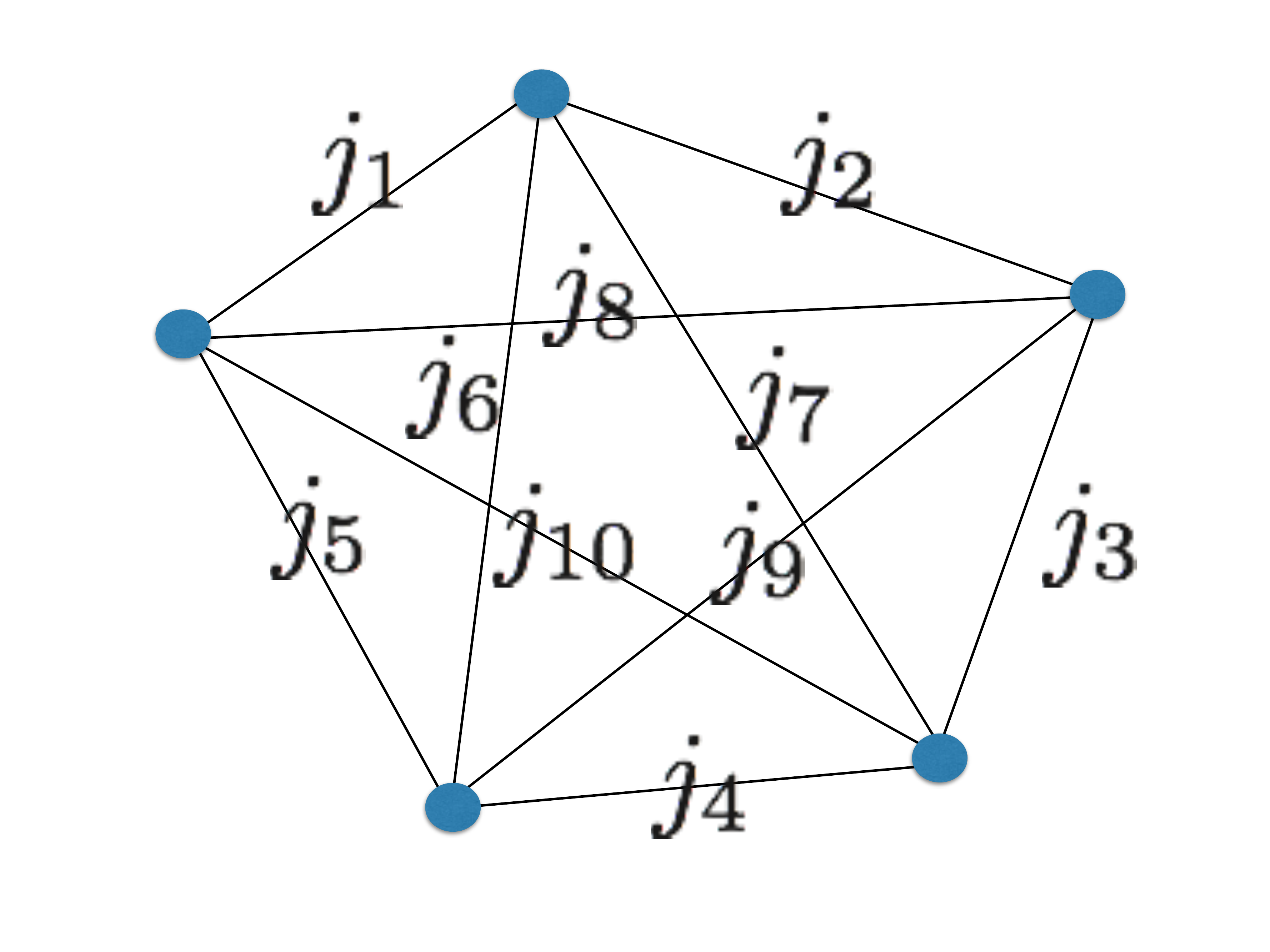}
\caption[Example of a spin-network for a $4$-simplex graph]{Example of a spin network with $5$ vertices connected by $10$ links labeled with ten values $\left\{ j_l\right\}_{l=1}^{10}$ of irreducible representations of $SU(2)$. The label $j_l$ stands for the fact that on the link $l$ we have an Hilbert space $V_{j_l} \otimes \overline{V}_{j_l}$.}\label{fig:SN}
\end{figure}
An example is given in Fig.(\ref{fig:SN}). These are not yet the spin network states that describe a quantum geometry. In order to get there we need to impose invariance under SU(2). The reason is that, at the classical level, General Relativity has a local gauge group which is the Lorenz group. If we perform a partial gauge fixing of this group, fixing the direction of the boosts, we are left with an SU(2) gauge invariance which accounts for invariance under rotation. For this reason, to each vertex $v$ of the graph one attaches an intertwiner $\mathcal{I}_v$, which is an $SU(2)$-invariant map between the representation spaces $V^{j_l}$ associated to all the links $l$ meeting at the vertex $v$,
\begin{align}
\mathcal{I}_v: \bigotimes_{l \,\,\text{ingoing}}  V^{j_l}\to \bigotimes_{l\,\, \text{outgoing}} \overline{V}^{j_l}\,\,.
\end{align}
One can alternatively consider $\mathcal{I}_v$ as a map from $\otimes_{l \in v} V^{j_l}\to \mathbb{C}\simeq V^0$ and call the intertwiner an invariant tensor (or a singlet state) between the representations attached to all the edges linked to the considered vertex. Once the $j_l$'s  are fixed, the intertwiners at the vertex $v$ form a Hilbert space $\mathcal{H}_v$
\begin{align}
\mathcal{H}_v \equiv  \text{Inv}_{SU(2)} [ \bigotimes_{l \in v}V^{j_l}] \,\,\, . \label{eq:Inv}
\end{align}
To understand better what is happening, it is useful to think about a large spin network and isolate a single vertex $v$, with the $N_v$ links around (see Figure \ref{fig:SN1N}). If we decompose the tensor product of all these $N_v$ representations into a sum over irreducible representations we have:
\begin{align}
&\bigotimes_{l=1}^{N_v} V_{j_l} = \left( \underbrace{V_0 \oplus \ldots \oplus V_0}_{d_0 \, \mathrm{times}} \right) \bigoplus \left( \underbrace{V_1 \oplus \ldots \oplus V_1}_{d_1 \, \mathrm{times}} \right) \bigoplus \ldots = \mathcal{K}_0 \oplus \mathcal{K}_1 \oplus \ldots
\end{align} 
$\mathcal{K}_0$ is the vector subspace of the whole tensor product which is invariant under $SU(2)$ and it is precisely the Intertwiner space defined in Eq.(\ref{eq:Inv}):
\begin{align}
&\mathcal{K}_0 = \mathrm{Inv}_{SU(2)} \left[ \bigotimes_{l \in v} V_{j_l} \right] 
\end{align}
Its dimension $d_0=d_0(\left\{ j_l\right\})$ depends on the precise value of the representations $j_l$ living on the links around the vertex $v$:
\begin{align}
&d_0(\left\{ j_l\right\}) = \int_{\mathrm{SU(2)}} dg \, \left(\prod_{l \in v} \chi_l(g) \right) = \frac{2}{\pi} \int_{0}^\pi \sin^2 \theta d \theta \, \left(\prod_{l \in v}\chi_{j_l}(\theta) \right)  \,\, ,
\end{align}
where  $\chi_{j_l}(g)$ is the character of SU(2) in the representation $j_l$. It's definition is as the trace of the group element $g$ when it is written in the $j_l$-th representation. It depends only on the rotation angle $\theta$:
\begin{align}
&\chi_j(g) = \chi_{j}(\theta) = \frac{\sin (2j+1)\theta}{\sin \theta}
\end{align}
The presence of a non-trivial intertwiner space accounts for all the possible ways in which we can contract the $\left\{j_l\right\}_{l\in v}$ irreducible representations of SU(2) (around the vertex $v$) to get an invariant tensor. 
\begin{figure}[h!]
\centering
\includegraphics[scale=0.2]{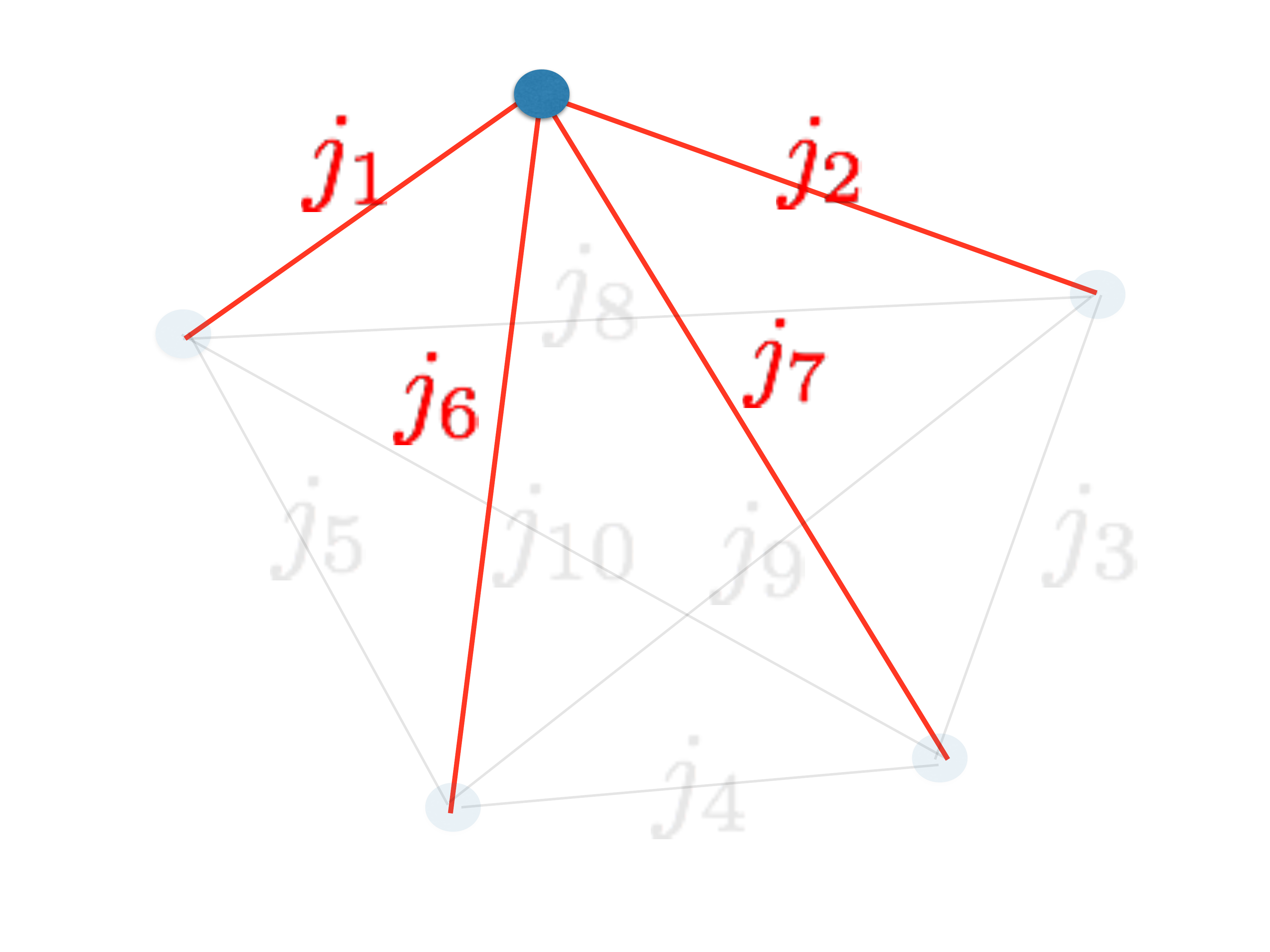}
\caption[Example of a single node in a spin-network]{Here we take the previous example of a spin network and we isolate a single node}\label{fig:SN1N}
\end{figure}
Eventually, a gauge-invariant spin network state $|\Gamma, \{j_e\},\{\mathcal{I}_v\}\rangle$ is defined as the assignment of representation labels $j_e$ to each edge and the choice of a vector $|\{\mathcal{I}_v\}\rangle \in \otimes_v\mathcal{H}_v$ for the vertices. The spin network state defines a wave function on the space of discrete connections $SU(2)^E/SU(2)^V$ ,
\begin{align}
\phi_{\{j_e\},\{\mathcal{I}_v\}} [g_e ] = \langle h_e |  \mathcal{I}_v \rangle \equiv  \text{tr} \bigotimes_e D^{j_e}(h_e) \otimes \bigotimes_v \mathcal{I}_v
\end{align}
where we contract the intertwiners $\mathcal{I}_v$ with the (Wigner) representation matrices of the group elements $g_e$ in the chosen representations $j_e$. Therefore, upon choosing a basis of intertwiners for every assignment of representations ${j_l}$, the spin network states are a basis of the space of wave functions associated to the graph $\Gamma$, 
\begin{align}
\mathcal{H}_{\Gamma} = \bigoplus_{\{j_e\}} \bigotimes_v \mathcal{H}_v = L_2[SU(2)^E/SU(2)^V] \,\, .
\end{align} 
Such discrete and algebraic objects provide a description of the fundamental excitations of quantum spacetime. We will now discuss the underlying notion of quantum geometry of the SU(2)-invariant spin network states. 

\subsection{Quantum Geometry}\label{Sec:QGeos}

The geometrical interpretation of spin network states was mostly developed in \cite{Bianchi,Barbieri1998,Baez2000,Freidel2011,Livine2014}. Here we summarize the resulting picture. Choose an arbitrary graph $\Gamma$ and start from the space of the states which are non-gauge invariant: $\mathcal{H}_\Gamma = \bigoplus_{j_l} \bigotimes_{l} V_{j_l} \otimes \overline{V}_{j_l}$. The direct sum over different irreps accounts for different choices of the spins over the links and states belonging to different choices are necessarily orthogonal. For this reason we choose a labelling $(j_1,\ldots,j_N)$ and everything we say will hold for arbitrary choices of the labeling. We now focus on the Hilbert space around an arbitrary node $v$: $\otimes_{l \in v} V_{j_l}$. Given that these are angular momentum states, to each link $l$ we can associate $3$ real numbers describing the angular momentum along the $3$ directions $x,y,z$: $\vec{J}_l$. Their algebra satisfies SU(2) commutation relations $[J_l^a,J_l^b]= \, i \epsilon_{abc}J_l^c$ and their matrix representation is given by the $j_l$-th irreducible representation of SU(2). The spin network states $\ket{j_l,m_l}$ provide a full basis as they diagonalize a complete set of commuting observables $(\vec{J}_l^2, J_l^z)$, for each link:

\begin{align}
&J_l^2 = \vec{J}_l \cdot \vec{J}_l \ket{j_l,m_l} = j_l(j_l+1) \ket{j_l,m_l} && J_l^z  \ket{j_l,m_l} = m_l \ket{j_l,m_l}
\end{align}

This means that, at each node, we have $N_v$ vectors $\vec{J}_l$ of fixed length $j_l(j_l+1)$ but with arbitrary orientation. See Figure \ref{fig:vectors} for an example of a 4-valent vertex.
\begin{figure}[h!]
\centering
\includegraphics[scale=0.3]{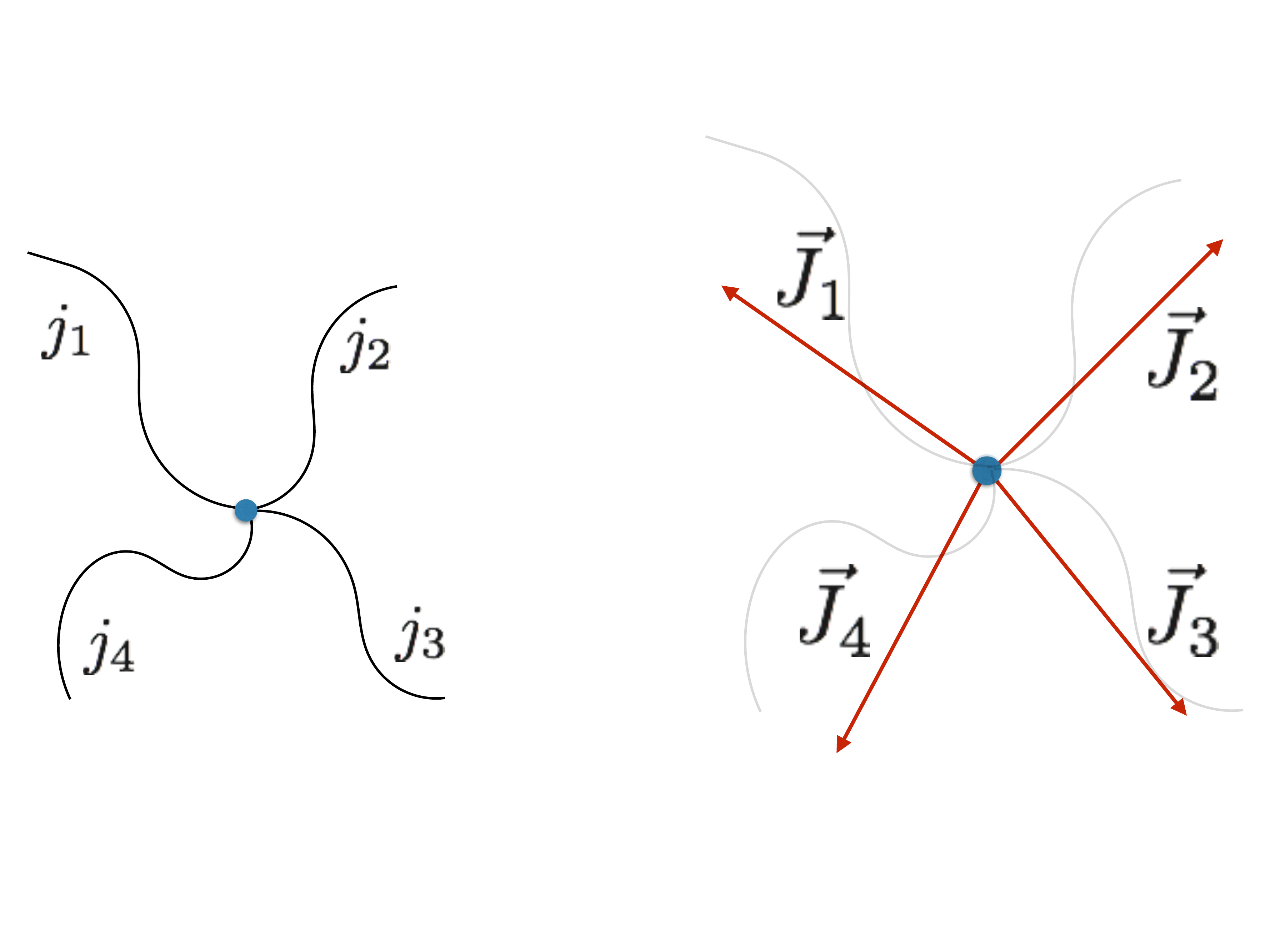}
\caption[Angular momentum vectors in a spin network]{In a four-valent vertex $v$ we have four angular momentum vectors $\left\{\vec{J}_l\right\}_{l=1}^4$ which can be defined around $v$. The non-gauge-invariant Hilbert space of the states spans all the possible orientations of these vectors.}\label{fig:vectors}
\end{figure}
The full geometric interpretation is available on the gauge-invariant Hilbert space $\mathcal{K}_0$. Indeed, the constraint imposing gauge-invariance is equivalent to the request that the sum of the vectors $\vec{J}_l$ around a node must be zero \cite{Bianchi,Barbieri1998,Baez2000,Freidel2011,Livine2014}. Hence, for a given vertex $v$, the allowed configurations of the vectors $\vec{J}_l$ in the gauge-invariant Hilbert space are the ones where $\sum_{l \in v} \vec{J}_l  = 0$.

Thanks to a theorem by Minkowski \cite{Bianchi}, we know that for an arbitrary set of $3D$ vectors $\left\{\vec{J}_l\right\}_{l=1}^{N_v}$ such that $\sum_{l=1}^{N_v} \vec{J}_l=0$ there is a $3D$ polyhedron where the vectors $\vec{J}_l$ are the surface vectors describing orientation and area of the faces of the polyhedron. This means that an arbitrary $N_v-$valent node is dual to a $3D$ polyhedron with $N_v$ faces. This is not merely an analogy. It can be shown that the intertwiner Hilbert space of an $N_v-$valent vertex is the quantization of the phase space of the shapes of a polyhedron with $N_v$ faces, at fixed areas\cite{Bianchi}. For this reason, given a cellular decomposition of a three-dimensional manifold, a spin-network graph with a node in each cell and edges connecting nodes in neighbouring cells is said to be dual to this cellular decomposition. Links carry ``area excitations'' while vertices carry ``volume excitations''. Each link is dual to a surface patch and the area of such patch depends on the Casimir of the representation $j_e$. Analogously, vertices of a spin network can be dually thought of as chunks of volume. See Figure \ref{fig:duality} for an example of a four-valent vertex.
\begin{align}
&\sum_{l\in v} \vec{J}_l = 0 \,\, .
\end{align}
\begin{figure}[h!]
\centering
\includegraphics[scale=0.3]{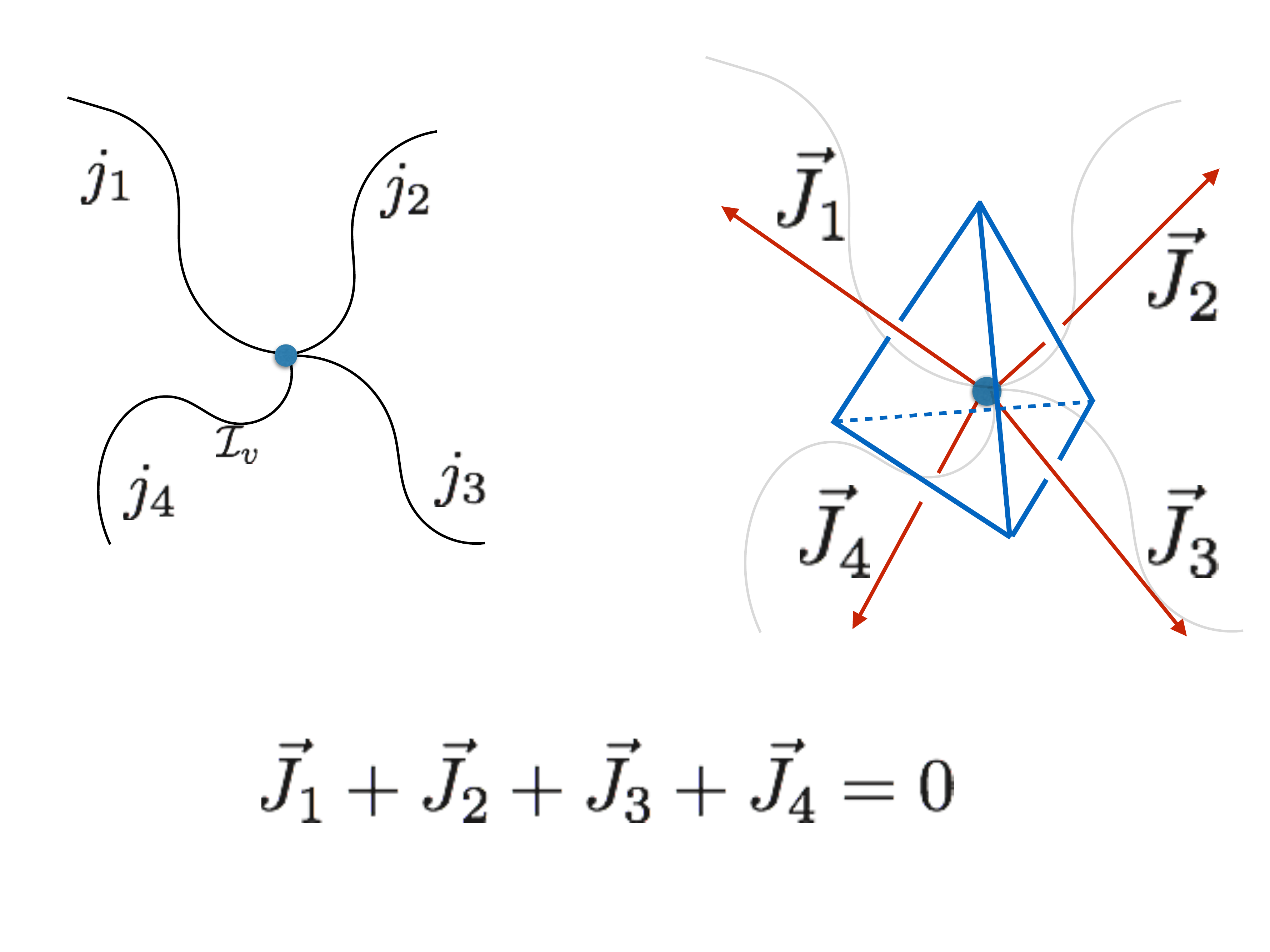}
\caption[Geometric interpretation of a node in a spin network]{Example of a four-valent node dual to a tetrahedron, describing the fundamental cell of a triangulated 3d space. The edges of the dual graph are labelled with spins $\{j_i\}$ which give the length of the area vectors $\vec{J}_i$. At the level of the area vectors, the gauge-invariance conditions is called ``closure constraint'' as the sum of all vectors $\vec{J}_i$ around a node is zero.}\label{fig:duality}
\end{figure}
This leads to the geometrical interpretation of spin network states as a collection of adjacent polyhedra and provides a well-defined notion of quantum geometry in $3D$ space. For example in Fig.(\ref{fig:4SimpGeos})
we draw the dual geometric picture of a spin network with $5$ vertices and $10$ links. For a generic spin network the dual geometric picture with be similar to the one in Fig.(\ref{fig:collection}). For all the details about the geometric interpretation of spin-networks we send the interested reader to the relevant literature\cite{Bianchi, Barbieri1998, Baez2000, Freidel2011, Livine2014}. We will now show how quanta of area and of volumes can be defined in this formalism.
\begin{figure}[h!]
\centering
\includegraphics[scale=0.2]{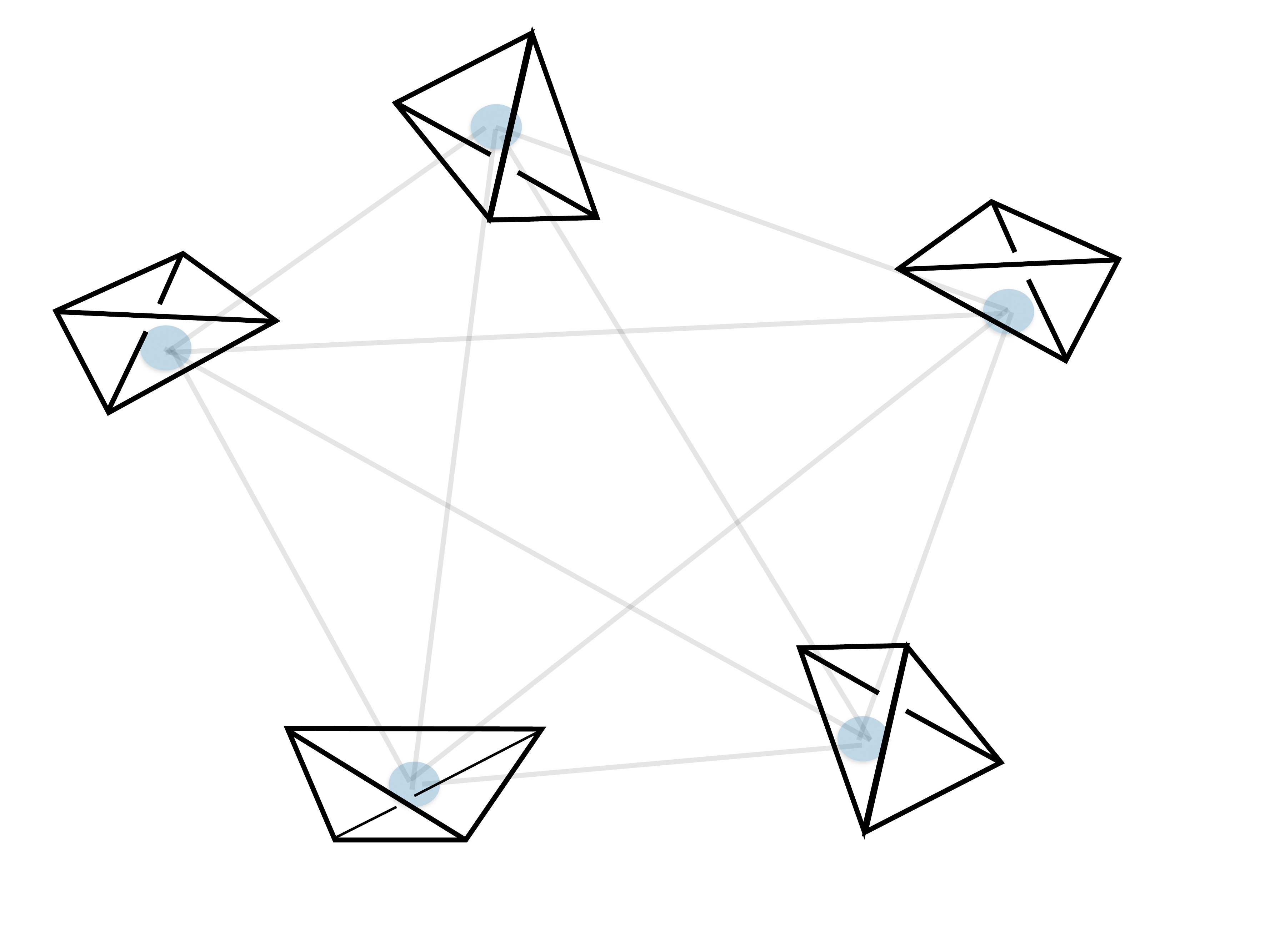}
\caption[Spin networks as collections of polyhedra]{Graph of a $4-$simplex which is geometrically interpreted as made by five tetrahedra.}\label{fig:4SimpGeos}
\end{figure}

\subsection{Area, Volume and fuzzy geometry}

Here we present some fundamental definitions and calculations which clarify the geometric interpretation of spin networks and their quantum nature. In particular, we will focus on the area and volume operators, for the simplest case of a quantum tetrahedron. The area operator $A_l$ on the link $l$ is diagonal in the spin network basis and its eigenvalues are simply $\sqrt{j_l(j_l+1)}$:
\begin{align}
&A_l \ket{j_l,m_l}=\gamma L_p^2 \sqrt{\vec{J}_l \cdot \vec{J}_l}  \ket{j_l,m_l} = \gamma L_p^2 \sqrt{j_l(j_l+1)} \ket{j_l,m_l} \,\,\, ,
\end{align}
where $L_p$ is the Planck length and $\gamma$ is a real parameter called Immirzi parameter. The volume operator is a bit more involved, but the case of a four-valent node can be addressed easily \cite{Bianchi}. Indeed from standard $3D$ geometry it can be shown that volume of a tetrahedron with area vectors $\left\{ \vec{J}_l\right\}_{l=1}^4$ is 
\begin{align}
&V = L_p^3 \gamma^3 \frac{\sqrt{2}}{3} \sqrt{ \left\vert \vec{J}_1 \cdot \left( \vec{J}_2 \times \vec{J}_3 \right) \right\vert}
\end{align} 
\begin{figure}
\centering
\includegraphics[scale=0.3]{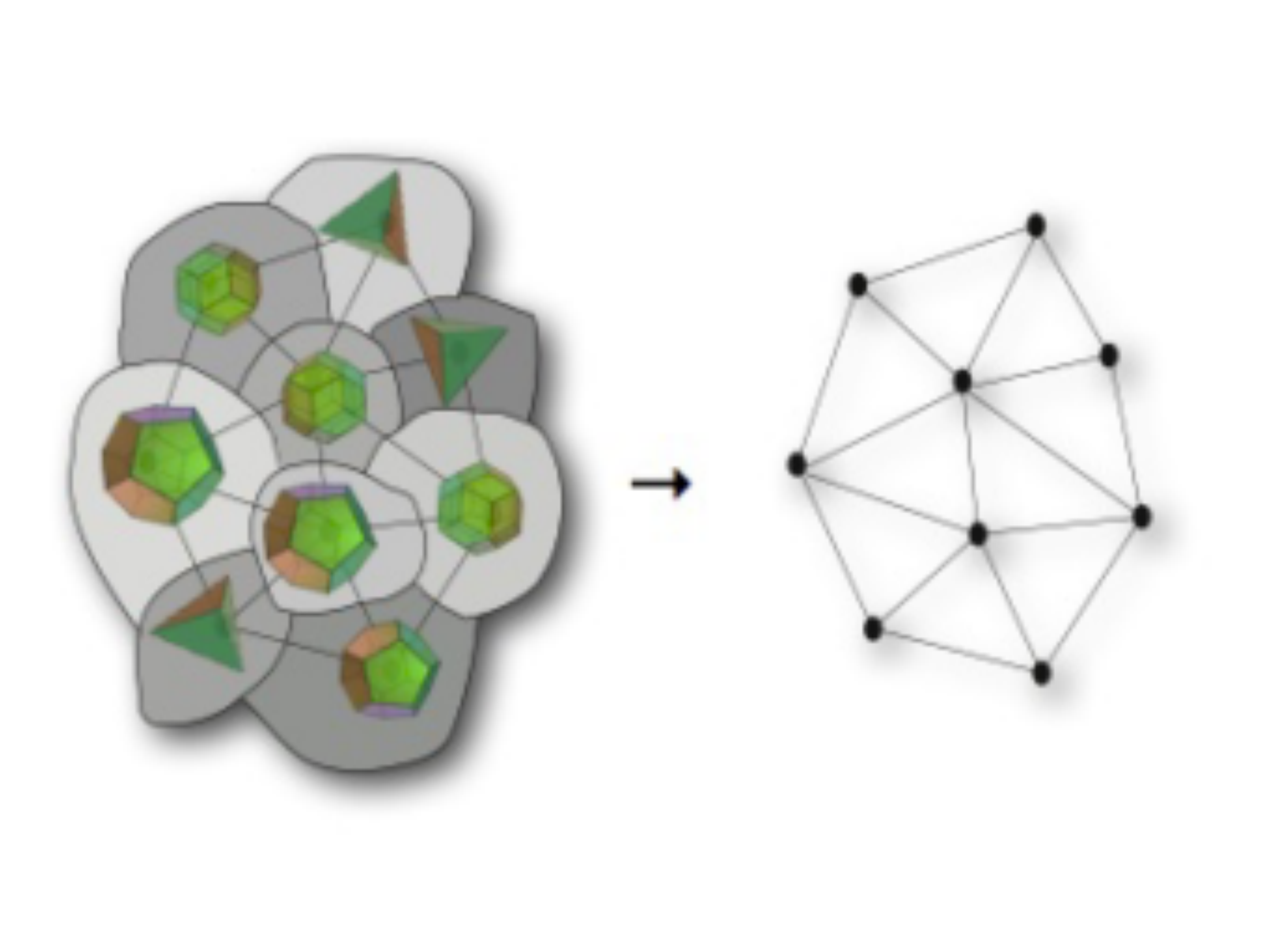}
\caption[Spin networks and their dual discrete geometry]{Geometric interpretation of a spin network as a collection of adjacent polyhedra.}\label{fig:collection}
\end{figure}
The four areas and the volume form a complete set of commuting observables and their basis states are the spin network states. But, is this a real tetrahedron? The answer is obviously no. The most important way to see this is that at the classical level the geometry of a tetrahedron is fixed when we give six numbers: the lengths of each side. However, at the quantum level we have only five numbers: the four areas and the volume. The situation is analogous to the standard case of angular momentum which, classically is defined by three numbers while at the quantum level is defined by two of them: $L^2$ and $L_z$. For this reason, the geometry of the quantum tetrahedron, in general will not be sharp, in the same way in which a quantum rotor does not have a sharp angular momentum $\vec{L}$. The fundamental reason is that the different components of the angular momentum do not commute. They satisfy the SU(2) commutation relation and for this reason they all define basis which are Mutually Unbiased. Because of that, if we are in an eigenstate of the angular momentum along the $z$ direction the components along $x$ and $y$ will be maximally spread. If $L_z$ is sharp, $L_x$ and $L_y$ are necessarily fuzzy. If we now apply the same logic to the quantum tetrahedron we see that the geometry can never be completely sharp. We will always have residual fuzziness, due to the fact that components along different directions do not commute. 

To summarize, there are two aspects where we can see that quantum geometry is fundamentally different from Riemannian geometry. First, area and volume can only take on discrete eigenvalues; second, geometry at the Planck scale is always fuzzy. This ends our introduction to spin networks and their quantum geometry. In the next chapter we will investigate the macroscopic behaviour of spin networks by means of the concentration of measure phenomenon.

\chapter{Typicality in spin-networks}\label{Ch7}


In this chapter we extend the so-called typicality approach, originally formulated in statistical mechanics contexts, to $SU(2)$ invariant spin network states. Our results do not depend on the physical interpretation of the spin-network. However, they are mainly motivated by the fact that spin-network states can describe states of quantum geometry, providing a gauge-invariant basis for the kinematical Hilbert space of several background independent approaches to quantum gravity. The first result is, by itself, the existence of a regime in which we show the emergence of a typical state. We interpret this as proof that, in that regime there are certain (local) properties of quantum geometry which are ``universal''. Such set of properties is heralded by the typical state, of which we give the explicit form. This is our second result. In the end, we study some interesting properties of the typical state, proving that the area-law for the entropy of a surface must be satisfied at the local level, up to logarithmic corrections which we are able to bound. This chapter is based on the work done by the author, in collaboration with Dr. G. Chirco, published in Ref. \cite{Anza2016}.

\section{Introduction}

In quantum statistical mechanics, \emph{canonical typicality} states that almost every pure state of a large quantum mechanical system, subject to the fixed energy constraint, is such that its reduced density matrix over a sufficiently small subsystem is approximately in the {\it canonical} state described by a thermal distribution \`{a} la Gibbs \cite{Hamma2018,Pranzetti,Dymarsky2018,Tasaki1998}. Such a statement goes beyond the thermal behaviour. For a generic closed system in a quantum pure state, subject to some \emph{global constraint}, the resulting canonical description will not be thermal, but generally defined in relation to the constraint considered \cite{Popescu2006,Linden}. Again, in this case, some specific properties of the system emerge at the local level, regardless of the nature of the global state. These properties depend on the physics encoded in the choice of the global constraints.
Within this generalised framework, we exploit the notion of \emph{typicality} to study whether and how  ``universal''  statistical features of the local correlation structure of a spin-network state emerge in connection with the  choice of the global constraint. We focus our analysis on the space of the $N$-valent $SU(2)$ invariant intertwiners, which are the building blocks of the spin network states. In LQG, such intertwiners can be thought of dually as a region of $3d$ space with an $S^2$  boundary (see Section \ref{Sec:QGeos}). We reproduce the typicality statement in the full space of $N$-valent intertwiners with fixed total area and we investigate the statistical behaviour of the canonical reduced state, dual to a small patch of the $S^2$  boundary, in the large $N$ limit. Eventually, we study the entropy of such a reduced state and its behaviour in different thermodynamic regimes.


\section{Intertwiner Typicality}  \label{tysn}


In the following we will consider the ``typicality'' argument, summarised in Section \ref{sec:Ch2Typicality} for the spin network states. The interest in testing the notion of typicality in quantum gravity resides in the kinematic nature of the statement, which is a fundamental feature to study the possibility of a thermal characterisation of reduced states of quantum geometry, regardless of any Hamiltonian evolution in time. In the following, we will focus on a fundamental building block of a spin network graph, the Hilbert space of a single intertwiner with $N$ legs. For this reason we consider a large quantum system, given by a collection of $N$ edges, represented by $N$ independent edges states. The Hilbert space of the system is the direct sum over $\{	j_i\}$'s of the  tensor product of $N$ irreducible representations $V^{j_i}$, 
\begin{align}\label{space}
\mathcal{H}= \bigoplus_{\{j_i\}} \bigotimes_{i=1}^N V^{j_i}.
\end{align} 
This set of independent edges plays the role of the ``universe''. 
Notice that, despite its extreme simplicity, this system has a huge Hilbert space. The single representation space $V^j$ has finite dimension $d_{j_i} = 2 j_i +1$. However, $d_{j_i}$ is summed over all $j_i \in \frac{\mathbb{N}}{2}$.  Therefore each Wilson line state (edge) lives in an infinite dimensional Hilbert space\footnote{An important detail is how we deal with spin-0 representations. In LQG these are avoided introducing  cylindrical consistency which requires that such links are equivalent to non-existent links. We do not require cylindrical consistency, hence spin-0 representations are allowed.}. In the following, we will always consider a cut-off in the value of the $SU(2)$ representation labelling the edge\footnote{Another way to introduce a cut-off in the representations, which has already been explored in literature \cite{Dupuis, Buffenoir1999,Delcamp2003,Fairbairn2012a,Han,Kassel}, is to consider the so-called q-deformation of $SU(2)$. This is usually done in LQG to include a cosmological constant.},

\begin{align}
\bigoplus_{\{j_i\}}\to \bigoplus_{\{j_i\}}^{j_i \le {J}_{max}}, \quad \text{with} \quad {J}_{max}\gg1.
\end{align}
This will allow us to deal with a very large but finite-dimensional space. Now, we want to split the universe in a ``small system plus environment'' fashion. We do so simply by defining two subsets of edges $E$ and $S$, with $\{1,\cdots , k \} \in {S}$  and $\{ k+1,\cdots , N \} \in{E}$, such that $k \ll N$. Consequently, we can write the Hilbert space of the universe as the tensor product $\mathcal{H}=\mathcal{H}_E \otimes \mathcal{H}_S$. We would like to stress here that, for our result to hold the $k$ links of the system do not need to be adjacent. Despite that, we are interested in the local properties of an intertwiner, therefore we will always think about these links as adjacent and forming a simply connected 2D patch.\\

The next step toward typicality consists in defining the constraint which restricts the allowed states of the system and environment to a subspace of the total Hilbert space.

\subsection{Definition of the constraint}

Two main ingredients are necessary to the definition of the constraint. The first one is the $SU(2)$ gauge invariance. This choice reduces the universe Hilbert space to the collection of the $SU(2)$-{invariant} linear spaces
\begin{align}\nonumber
\mathcal{H}_{N}=\bigoplus_{\{j_i\}}\text{Inv}_{SU(2)} [\bigotimes_{i=1}^N V^{j_i}],
\end{align} 
spanned by $N$-valent intertwiner states. 
Invariance under $SU(2)$ is the the first ingredient defining our subsystem constrained space.  
\begin{figure}[h]
    \centering
\includegraphics[width=2.5 in]{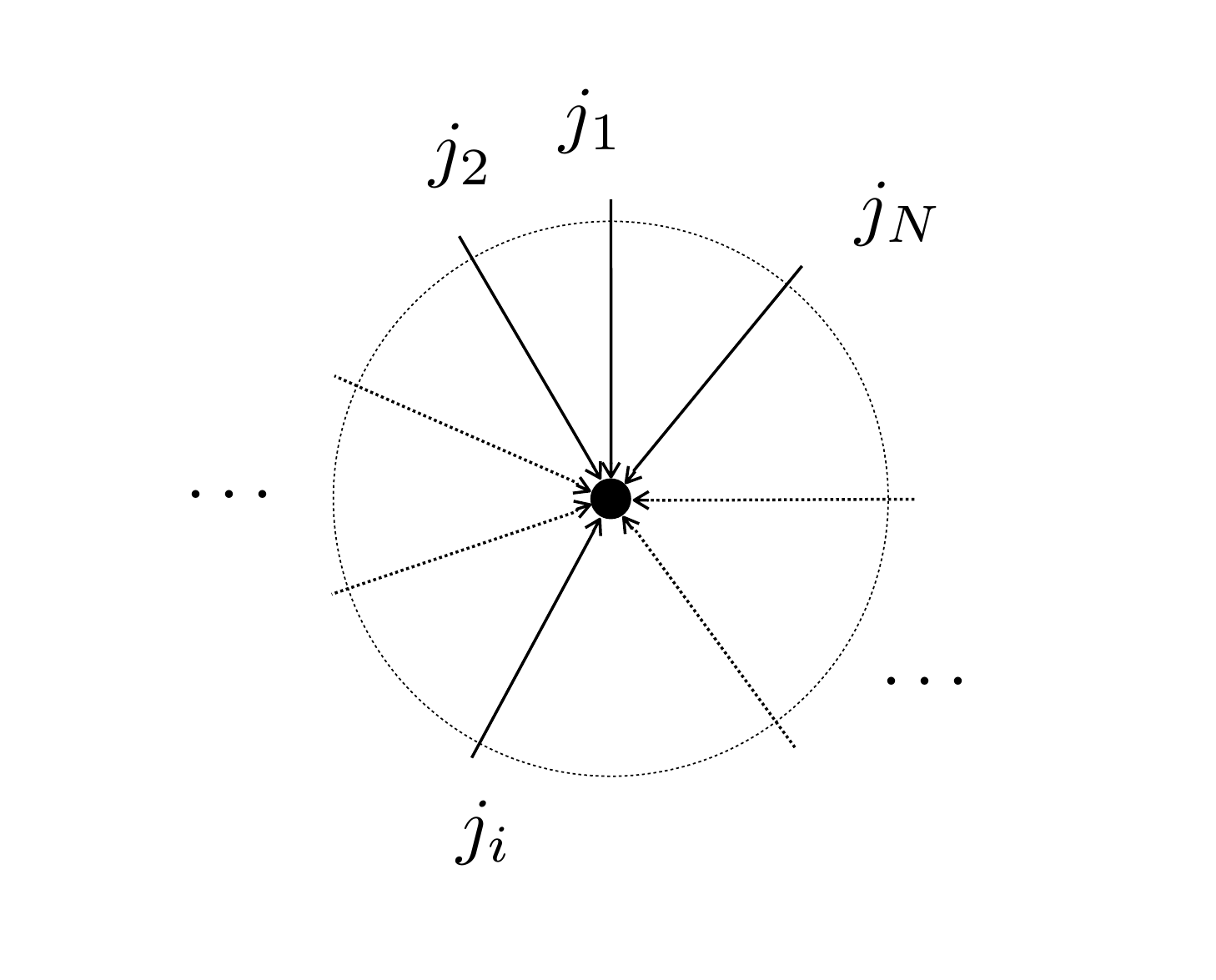}
\caption[The $N$-legged intertwiner]{The $N$-legged intertwiner system describes a convex polyhedron with $N$ faces, with the topology of a \emph{2-sphere}. The $N$ (oriented) edges are dual to the elementary $N$ surfaces comprising the surface of the polyhedron \cite{Bianchi,Barbieri1998,Baez1996,Freidel2011,Livine2013}. The intertwiner contains information on how the elementary surfaces, dual to the links, are combined together to form a surface boundary of the space region dual to the node \cite{Thiemann2008,Rovelli2004}.}\label{sphere}
\end{figure}

It has been proven in \cite{Freidel2010} that the Hilbert space of the $N$-valent intertwiners naturally decomposes into subspaces of constant total area\footnote{The choice of a linear area spectrum $j \times l_P^2$ is favoured by the forthcoming approach involving the $U(N)$ structure of the intertwiner space.} which, following the notation in \cite{Freidel2010} we call $\mathcal{H}_N$. Therefore $\mathcal{H} = \bigoplus_{J} \mathcal{H}_N^{(J)}$. We further constrain our system by considering only the invariant \emph{tensor product} Hilbert space, with total spin fixed to $J=J_0$ (see Fig. \ref{sphere}). This is the last ingredient. Eventually, the constrained Hilbert space is given by $\mathcal{H}_{\mathcal{R}} = \mathcal{H}_N^{(J_0)}$. \\

It was also proven in \cite{Freidel2010} that each subspace $\mathcal{H}_N^{(J)}$ of $N$-valent intertwiners with fixed total area $J$ carries an irreducible representation of $U(N)$. In this context, one can interpret $J_0$ as the total area dual to the set of $N$ legs of the intertwiner. The main reason behind the choice of $\mathcal{H}_N^{(J_0)}$ as constrained space is that in the semiclassical limit one can think of this system, dually, as a closed surface with area $J_0 l^2_P \gg l^2_P$, where $l_P$ is the Planck length. \\


\subsection{The canonical states of the system}\label{redu}

Once the constrained space has been defined, in order to compute the canonical reduced state, we need the expression of the maximally mixed state $\mathcal{I}_{\mathcal{R}}$ over $\mathcal{H}_{\mathcal{R}}$. This is formally given by

\begin{align}
&\mathcal{I}_{\mathcal{R}} \equiv \frac{1}{d_{\mathcal{R}}}{\mathbb{I}_{\mathcal{R}}} = \frac{1}{d_{\mathcal{R}}} P_{\mathcal{R}} ,
\end{align}

where $P_{\mathcal{R}}$ projects the states of $\otimes_{l} \mathcal{H}^{j_l}$ onto the $SU(2)$ gauge invariant subspace with fixed total spin number $\mathcal{H}_N^{(J_0)}$.

When dealing with $SU(2)$ quantum numbers there are two common choices for the basis of the Hilbert space: the coupled and the decoupled basis. The coefficients which connect the two basis are the well-known Clebsh-Gordan coefficient. Considering that the main task is to perform the partial trace of $\mathcal{I}_{\mathcal{R}}$ over the environment, a suitable basis to write the projector $P_{\mathcal{R}}$ is a \emph{semi-decoupled} basis in which all the quantum numbers within the system and within the environment, respectively, are coupled, but the environment and the system are not.\\

\begin{figure}[h]
\centering
\includegraphics[width=3 in]{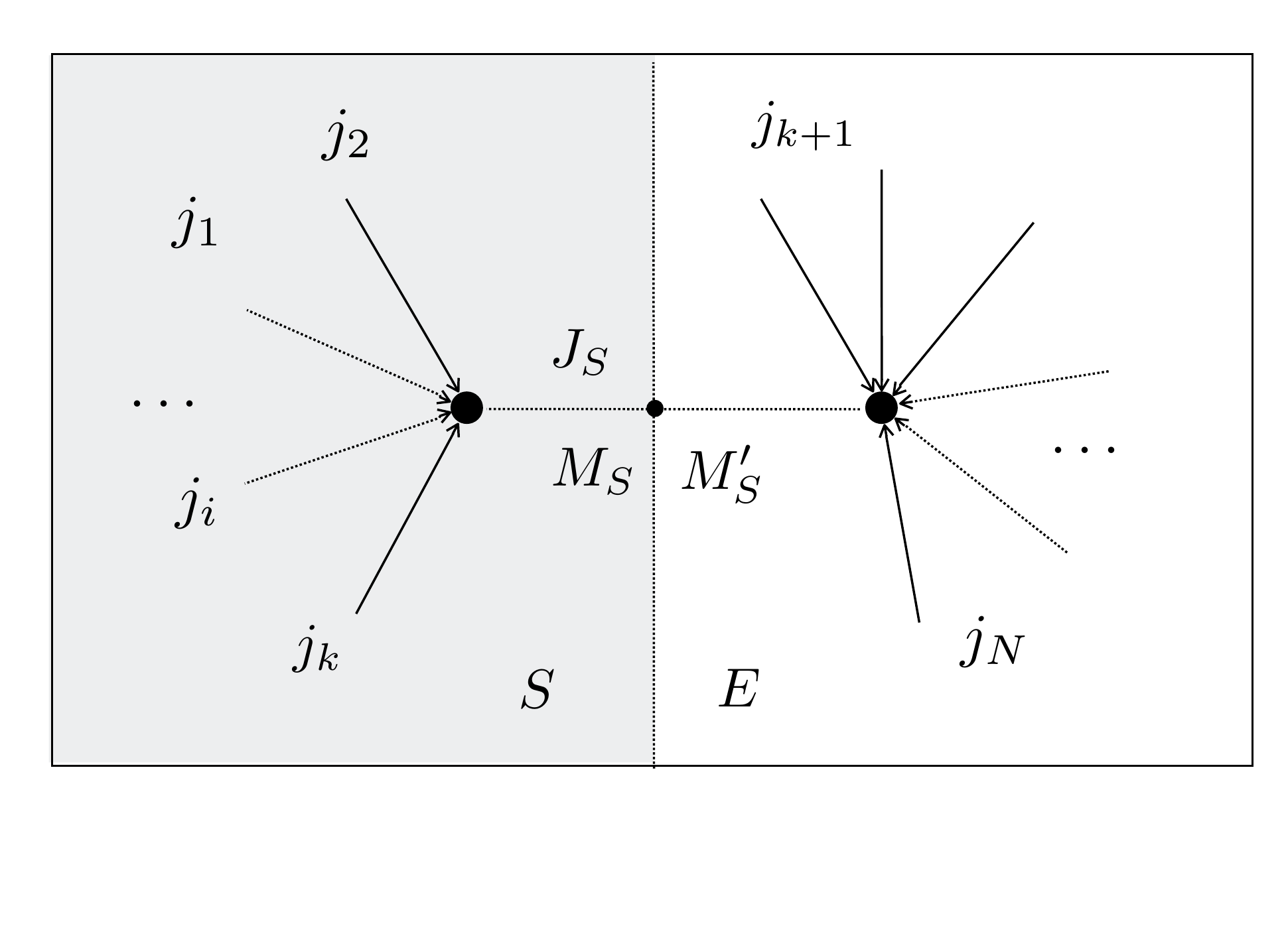}
\caption[Graphic illustration of a semi-decoupled basis]{Here we show a graphic illustration of the semi-decoupled basis that we are using to write the projector $P_{\mathcal{R}}$ onto the constrained Hilbert space $\mathcal{H}_{\mathcal{R}}$. For a recent work on the splitting of a gauge-invariant system we suggest \cite{Donnelly2016a}.}
\end{figure}

Using such semi-decoupled basis we can write the projector as

\begin{align} \nonumber
&P_{\mathcal{R}} = \sum^{(J_0)}_{\{j_E, j_S\}} \sum_{\eta_E, \sigma_S} \sum_{|\vec{J}_S|,M_S,{M'_S}} \frac{(-1)^{M_S + M'_S}}{d_{|\vec{J}_S|}} \cdot \nonumber \\
& \quad \cdot \Ket{\{j_E, j_S\}; \eta_E, \sigma_S; |\vec{J}_S|, -M_S; |\vec{J}_S|, M_S} \\ \nonumber
& \qquad  \qquad \Bra{\{j_E, j_S\}; \eta_E, \sigma_S; |\vec{J}_S|, -M'_S; |\vec{J}_S|, M'_S} 
\end{align}

where with $\Ket{\{j_E, j_S\}; \eta_E, \sigma_S; |\vec{J}_S|, -M_S; |\vec{J}_S|, M_S}$ we mean $\Ket{\{j_E\} ; \eta_E; |\vec{J}_S|, -M_S}_E  \otimes \Ket{ \{j_S\}; \sigma_S; |\vec{J}_S|, M_S }_S$, $d_{|\vec{J}_S|} \equiv 2 |\vec{J}_S| +1$ and $\sum^{(J_0)}_{\{j_E, j_S\}}$ means that we are summing only over the configurations of the spins $\left\{ j_i\right\}$ such that $\sum_{i \in E} j_i + \sum_{k \in S} j_k = J_0$. The quantum numbers  $\sigma_S$ and $\eta_E$ stand for the recoupling quantum numbers necessary to write the state in the coupled basis, respectively within the system and the environment. Eventually, $|\vec{J}_S|$ and $M_S$ are, respectively, the norm of the total angular momentum of the system and its projection over the $z$ axis; $|\vec{J}_E|$ and $M_E$ have the same meaning but they refer to the environment.

The details of the generic element of the semi-decoupled basis and of the way in which we obtain the projector can be found in the Supplementary Material.\\

The dimension of the  constrained Hilbert space $d_{\mathcal{R}} \equiv \mathrm{dim} (\mathcal{H}_{\mathcal{R}}) $ counts the degeneracy of the $N$-valent intertwiners with fixed total spin $J_0$. Given the equivalence between the space $\mathcal{H}_N^{(J_0)}$ of $N$-valent intertwiners with fixed total area $\sum_i j_i=J_0$ (including the possibility of trivial SU(2) irreps) and the irreducible representation of $U(N)$ formalism for $SU(2)$ intertwiners \cite{Freidel2010}, $d_{\mathcal{R}}$ can be calculated as the dimension of the equivalent maximum weight $U(N)$ irrep with Young tableaux given by two horizontal lines with equal number of cases $J_0$,
\begin{align} \label{dime}
d_{\mathcal{R}}=  \frac{1}{J_0+1} \binom{N + J_0 -1}{J_0} \binom{N+ J_0 -2}{J_0} 
\end{align}

 Thanks to the tensor product structure of the semi-decoupled basis, with respect to the bipartition of the universe into system and environment, we can easily perform the partial trace operation over the environment. The details of the computation can be found in the Supplementary Material. The final expression of the canonical state of the system is
 
\begin{align}
&\Omega_S = \sum_{J_S\le J_0/2}  \sum_{\sigma_S, |\vec{J}_S|, M_S} \sum_{\{{j}_S \}}^{J_S} \frac{D_{(N-k)}( |\vec{J}_S|, J_0-J_S )}{d_{|\vec{J}_S|}\,d_{\mathcal{R}}} \cdot \\ 
& \quad \cdot \Ket{\{j_S\}, \sigma_S,|\vec{J}_S|,M_S}  \Bra{\{j_S\},\sigma_S,|\vec{J}_S|,M_S} \nonumber 
\end{align}\label{canonico}

\begin{center}
\begin{figure}[h]
\centering
\includegraphics[width=4.5 in]{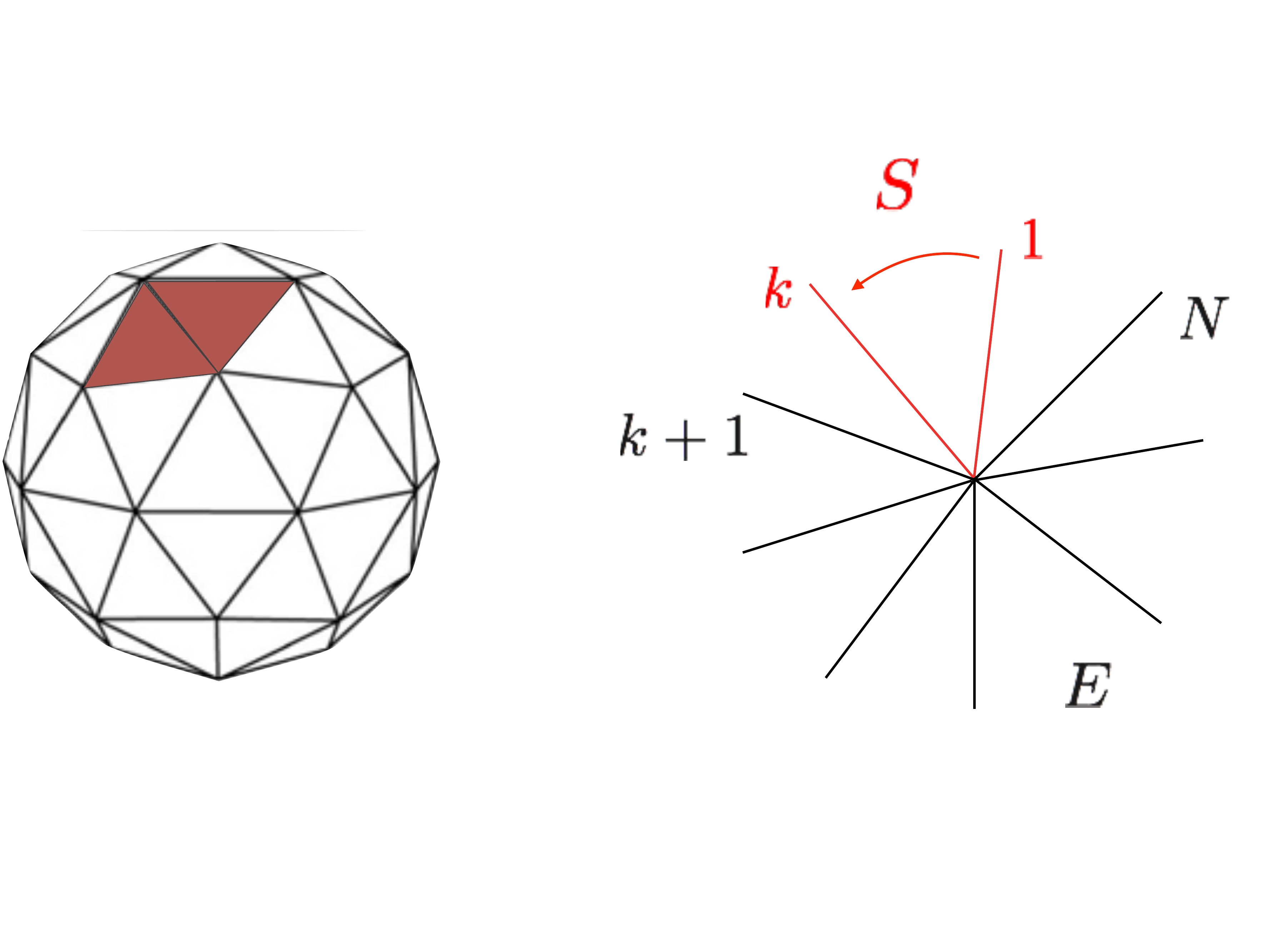}
\caption[Small patch of a larger surface]{A local patch of the 2d surface (in red), associated to a subset of intertwined links $\{ j_{1},\cdots ,j_k\}$ defining the ``system''. The ``environment'' is identified with the complementary 2d-surface associated to the set of links  $\{ j_{k+1},\cdots ,j_N\}$, with $N \gg k$.}\label{sphere2}
\end{figure}
\end{center}

Where $\sum_{\{{j}_S \}}^{J_S}$ means that we are summing over the configurations of the spins of the system $\left\{ j_S\right\}$ such that $\sum_{k \in S} j_k = J_S$. Moreover the definition of the $D-$functions is

\begin{equation} 
D_{(Q)}( x,y )\coloneqq  \frac{2x + 1}{x+y + 1} \binom{Q + y + x - 1}{x+y} \binom{Q + y - x - 2}{y - x}  \label{eq:canoni}
\end{equation}


We also define the following short-hand notation $W_{\mathcal{E}} \equiv D_{(N-k)}( |\vec{J}_S|, J_0-J_S )$. We will also call $W_{\mathcal{S}}\equiv D_{(k)}( |\vec{J}_S|, J_S )$ the dimension of the system's degeneracy space with fixed area $J_S$ and closure defect $ |\vec{J}_S|$, derived from the equivalent $U(N)$ representation as for the case of the environment in Eq.(\ref{eq:canoni}). 



The canonical weight $W_{\mathcal{E}}$ encodes all the information about the local structure of correlations of the reduced intertwiner state. The specific form of this factor tells us about the physics of the system, defined by the specific choice of constraints: the SU(2) gauge symmetry and the fixed total area constraint. Given the global constraint, the split in system and environment breaks the gauge symmetry. Due to the presence of the constraint, onto $\mathcal{H}_{\mathcal{R}}$ the quantum numbers of system $( \left\{ j_S \right\}, \sigma_S, |\vec{J}_S|, M_S )$ are intertwined with those of the environment $( \left\{ j_E \right\}, \eta_E, |\vec{J}_{E}| =|\vec{J}_{S}|, M_E = - M_S )$. This is why, beside the expected dependence on the total area of the system $J_S$, the canonical weight carries some interesting extra information on the local closure defect $ |\vec{J}_S|$.  

\section{Typicality of the reduced state}\label{typ}

In this section we study the region of the space of the parameters $(N,k,J_0,J_{max})$ where the canonical reduced state is typical. In other words, we investigate the distance of the canonical state from a randomly chosen pure state in $\mathcal{H}_{\mathcal{R}}$. 

Concretely, following the approach described in Section \ref{tysn}, we want to show that for the overwhelming majority of intertwiner states $|\mathcal{I}\rangle \in \mathcal{H}_{\mathcal{R}} \subseteq \mathcal{H}_E \otimes \mathcal{H}_S$, the trace distance $D(\rho_S,\Omega_S)$ between the reduced density matrix of the system $\rho_S = \Tr_E(|\mathcal{I} \rangle \langle \mathcal{I}|)$ and the canonical state $\Omega_S$ is extremely small\footnote{We remember that the trace-distance has an important physical interpretation: $D(\rho,\sigma)$ is the probability of telling apart $\rho$ and $\sigma$, by means of the most effective quantum measurement\cite{Nielsen2010}}. This proves two things: first, that the Hilbert space average of such trace distance is itself quite small in the regime in which we are interested in
\begin{align}
&\mathbb{E} \left[ D(\rho_S,\Omega_S)\right] \ll 1 \,\,,
\end{align}
where $\mathbb{E}$ indicates the Hilbert space average performed using the unique unitarily invariant Haar measure \cite{Gemmer2010,Bengtsson2008}. Second, that the fraction of states for which such distance is higher than a certain $\epsilon$ is exponentially vanishing in the dimension of the Hilbert space. \\

Following \cite{Popescu2006}, we have a bound on the average distance
\begin{align} \label{bbound}
&0 \leq \mathbb{E} \left[ D(\rho_S,\Omega_S)\right] \leq \sqrt{\frac{d_S}{d_E^{\mathrm{eff}}}} \leq \frac{d_S}{\sqrt{d_{\mathcal{R}}}}
\end{align}
Concretely, the first step toward the statement of typicality in our context amounts to study in which region of the parameters space $(J_0,N,k,J_{max})$ we have $d_S/\sqrt{d_{\mathcal{R}}} \ll 1$.

\subsection{Evaluation of the bound}

The Hilbert space of the system is the tensor product Hilbert space of the set of irreps $V^{j_i}$ with a given cutoff $J_{max}$. We assume $J_{max}\ge J_0$, in order to be sure that $\mathcal{H}_{N}^{(J_0)}$ will always carry an irreducible representation of $U(N)$.  Each $V^{j}$ has dimension $d_j=2j+1$. Therefore, considering the set of $k$ edges comprising the system, we have
\begin{align}
d_S = \prod_{i=1}^{k} \sum_{j_i=0,\frac{1}{2}}^{J_{max}}(2j_i+1) = \left( 2J_{max}+1\right)^{k} \left( J_{max}+1\right)^{k}
\end{align}
Analogously, for the environment we get $d_E = \left( 2J_{max}+1\right)^{N-k} \left( J_{max}+1\right)^{N-k}$. Since $d_{\mathcal{R}}$ is given in \eqref{dime}, we can focus on the last inequality in \eqref{bbound} and define the regime where $\mathbb{E} \left[ D(\rho_S,\Omega_S)\right] \ll 1 $. Studying the ratio
\begin{align}
&\frac{d_S^2}{d_{\mathcal{R}}} = \frac{(2 J_{\mathrm{max}} + 1)^{2k}(J_{\mathrm{max}} +1 )^{2k}}{\frac{1}{J_0 + 1} \binom{N+J_0 -1}{J_0} \binom{N+J_0 -2}{J_0}} \label{ratio}
\end{align}
we can see that $N$ and $J_0$ play a rather symmetric role in making this quantity small. The region of interest is certainly $J_0 \gg1$ or $ N \gg 1$, or both. As we will argue in the next section, $J_0,N \gg 1$ is precisely the regime of interest for the thermodynamical limit. Therefore we focus on this region, where there are two different regimes: $J_0 \gg N \gg 1$ or $N \geq J_0 \gg 1$. In both cases there are wide regions of the parameters space where the inequality $\mathbb{E} \left[ D(\rho_S,\Omega_S)\right] \ll 1$ holds. We were able to extract the following two conditions which guarantee an exponential decay of \eqref{ratio}, either on $N$ or on $J_0$:
\begin{subequations}
\begin{align}
&\frac{J_0}{k} > \log J_{max} && (J_0 \gg N \gg  1) \label{cut1}\\
&\frac{N}{k} \log j_0 > 2 \log J_{max} &&  (N \geq J_0 \gg  1) \label{cut2}
\end{align}
\end{subequations}
The details can be found in the supplementary material but we would like to present a physically motivated argument to provide a meaningful value for the cut-off $J_{max}$ and check the plausibility of the given bounds. As argued in \cite{Bianchi2011}, if we look at a sphere with small radius $l$, placed at a large distance $L$, we will see it within a small angle $\phi \sim \frac{l}{L}$. Therefore using the scale of the radius of the observed universe $L_U$ and assuming that there is nothing with size smaller than the planck length $l_P$, we will never see something with angular extension smaller than $\phi_{min} \sim \frac{l_P}{L_U}$.\\

A spherical harmonics of representation $j$ is able to discriminate angular distances of the order $\frac{4\pi}{2j+1}$. Therefore the existence of $\phi_{min}$ means that there is an upper bound to the representation which we need to consider which is $J_{max} \sim \frac{4\pi}{\phi_{min}^2}  = 4\pi \frac{L_U^2}{l_P^2}$. Using this argument we obtain the following cut-off
\begin{align}
&J_{max} \sim 4\pi \frac{L_U^2}{l_P^2} \approx 3 \times 10^{124} \sim e^{124 \times \log 10}
\end{align}
Putting the numbers in \eqref{cut1} and \eqref{cut2} we obtain
\begin{subequations}
\begin{align}
& \frac{J_0}{k} \gtrsim 3 \times 10^2 && (J_0 \gg N \gg  1)\\
& \frac{N}{k} \gtrsim 6 \times 10^2 && (N \geq J_0 \gg  1)
\end{align}
\end{subequations}

\subsection{Levy's lemma}

Following \cite{Popescu2006,Linden}, we can use Levy's lemma (see Appendix \ref{App:Levy}) to bound the fraction of the volume of states which are $\varepsilon$ more distant than $\frac{d_S}{\sqrt{d_{\mathcal{R}}}}$ from $\Omega_S$ as

\begin{align} 
&\frac{\mathrm{Vol} \left[  |\mathcal{I}\rangle \in \mathcal{H}_{\mathcal{R}}\, \vert \, D(\rho_S,\Omega_S) - \frac{d_S}{\sqrt{d_{\mathcal{R}}}}\geq \varepsilon \right] }{\mathrm{Vol} \left[|\mathcal{I}\rangle \in \mathcal{H}_{\mathcal{R}} \right] } \leq B_\epsilon(d_{\mathcal{R}})\\ \nonumber
& B_{\epsilon} (d_{\mathcal{R}})\equiv 4 \, \mathrm{Exp} \left[-\frac{2}{9\pi^3} d_{\mathcal{R}} \varepsilon^2 \right].
\end{align}

The dimension $d_{\mathcal{R}}$ can be evaluated numerically because we have an exact expression. We give a numeric example to show that it is not necessary to have huge areas or number of links for the typicality to emerge. Suppose we can evaluate the trace distance with precision: $\epsilon = 10^{-10}$. Moreover, $\frac{2}{9\pi^3} \sim  7 \times 10^{-3}$. With these numbers we have

\begin{align}
&B_{10^{-10}}(d_{\mathcal{R}}) = 4\mathrm{Exp} \left[-7\cdot 10^{-23} d_{\mathcal{R}}  \right].
\end{align}

Suppose we look at the most elementary patch, just a few links ($k=1,2$). The set of numbers $J_0 = N = 10^4$ gives the following bounds, using a cut-off given by the cosmological horizon 

\begin{subequations}
\begin{align}
&\frac{N}{k} \sim 10^4 \gg 6 \times 10^2 \\
& B_{10^{-10}}(d_{\mathcal{R}}) = 4\mathrm{Exp} \left[-5.6 \times 10^{5992}  \right] \ll 1
\end{align}
\end{subequations}

As we can see, the typicality emerges quite easily, due to the exponential-like growth of the constrained Hilbert space on the number of links $N$ and on the total area $J_0$.\\

The existence of a typical behaviour indicates the emergence of a regime where the properties of the reduced state of the $N$-valent intertwiner state are \emph{universal}. The structure of \emph{local} correlations carried by the reduced state is independent from the specific shape of the pure intertwiner state and it is locally the same everywhere. Due to the global symmetry constraint though, the canonical weight presents a very involved analytic form, despite the extreme simplicity of the system under study. In order to extract some physical information from this coefficient we are going to study its behaviour in the thermodynamic limit. 

\section{Thermodynamic limit \& area laws}\label{thermo}

In the standard context of statistical mechanics, when performing the thermodynamic limit the density of particles must be finite otherwise the energy density would diverge: $N,V \to + \infty$ with $\frac{N}{V} < + \infty$. As we will see in the forthcoming argument, the area is playing here the role of the energy, therefore we think that the correct way of performing the thermodynamic limit consists in taking $N,J_0 \to \infty$ with $\frac{J_0}{N} \equiv j_{0}  < +\infty$, where $j_0$ is the average spin of the intertwiner.\\

The entropy of the system is given by the von Neumann entropy,
\begin{align}
&S(\Omega_S)=-\text{Tr}[\Omega_S\,\log \Omega_S] \,\, .
\end{align}
Given the diagonal form of the canonical reduced density matrix $\Omega_S$ in Eq. (\ref{canonico}), this can be written as
\begin{align} \label{entro}
&S(\Omega_S)= -  \frac{1}{d_{\mathcal{R}}} \sum_{J_S\le J_0/2, |\vec{J}_S|} W_{\mathcal{S}} \,W_{\mathcal{E}}\log{\left(\frac{W_{\mathcal{E}}}{d_{|\vec{J}_S|}\,d_R}\right)}.
\end{align}

Within the typicality regime ($N, J_0 \gg1$) we can use the Stirling approximation for the factorials, to simplify the form of the binomial coefficients in $W_{\mathcal{E}}, W_{\mathcal{S}}$ and $d_{\mathcal{R}}$. We will study separately the three regimes $j_0 \gg 1$, $j_0 \ll 1$ and $j_0 \sim 1$. The details of the computation can be found in the supplementary material, here we only summarise the results.

\paragraph{Small average spin: $j_0 \ll 1$ -} In the case of small average spin, the leading term in the thermodynamic limit is 

\begin{align} \label{area}
&S(\Omega_S) \simeq \beta \langle 2J_S\rangle+\text{small corrections}\end{align}
where $\langle \cdot \rangle$ is the quantum mechanical average, on the canonical state $\Omega_S$, while
\begin{align} 
\beta \equiv \left(1+\log{\frac{N-k}{J_0}} \right)
\end{align}
is formally identified as the ``temperature'' of the environment. It turns out to be a function of the averaged spin of the environment.

Despite being quite far from the standard setting, a hint toward a thermodynamical interpretation of this result comes from the $U(N)$ description of the $SU(2)$ intertwiner space. Using the Schwinger representation of the $\mathfrak{su}(2)$ Lie algebra \cite{Girelli2005,Freidel2010}, one can describe the $N$-valent intertwiner state as a set of $2N$ oscillators, $a_i, b_j$. The quadratic operators $E_{ij}\equiv(a^{\dagger}_i a_j- b^{\dagger}_i b_j), E^{\dagger}_{ij}= E_{ji} $ acts on couples of punctures $(i, j)$ and form a closed $\mathfrak{u}(N)$ Lie algebra. The $\mathfrak{u}(1)$ Casimir operator is given by the oscillators' energy operator $E\equiv \sum_i E_i$, with $E_i \equiv E_{ii}$, and its value on a state gives twice the sum of the spins on all legs, $2 \sum_i j_i$. Therefore, one can interpret $E$ as measuring (twice) the total area of the boundary surface around the intertwiner.\\

In statistical mechanics, the thermal behaviour of the canonical state relies on the constraint of energy conservation. The emergence of the canonical state from the micro-canonical occurs as the degeneracy of the environment grows exponentially with the energy, hence decreasing exponentially with the system energy.

In these terms, constraining the total area is equivalent to fix a shell of eigenvalues (in fact a single eigenvalue) of the energy operator acting on the full system. In the limit $N\gg J_0 \gg 1$, the degeneracy of the single energy level grows exponentially.

For such a reason the area scaling described by \eqref{area} is consistent with a thermal interpretation for our reduced surface state. It is also worth to mention that the departure from the exact thermal behaviour, \`{a} la Gibbs, is a signature of the breaking of the global SU(2) symmetry (closure defect), witnessed by the explicit dependence of the reduced state on $|\vec{J}_S|$. 

\paragraph{High average spin: $j_0 \gg 1$}

Here we study the behaviour of the entropy in the regime $J_0\gg N \gg 1$. Up to $O(1/J_0)$ the logarithm of the normalised canonical weight is given by
\begin{align} \label{mah}
&-\log{\left(\frac{W_{\mathcal{E}}}{d_{J_S}\,d_R}\right)} \simeq -\log \left( \frac{J_0e}{N-k} \right)^{-2k} + \frac{3k}{N}-\frac{2kJ_S}{J_0}\\ \nonumber
&- \frac{2J_S+2|\vec{J}_E|}{J_0} \simeq k \log \left( \frac{J_0e}{N-k} \right)^2 + \text{small corrs}
\end{align}

Interestingly, the leading term does not depend on the quantum numbers of the system. Therefore the entropy is counting the number of orthogonal states on which the canonical state has non-zero support

\begin{align}
&S(\Omega_S) \simeq 2k \left( 1+ \log \left( \frac{J_0}{N} \right) \right)+ O\left(\frac{k}{N}, \frac{1}{J_0}\right) \,\,\, . \label{eq:ent2}
\end{align}
This makes the entropy extensive in the number of edges comprising the dual surface of the system. In this sense, the term $\left[ \left( \frac{J_0e}{N-k} \right)^2\right]^k$ defines some kind of \emph{effective} dimension of the system, suggesting that the following two things happen in such regime: first, the canonical state has approximately a tensor product structure; second that the total spin is equally distributed among all spins in the universe therefore the accessible Hilbert space of each spins is roughly limited by a representation of the order of $j_0$. 
The validity of this interpretation can be checked assuming a tensor product structure of $k$ links with single-link Hilbert space limited to the representation $\alpha \times j_0$ and computing the entropy $S_{\mathrm{eff}}$ as the logarithm of the dimension of this space. If we can find an $\alpha \sim 1$ such that the difference $S(\Omega_S) - S_{\mathrm{eff}}$ is proportional only to small corrections $O(\frac{1}{N},\frac{1}{J_0})$, we can say that our argument is not too far from what is happening in such a regime. With these assumptions the effective dimension of the Hilbert space of the system is 

\begin{align}
&d_S^{\mathrm{eff}} =\prod_i  \sum_{j_i = 0, \frac{1}{2}}^{\alpha j_0}  (2j_i +1) = (2 \alpha j_0 + 1)^{k}(\alpha j_0 + 1)^k 
\end{align}

In the $j_0 \gg 1$ regime we can write it as $d_S^{\mathrm{eff}} \simeq 2^k \alpha^{2k} j_0^{2k} + O(\frac{k}{j_0})$ which gives

\begin{align}
&S_{\mathrm{eff}} \equiv \log d_S^{\mathrm{eff}} \simeq 2k \log \left( \frac{J_0}{N} \right) + k \left(\log 2 \alpha^2 \right)
\end{align}

The difference between the two entropies 

\begin{align}
&S(\Omega_S) - S_{\mathrm{eff}} \simeq k(2-\log 2\alpha^2 ) + O\left( \frac{1}{N}, \frac{1}{J_0}\right)
\end{align}

is given only by small corrections of order $O\left( \frac{1}{N}, \frac{1}{J_0}\right)$ when $\alpha \simeq \frac{e}{\sqrt{2}}\approx 1.92$. This simple computation provides evidence that the result in Eq.(\ref{eq:ent2}) follows from the two aforementioned assumptions.

\paragraph{Order $1$ average spin: $j_0 \sim 1$.} Eventually, we compute the behaviour of the entropy in the intermediate regime $J_0 \sim N \gg 1$. With respect to the previous cases, this regime does not add anything new to the analysis. The observed behaviour is  extensive in the number of links of the system, with a coefficient which is slightly different from the previous one:

\begin{align}
&S(\Omega_S) \simeq \left(2k-3\right) (1+\mathcal{O}\left[ \left( \frac{k}{N}\right)^2\right])
\end{align}

The relevant computation can be found in the supplementary material.

\section{Summary and Discussion}\label{fine}

In this manuscript we extend the so-called typicality approach, originally formulated in statistical mechanics contexts, to a specific class of tensor network states given by $SU(2)$ invariant spin networks. In particular, following the approach given in \cite{Popescu2006}, we investigate the notion of canonical typicality for a simple class of spin network states given by $N$-valent intertwiner graphs with fixed total area.  Our results do not depend on the physical interpretation of the spin-network, however they are mainly motivated by the fact that spin networks provide a gauge-invariant basis for the kinematical Hilbert space of several background independent quantum gravity approaches, including loop quantum gravity, spin-foam gravity and group field theories.

The first result is the very existence of a regime in which we show the emergence of a canonical typical state, of which we give the explicit form. Geometrically, such a reduced state describes a patch of the surface comprising the volume dual to the intertwiner. The structure of correlations described by the state should tell us how local patches glue together to form a closed connected surface in the quantum regime. 

We find that, within the typicality regime, the canonical state tends to an exponential of the total spin of the subsystem with an interesting departure from the Gibbs state. 
The exponential decay \`a la Gibbs of the reduced state is perturbed by a parametric dependence on the norm of the total angular momentum vector of the subsystem (closure defect).  Such a feature provides a signature of the non local correlations enforced by the global gauge symmetry constraint. This is our second result. 

We study some interesting properties of the typical state within two complementary regimes, $N \gg J_0 \gg 1$ and $J_0 \geq N  \gg 1$. In both cases, we find that the area-law for the entropy of a surface patch must be satisfied at the local level, up to sub-leading logarithmic corrections due to the unavoidable dependence of the state from the closure defect. However, the area scaling interpretation of the entropy in the two regimes is quite different. In the $N \gg J_0 \gg 1$ regime, the result is related to the definition of a generalised Gibbs equilibrium state. The area is playing the role of the energy, as imposed by the specific choice of the global constraint, requiring total area conservation.  

On the other hand, in the $J_0 \geq  N  \gg 1$ regime, the area scaling is given by the extensivity of the entropy in the number of links comprising the reduced state, as for the case of the generalised (non $SU(2)$ invariant) spin networks \cite{Donnelly}. In this regimen, each link contributes independently to the result, indicating that the global constraints are very little affecting the local structure of correlations of the spin network state. Still, interestingly, the remainder of the presence of the constraints can be read in the definition of what looks like an effective dimension for the single link Hilbert space. 

We interpret these results as the proof that, within the typicality regime, there are certain (local) properties of quantum geometry which are ``universal'', namely independent of the specific form of the global pure spin network state and descending directly from the physical definition of the system encoded in the choice of the global constraints. These properties are heralded by the specific form of the canonical state and pertain the macroscopic-scale phenomenology of the quantum system under scrutiny. Moreover, it is important to stress here that this technique allows to address both the quantum and the semiclassical regime of spin networks. In other words, by focussing on the macroscopically relevant properties we have a simplified picture of the phenomenology of spin networks. This allows to improve our understanding of the semiclassical properties of spin networks and their connection with classical gravity. With this perspective in mind, here we discuss a connection between this work and a result in classical gravity.

In 1995 Ted Jacobson proved\cite{Jacobson} that Einstein Equations of General Relativity can be derived in a thermodynamic fashion, as an equation of state. He was able to show that Einstein Equations follow from the technical request that $\delta Q = T dS$ must hold for any local Rindler causal horizon, with $\delta Q$ being the energy flux, $T$ being the Unruh temperature and $S$ being the entropy of the horizon surface, which is assumed to be proportional to its area. Our result here is in full agreement with Jacobson's assumption that $S \propto \mathrm{Area}$. Indeed, thanks to typicality we observe that a small patch of surface has, almost always, an entropy proportional to its area. The fraction of states which evade this statement is exponentially suppressed in the dimension of the Hilbert space. Therefore, in the macroscopic regime, where classical gravity is expected to hold, we should have that the entropy of a surface is proportional to its area. This means that the Hilbert space of spin networks, in the macroscopic regime, is able to capture this particular aspect of classical gravity, which seems to be paramount for the emergence of Einstein Equations of General Relativity. This argument therefore supports the spin networks as underlying Hilbert space of the states for a quantum theory of gravity.

We would like to stress that our result is purely kinematic, being a statistical analysis on the Hilbert space of spin-network states. For the case of a simple intertwiner state, such study necessarily requires to consider a system with a \emph{large number} of edges, beyond the very large dimensionality of the Hilbert space of the single constituents. 
In this sense, the presented statistical analysis and thermal interpretation is very different from what recently done in \cite{Chirco2015,Chirco2014,Hamma2018}, considering quantum geometry states characterised by few constituents with a high dimensional Hilbert space. 
In fact, we expect a {large number} statistical analysis to play a prominent role in facing the problem of the continuum in quantum gravity. Therefore we think it is important to propose and develop new technical tools which are able to deal with a large numbers of elementary constituents and extract physically interesting behaviours.\\

The kinematic nature of the statement of typicality, together with its general formulation in terms of constrained Hilbert spaces given in \cite{Popescu2006}, provide an important tool to study the possibility of a thermal characterisation of reduced states of quantum geometry, regardless of any hamiltonian evolution in time. Beyond the simple case considered in the paper and in a more general perspective, we expect typicality to be useful to understand how large the effective Hilbert space of the theory can be, given the complete set of constraints defining it. It will also help in understanding which typical features we should expect to characterise a state in such space. If we think of dynamics as a flow on the constrained Hilbert space, we generally expect that, even if the initial state is highly non-typical, after a transient regime we will find the system in a state which is very close to the typical state. This happens because, as it has been shown in the original paper on typicality, the number of states close to the typical state are the overwhelming majority.\\

Finally, it is interesting to look at the proposed ``generalised'' thermal characterisation of a local surface patch, within the standard LQG description of the horizon, as a closed surface made of patches of quantized area. Differently from the \emph{isolated horizon} analysis (see e.g. \cite{Pranzetti,Oriti,Ghosh}), in our description the thermal character of the local patch is not (semi)classically induced by the thermal properties of a black hole horizon geometry, but emerges from a purely quantum description. In this sense, our picture goes along with the informational theoretic characterisation of the horizon proposed in \cite{Livine2006}. 

In fact, we believe typicality could be used to define an information theoretic notion of quantum horizon, as the boundary of a generic region of the quantum space with an emergent thermal behaviour.






\chapter{Fate of the Hoop Conjecture}\label{Ch8}


We consider a closed region $R$ of 3D quantum space described via SU(2) spin-networks. Using the concentration of measure phenomenon we prove that, whenever the ratio between the boundary $\partial R$ and the bulk edges of the graph overcomes a finite threshold, the state of the boundary is always thermal, with an entropy proportional to its area. The emergence of a thermal state of the boundary can be traced back to a large amount of entanglement between boundary and bulk degrees of freedom. Using the dual geometric interpretation provided by loop quantum gravity, we interpret such phenomenon as a pre-geometric analogue of Thorne's ``Hoop conjecture'', at the core of the formation of a horizon in General Relativity. This chapter is based on the work published by the author, in collaboration with Dr. G. Chirco, in Ref. \cite{Anza2017b}.

\section{Introduction}


In the previous chapter we showed that typicality can be fruitfully applied to the Hilbert space of the spin-network states of quantum geometry. In this context, it provides a remarkable tool to study the local behavior of a quantum geometry, in a fully kinematic approach, hence consistently with the general covariant nature of the spin-network description. Here we propose a radical shift of setting where the role of ``system'' and ``bath'' is played by the boundary and bulk degrees of freedom of a generic \emph{3-ball} of quantum space. Along with the statement of canonical typicality we prove that, whenever the bulk graph is sufficiently complex, the reduced state of the boundary is almost always a thermal state, regardless of what is the specific pure state of the whole region.  We find that the emergence of a typical state of the boundary is regulated by a \emph{threshold condition} on the ratio between boundary and bulk degrees of freedom.

In a series of works, similar thermal states of a boundary surface have been proposed as the pre-geometric equivalent of black hole configurations in non-perturbative quantum gravity~\cite{Charles2016,Livine2006,Livine2008}. By reading the topological defect of the boundary as a measure of  curvature in the bulk\cite{Livine2014}, we interpret the threshold condition for the emergence of typicality as a pre-geometric quantum analogue of the threshold mechanism associated to the formation of a black hole horizon in General Relativity, the famous Thorne's ``Hoop conjecture''~\cite{Flanagan1991}: {\it an horizon will form if and only if a mass $M$ gets compressed into a region with circumference in any direction $\mathcal{C}$ proportional to its mass}
\begin{align}
&\mathcal{C} \leq \mathcal{C}_{\mathrm{HC}} && \mathcal{C}_{\mathrm{HC}} \sim M
\end{align}
in units $G=c=1$.

Our result strongly supports the relation among thermal boundary states and horizons and it will help in the process of understanding how the thermal behavior of a black hole can emerge from a quantum theory of gravity. Moreover, it resonates with some recent results from tensor network theory which investigate the holographic mapping from bulk to boundary degrees of freedom \cite{Hayden2016,Pastawski2015,Qi2013,Ryu2006a,Chirco2017}. Thanks to its purely kinematic character, our result can provide new insights to the issue of the implementation of the holographic principle and its possible relation with coarse graining in loop quantum gravity \cite{Dittrich2014,Livine2017}. 

\subsection{A 3D region of quantum space} 

A spin network state $|\Gamma, \{j_e\},\{i_v\}\rangle$ is defined as the assignment of representation labels $\{j_e\}$ to each edge and the choice of a vector $|\{i_v\}\rangle \in \otimes_v\mathcal{H}_v$ for the vertices. For a given graph $\Gamma$, the spin networks provide a basis for the space of square-integrable wave-functional, endowed with the natural scalar product induced by the Haar measure on SU(2)
\begin{align}
\mathcal{H}_{\Gamma} = L_2[SU(2)^E/SU(2)^V]= \bigoplus_{\{j_e\}} \bigotimes_v \mathcal{H}_v \,\,\, .
\end{align}
Such spin-network states have a well defined geometric interpretation as collection of polyhedra glued together\cite{Bianchi,Barbieri1998}. Each node is dual to a polyhedron while each edge is dual to a surface patch with area proportional to the spin labelling its irrep $j_e$. 

\begin{figure}[h!]
\centering
\includegraphics[scale=0.3]{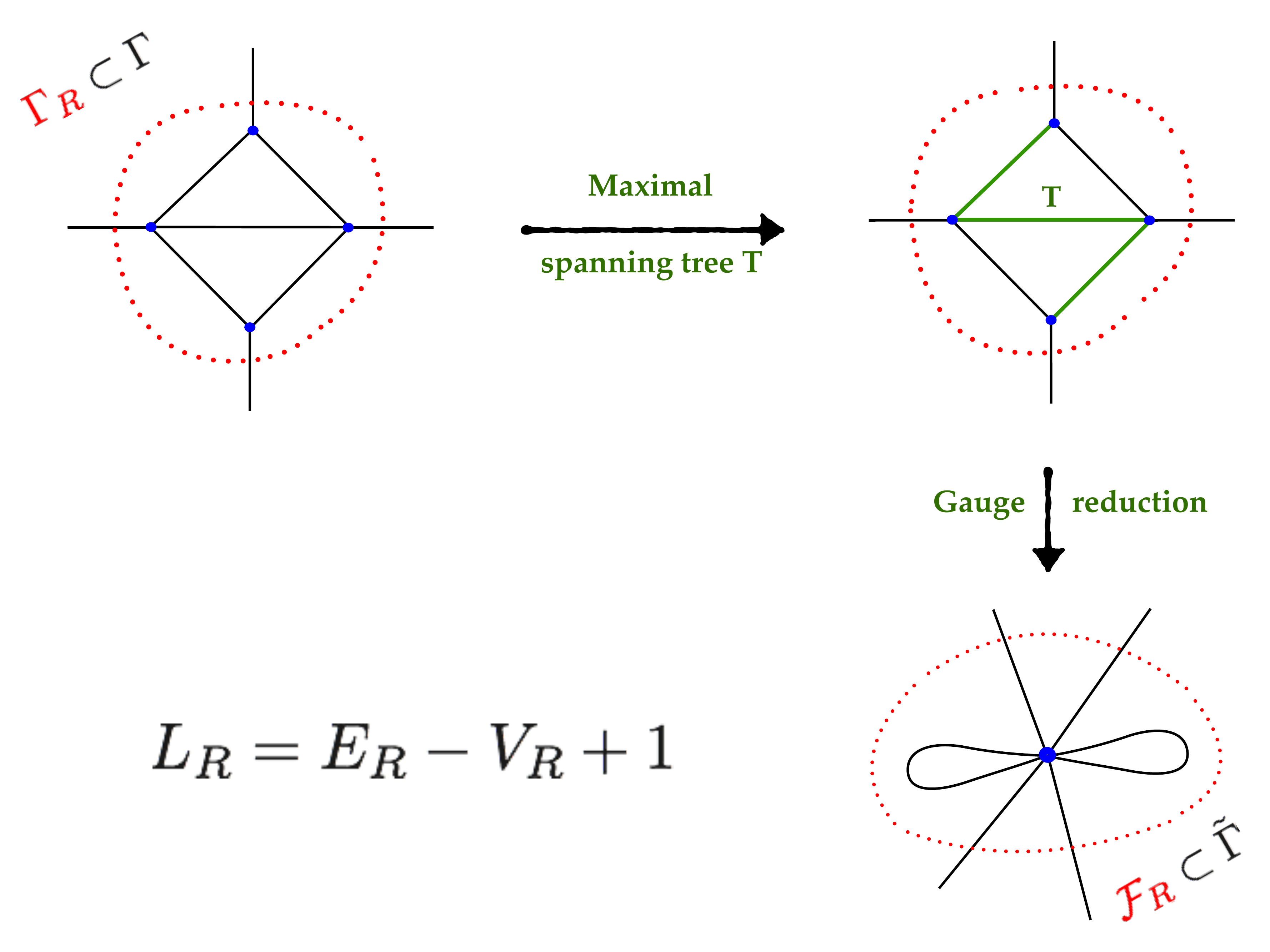}
\caption[Gauge-fixing procedure for a spin network]{Simple example of a subgraph $\Gamma_R$ reduced to a flower graph $\mathcal{F}_R$ with two loops and four external edges. A maximal spanning tree $T$ is chosen out of six different ways. Within the region $R$ the number of independent loops $L_R$ is fixed by the number of edges $E_R$ and vertices $V_R$ within $R$: $L_R = E_R - V_R + 1$. In this case we have $E_R=5$ and $V_R=4$. Thus the number of loops of $\mathcal{F}_R$ is $L_R=2$.}
\label{flowers}
\end{figure}
In particular, a bounded region $R$, dual to a subregion $\Gamma_R \subset \Gamma$ of the global graph $\Gamma$, includes a finite number of vertices $V_R$ and edges $E_R$. We define a boundary $\partial R$ as the set of edges which have only one end vertex in $R$. Their number is called $E_{\partial R}$. Consistently, bulk edges are paths connecting only vertices within $R$. The number of vertices and edges living outside $R$ is called $V_{\overline{R}}$ end $E_{\overline{R}}$, respectively. We can picture $R$ as a region of $3D$ space with the topology of a 3-ball and $\partial R$ as its boundary 2-sphere.

Exploiting the gauge invariance at each node inside $R$ we can simplify the structure of $\Gamma_R$, without loosing information \cite{Livine2006,Livine2008,Freidel2003,Livine2006b}. The gauge-invariant Hilbert space associated to the original graph is isomorphic to the gauge-invariant space defined on a new graph $\tilde{\Gamma}$, consisting of $V_{\overline{R}}+1$ vertices intertwining $E_{\overline{R}} + E_{\partial R}$ edges, together with a certain number of loops $L_R$ which depends on the internal structure of $\Gamma_R$. The number of independent loops is $L_R=E_R-V_R +1$, where $V_{R}$ and $E_R$ are, respectively, the number of vertices and edges inside $\Gamma_R$. The subgraph $\mathcal{F}_R \subset \tilde{\Gamma}$ describing the region $R$ after the gauge-fixing procedure (see Figure (\ref{flowers})) has the structure of a \emph{Flower} graph\footnote{To be precise, a flower graph is usually one that has only loops (the petals of a flower). However, we will use this convenient nomenclature to indicate a more generic graph with a certain number of loops and of external legs (the stems of a flower).}. The Hilbert space of $\tilde{\Gamma}$ takes the form
\begin{align}
\mathcal{H}_{\tilde{\Gamma}} & \equiv L_2[SU(2)^{E_{\overline{R}}+E_{\partial R} + L_R}/SU(2)^{V_{\overline{R}+1}}] = \bigoplus_{\{j_e,j_l\}} \left[ \left( \bigotimes_{v=1}^{V_{\overline{R}}} \mathcal{H}_{v} \right) \otimes \mathcal{H}_{\mathcal{F}_R} \right] \, .
\end{align}
In the last equality we split the tensor product over the intertwiners in two parts. The first one comprises the $V_{\overline{R}}$ vertices outside the region $R$ while the last one is the vertex which remains after performing the gauge-fixing procedure over the region $R$:
\begin{align}
&\mathcal{H}_{\mathcal{F}_R} \coloneqq \mathrm{Inv}_{\mathrm{SU(2)}} \left[ \bigotimes_{e=1}^{E_{\partial R}} V_{j_e} \bigotimes_{l=1}^{L_R} \left( V_{j_l} \otimes \overline{V}_{j_l}\right)\right] \, .
\end{align}
This Hilbert space is the focus of our investigation as it provides a synthetic description for a region $R$ of quantum space with fixed boundary areas and non-trivial internal degrees of freedom. In the next section we comment on its geometric interpretation.

\subsection{Separating Boundary and Bulk} 
Within $\mathcal{F}_R$ we identify the \emph{boundary} and \emph{bulk} degrees of freedom respectively with open edges and internal loops. 
The separation is motivated by the different geometric information they carry: the open edges irreps describe boundary surface patches, while internal non-contractible loops are associated with ``curvature excitations'' at the vertex \cite{Livine2006,Livine2008}. With no loops, the intertwined boundary edges are dual to the closed surface of a flat polyhedron. When non-contractible loops take part to the overall gauge invariance of the intertwiner state, the surface dual to the sole boundary edges can be seen as a convex polyhedron with a missing face. In this sense, the \emph{closure defect} induced by the loops is  interpreted as the discrete counterpart of the curvature. A gauge-invariant measure of such quantity is given by the trace of the holonomy around the loops.

Along with such qualitative separation of degrees of freedom, we can think of $\mathcal{H}_{\mathcal{F}_R}$ as a bipartite quantum system, with boundary and bulk correlated by the presence of the  SU(2) gauge-invariance constraint. In particular, our constrained space $\mathcal{H}_{\mathcal{F}_R}$ can be seen as embedded in a tensor product space 
\begin{align}\label{space}
\mathcal{H}_{\mathcal{F}_R} \subseteq \mathcal{H}\coloneqq \mathcal{H}_{\partial R} \otimes  \mathcal{H}_R,
\end{align} 
where we define the (unconstrained) boundary and bulk Hilbert spaces, respectively, as 
\begin{align}
&\mathcal{H}_{\partial R} \coloneqq \bigotimes_{e=1}^{E_{\partial R}} V_{j_e} && \mathcal{H}_R \coloneqq \bigotimes_{l=1}^{L_R} V_{j_l} \otimes \overline{V}_{j_l} \,\,\, .
\end{align}
We now focus on the properties of the boundary.

\section{Typicality of the boundary} 
As far as an observer external to the region $R$ is concerned, the geometry of $R$ is described by the information measured on its boundary $\partial R$. Such information is encoded by the reduced density matrix obtained by taking the partial trace of the whole state $ |\varphi_{\mathcal{F}_R}\rangle \langle \varphi_{\mathcal{F}_R}|$ over the loops $\rho_{\partial R} \coloneqq \text{Tr}_{L}[|\varphi_{\mathcal{F}_R}\rangle \langle \varphi_{\mathcal{F}_R}|]$. Starting from the description of a region of quantum space given before, we focus on the boundary edges.  Using the typicality tools (summarised in the supplemental material) one can prove that {\it whenever the dimension of the boundary Hilbert space is much smaller than the dimension of the constrained space, the reduced boundary state is almost always extremely close to the canonical state on the boundary $\Omega_{\partial R}$, regardless what is the global state of the whole region}. 

The {\it canonical state} of the boundary $\Omega_{\partial R}$ is defined as the partial trace, over the bulk degrees of freedom, of the microcanonical state $\mathcal{I}_{\mathcal{F}_R}$ of the space $\mathcal{H}_{\mathcal{F}_R}$:
\begin{align}
&\Omega_{\partial R} \coloneqq \Tr_L \mathcal{I}_{\mathcal{F}_R} &&\mathcal{I}_{{\mathcal{F}_R}} \coloneqq  \frac{1}{d_{{\mathcal{F}_R}}} P_{{\mathcal{F}_R}} 
\end{align}
Here $P_{{\mathcal{F}_R}}$ is the projector of the states of $\mathcal{H}_{\partial R} \otimes \mathcal{H}_R$ onto $\mathcal{H}_{{\mathcal{F}_R}}$, while $d_{\mathcal{F}_R} \coloneqq \Tr P_{{\mathcal{F}_R}}$ gives the dimension of the constrained space. We are interested in the trace-distance~\cite{Nielsen2010} between a generic reduced state $\rho_{\partial R}$ and $\Omega_{\partial R}$: $D(\rho_{\partial R},\Omega_{\partial R})$. Using the concentration of measure argument developed in \cite{Popescu2006} and summarised in Section \ref{sec:Ch2Typicality}, its average over the global Hilbert space $\mathbb{E}\left[D(\rho_{\partial R},\Omega_{\partial R})\right]$ satisfies 
\begin{align} \label{eq:bound}
&0 \leq \mathbb{E}\left[D(\rho_{\partial R},\Omega_{\partial R})\right] \leq \frac{1}{2}\frac{d_{\partial R}}{\sqrt{d_{\mathcal{F}_R}}} 
\end{align}
Moreover, the fraction of states which are $\epsilon > 0$ away from this average is exponentially suppressed in the dimension of the Hilbert space. Therefore, whenever the right-hand side of Eq.(\ref{eq:bound}) is much smaller than one, it will be concretely impossible to distinguish the actual reduced state $\rho_{\partial R}$ from $\Omega_{\partial R}$. The goal is then to evaluate this bound and find the regime where the average distance is close to zero.

\subsection{A simplified setting} 
In order to do this we need a closed expression for both $d_{\partial R}$ and $d_{\mathcal{F}_R}$. This is not always possible for $d_{\mathcal{F}_R}$, due to the fact that it depends on the specific representations of the boundary edges $\left\{j_e\right\}$ and of the loops $\left\{j_l\right\}$. As a first step toward a more complete study and to provide a quantitative analysis of our typicality argument, we now consider a simplified set-up. We work on the case where all the spin representations have the same value: $j_e = j_l = j_0$. While this assumption might seem drastic, the core of our argument does not depend on it. Only our ability to concretely evaluate the two dimensions involved. Moreover, it is worth noting that the system maintains a well-defined geometric interpretation. It describes a polyhedron with $E_{\partial R}$ faces with equal areas and non-trivial internal degrees of freedom. Further investigations aiming at extending the result to the general case will be reported elsewhere.

Within $\mathcal{H}_{\mathcal{F}_R}$, a consistent reorganisation of boundary and bulk degrees of freedom is realised by considering the unfolding of our intertwiner into two intertwiners, linked by a virtual edge, separately re-coupling the edges and the loops irreps. An example is given in Figure (\ref{flowers}). This corresponds to the following re-writing of the gauge-invariant space 
\begin{align} \label{unfolded2}
\mathcal{H}_{\mathcal{F}_R}=  \bigoplus_{k}\mathcal{H}^{(k)}_E \otimes \mathcal{H}^{(k)}_L\,\,,
\end{align}
where $k$ runs over the irreps of the virtual link. Moreover, $\mathcal{H}^{(k)}_E = V_k^{(E)} \cdot \mathcal{D}_k^E$, $\mathcal{H}^{(k)}_L = V_k^{(L)} \cdot \mathcal{D}_k^{2L}$ and $\mathcal{D}_k^E$,$\mathcal{D}_k^{2L}$ counts the degeneracies of the space $V_k$ in the boundary and in the bulk recoupling. The details of the way in which such decomposition is done can be found in the supplemental material, along with the information about the dimensions of $\mathcal{H}_E^{(k)}$ and $\mathcal{H}_L^{(k)}$. Here we just fix the notation. If we write a generic decomposition as $\otimes_n V_{j_0} = \bigoplus_{k} V_k {}^{j_0}F_k^{n} $, we call ${}^{j_0}d_k^{n} \coloneqq \mathrm{dim} \,\,\, {}^{j_0}F_k^{n} $. All dimensions needed in this paper can be written using ${}^{j_0}d_k^{n}$.

\subsection{Evaluation of the bound} 
Building on the technical results derived in \cite{Livine2006,Livine2008}, along with our statistical approach, we work in the regime $E_{\partial R},2L_{R} \gg 1$. Using the expression derived
in \cite{Livine2006,Livine2008} we obtain
\begin{align}
&\frac{d_{\partial R}}{\sqrt{d_{\mathcal{F}_R}}} \sim (2j_0+1)^{\frac{E_{\partial R}}{2}-L_{R}} [(j_0(j_0+1))(E_{\partial R}+2L_{R})]^{3/4}
\end{align}
The details can be found in the supplemental material. This quantity has a leading exponential behavior, both in the number of external edges and in the number of loops over which we trace. The exponent becomes negative when $E_{\partial R} < 2L_{R}$.  Such fast decay is present for any choice of $j_0$, even far from the ``semiclassical regime'' of large areas  $j_0 \gg 1$.
This is not exactly a threshold behavior but it is a fast exponential decay, which becomes faster as we approach the semiclassical regime $j_0 \gg 1$. In that regime the exponential decay of $\mathbb{E} \left[ D(\rho_{\partial R},\Omega_{\partial R})\right]$ to zero approaches a threshold behavior, regulated by the condition
\begin{align}\label{eq:threshold}
&E_{\partial R} < 2L_{R}
\end{align}
Intuitively speaking, the trace of the loop holonomy measures the curvature around a path. If one averages over the whole Hilbert space, each loop will contribute in an equal way to the average scalar curvature, which will be proportional to the number of loops. Thus Eq.\eqref{eq:threshold} suggests the existence of an inequality relating the area of the boundary surface with its scalar curvature. We will discuss the physical meaning of this condition in the conclusions.

\subsection{The typical boundary state} 
We now focus on the explicit form of the typical state of the boundary.  Starting from the decomposition of the constrained space given in Eq.(\ref{unfolded2}), a convenient basis in either of the two subspaces is labeled by three numbers, respectively $|k,m,a_k\rangle$ and $|k,m,b_k\rangle$, with $a_k, b_k$ running over the degeneracy of the irrep $V_k$ at fixed value of $k$ \cite{Anza2016,Livine2006}. A basis for the single intertwiner space is given by 
\begin{align} \label{basis}
| k, a_k , b_k \rangle \coloneqq \sum_{m=-k}^k \frac{(-1)^{k-m}}{\sqrt{2k+1}}  |k, -m, a_k \rangle_E\, \otimes  | k, m, b_k \rangle_L \, .
\end{align}

Each basis state\footnote{Notice that in the basis of Eq.\eqref{basis} the dependence on the spins $\{j_e\}, \{j_l\}$ is hidden in the re-coupled spin label $k$.} can be represented as a tensor product state on three subspaces,
\begin{align} \label{biortho}
|k, a_k , b_k \rangle \coloneqq  |{k} \rangle_{V_k^{(E)}\otimes V_k^{(L)}} \otimes |{a_k}\rangle_{\mathcal{D}_k^E} \otimes |{b_k}\rangle_{\mathcal{D}_k^L}
\end{align}
where $k$ runs over the global angular momentum of the boundary and of the bulk, which have to be equal;  $|a_k\rangle$ labels a basis vector of $\mathcal{D}_k^E$ and $|b_k\rangle$ labels a basis vector of $\mathcal{D}_k^{2L}$. The microcanonical state $\mathcal{I}_{\mathcal{F}_R}$ on $\mathcal{H}_{\mathcal{F}_R}$ is given by the normalised identity, which we write in the basis $\left\{ \Ket{k,a_k,b_k}\right\}$
\begin{align} \label{ro}
&\mathcal{I}_{\mathcal{F}_R} \coloneqq \frac{1}{d_{\mathcal{F}_R}}\sum_{\substack{k, a_k, b_k}} \Ket{k,a_k,b_k}\Bra{k,a_k,b_k}.
\end{align}
Thanks to the specific basis chosen the partial trace of $\mathcal{I}_{\mathcal{F}_R}$ is easily computed. The canonical state of the boundary is
\begin{align}
&\Omega_{\partial R} =  \sum_{\substack{k}} \frac{\,^{j_0}d_{k}^{(2L)}}{d_{\mathcal{F}_R}(2k+1)}\,\,\,{\mathbbm{1}}_{{V}_k^{(E)}} \otimes \mathbbm{1}_{\mathcal{D}_k^E} \, ,
\end{align}
where ${\mathbbm{1}}_{{V}_k^{(E)}}$ and $\mathbbm{1}_{\mathcal{D}_k^E}$ are, respectively, the identity over $V_k^{(E)}$ and $\mathcal{D}_k^E$.

\subsection{Canonical coefficient} 
Here we study the canonical coefficient $W_k^{E} \coloneqq \frac{\,^{j_0}d_{k}^{(2L)}}{d_{\mathcal{R}}(2k+1)}$ in order to understand what is the predicted behavior in the thermodynamic regime, which is $2L_{R} \gg E_{\partial R} \gg 1$. Using the expression given in \cite{Livine2006,Livine2008}  we obtain
\begin{align}
&W_k^E \stackrel{2L_{R} \gg k}{\sim} (2j_0+1)^{-E_{\partial R}} \left(1+\frac{E_{\partial R}}{2L_{\partial R}}\right)^{3/2} \left(\frac{k+1}{2k+1}\right) \label{eq:expansion}
\end{align}
$W_k^{E}$ depends on $k$ only through $\left(\frac{k+1}{2k+1}\right) \in [\frac{1}{2},1]$. This is a very mild dependence and as $k$ increases it fades away. Such behavior can be confirmed by using the expression of $W_k^E$ when $k = k_{\mathrm{max}} = j_0 E_{\partial R}$, which was given in \cite{Livine2006,Livine2008}. More details can be found in the supplemental material. We conclude that the canonical coefficient $W_k^E$ depends in an extremely weak way on the topological defect $k$. Therefore $\Omega_{\partial R}$ is very close to a completely mixed state of the boundary. This picture can  be checked by computing the von Neumann entropy $S_{\mathrm{vN}}(\Omega_{\partial R}) \coloneqq - \Tr \Omega_{\partial R} \log \Omega_{\partial R}$: 
\begin{align}\label{entrop}
&S_{\mathrm{vN}}(\Omega_{\partial R}) \simeq E_{\partial R} \log (2j_0+1) - \frac{3}{2} \log \left( 1+ \frac{E_{\partial R}}{2L_{R}}\right)
\end{align}
where $E_{\partial R} \log (2j_0+1) = \log d_{\partial R}$ is the maximum entropy. This confirms that $\Omega_{\partial R}$ is almost a microcanonical state, with an entropy proportional to its area $A_{\partial R}$:  $S(\Omega_{\partial R}) \propto A_{\partial R} \propto E_{\partial R}$.

\section{Summary and conclusions} We studied the properties of the boundary $\partial R$ of a generic region $R$ of 3D quantum space. We exploit the SU(2) spin-network formalism to provide a synthetic description of $R$ in terms of a flower graph $\mathcal{F}_{R}$. In this setting, the boundary degrees of freedom are living on the boundary edges $E_{\partial R}$ while the bulk degrees of freedom live on the internal loops $L_R$. The state of the system is pure and belongs to a SU(2) gauge-invariant Hilbert space $\mathcal{H}_{\mathcal{F}_R}$ while the state of the boundary is obtained by performing a partial trace over the bulk degrees of freedom. The result is a boundary state which does not satisfy the closure constraint induced by gauge invariance. Such a \emph{closure defect} encodes the total topological defect carried by the loops and it is interpreted as a measure of the average scalar curvature\cite{Livine2006,Livine2008,Freidel2003,Livine2006b} within the region $R$. 

The main result of the work is the existence of a condition $d_{\partial R}/\sqrt{d_{\mathcal{F}_R}} \ll 1$  whose validity guarantees the emergence of a typical state of the boundary. In order to evaluate such condition in a concrete case we studied the simplified setup where all the links have the same spin $j_0$. We observe the emergence of a threshold condition $2L_{R} > E_{\partial R}$ which guarantees that, regardless what is the specific state of the whole system, in such regime the state of the boundary is almost always a thermal state with entropy proportional to the area. The fraction of states eluding this result is exponentially suppressed in the dimension of the total Hilbert space. In the semiclassical regime $j_0 \gg 1$ such exponential suppression becomes a true threshold condition. We thereby argue that our results suggest a strong analogy with the Hoop Conjecture.

In General Relativity, Thorne's Hoop conjecture states that an object which collapses will form a black hole when a `circular hoop', with a circumference proportional to its Schwarzschild radius, can be placed around the object, in all the three directions. Once a black hole is formed, its horizon behaves as a thermal state, with an entropy proportional to the area. Our result shows that something similar happens at the pre-geometric level, within the spin-network formalism of non-perturbative quantum gravity, on a purely {\it information-theoretic} ground: when the number of internal loops exceed a certain threshold, the boundary state is a thermal state and its entropy is proportional to the area. This implies that the boundary does not have retrievable information about the bulk. Both statements point toward an interpretation of such a boundary state as the pre-geometric counterpart of a horizon. The proposed analogy with the Hoop conjecture is strengthened by the explicit form of the threshold condition, $2L_R > E_{\partial R}$. The trace of the loop holonomy provides a measure of the curvature and its average over the Hilbert space is proportional to the number of loops. Moreover since we are looking at the case were all the links have the same irrep, the number of edges $E_{\partial R}$ measures the area of the surface. For this reason we read the semiclassical limit of Eq.\eqref{eq:threshold} as an inequality relating the area of a closed surface with the curvature. Whenever such inequality is fulfilled for a region $R$, the state of its boundary behaves as an horizon as it has no retrievable information about the bulk. 

We would also like to stress that the thermal character of the boundary emerges because of its entanglement with the bulk, induced by the gauge invariance. In this sense, the proposed setting goes along with the vision according to which a classical spacetime, with its causal structure, is an emergent feature which results from the structure of correlations of a more fundamental level of description \cite{Raamsdonk2010,Maldacena2013,Cao2016}. In this picture, the concentration of measure phenomenon appears as a general and concrete tool to unravel the interplay between the structure of correlations in a quantum geometry and the emergence of semiclassical properties in the thermodynamic limit.

\chapter{Conclusions}\label{Ch9}


This conclusive chapter has two sections. In the first one, we give a quick summary of the results achieved, adopting a more general perspective. In the second one, we draw some conclusions, putting the emphasis on the overall picture. We also outline interesting developments for future work which we would like to pursue.

 \section{Summary}

The leitmotif of the research presented in this thesis is the concept of thermal equilibrium. Statistical mechanics provides tools to predict the behaviour of macroscopic systems, under the assumption that the system is at thermal equilibrium. While it is known that there are systems which escape this description, it is also known that predictions from statistical mechanics have been shown to be correct in several cases. An example is given by isolated quantum systems. Due to the unitary character of the dynamics of isolated systems, there is always a huge number of conserved quantities. Therefore, generic initial states will not evolve toward a thermal state as memory of the initial conditions will not be erased. 

Our approach to the problem is based on the idea that having a thermal state is a sufficient condition to ensure a thermal behaviour, but it is not a necessary one. This goes along with the fact that a definition of thermal behaviour which goes through the full characterization of the state can not be experimentally probed at the macroscopic scale: Too many observables have to be measured to reconstruct the state. In the first part of Chapter \ref{Ch3}, the analysis leading to this conclusion is given in full details and we proposed a notion of thermal equilibrium which is experimentally verifiable as it is focused on an arbitrary observable. 

Our approach is not the only one on the table. Pure-states statistical mechanics is an attempt to understand how the tools of statistical mechanics can be justified from an underlying quantum theory. Despite such similarity in the spirit, we would like to stress that our approach takes a further step, as it does not ascribe the emergence of thermal behaviour to the state of the system being of the ``thermal form''. Rather, we focus on what is measurable and address the emergence of thermal equilibrium without going through a characterization of the thermal form for the state of the system. 

In the second part of Chapter \ref{Ch3}, we unravel a connection between our approach and one of the pillars of pure-states statistical mechanics: The Eigenstate Thermalization Hypothesis. The concept of Hamiltonian Unbiased Observables provides a prototype for the observables which satisfy the ETH. The characterization is not complete, as these observables are completely insensitive to the energy while concrete observables can exhibit a mild and smooth dependence on the energy. In Chapters \ref{Ch4} and \ref{Ch5}, we study the consequences of the proposed approach to thermal equilibrium. In particular, in Chapter \ref{Ch4} we investigate the connection between the ETH and HUOs and give a theorem which can be used to understand which physical properties of the Hamiltonian ensure the validity of the ETH for physical observables. In Chapter \ref{Ch5}, the predictions of our approach are tested against a class of systems which is well-known to escape the description of statistical mechanics: The many-body localized systems. Our theoretical analysis predicts that, in the XXZ model, the equilibrium probability distribution of the local magnetization along the $x$ and $y$ direction should have thermal properties, described by the microcanonical ensemble. The numerical analysis agrees with this picture. This confirms that our approach goes beyond the standard statistical mechanics arguments, which rely on the assumption of a thermal state. The intuition developed is then used to show that an important feature of MBL, which has never been experimentally probed so far, can be unraveled by measuring local observables: The logarithmic growth of entanglement in time.\\

In the second part of the thesis, we exploited a technique from pure-states statistical mechanics to study the macroscopic behaviour of spin networks. These are the underlying states of Loop Quantum Gravity, introduced in Chapter \ref{Ch6}, and provide a well-defined notion of quantum geometry. In LQG, they describe the microscopic degrees of geometry at the quantum level. A spin network with a certain number of links and nodes has a dual geometric interpretation as a collection of adjacent polyhedra. Nodes are dual to polyhedra while links are dual to patches of surface shared by two adjacent polyhedra. Within this context, the concentration of measure technique can be used to argue that, in the macroscopic regime, there are some local properties of quantum systems which become true irrespectively of what is the precise state of the overall system. These properties are heralded by the ``Typical Local State'', of which we give the explicit form in two concrete situations. In the first part of Chapter \ref{Ch7} the technique is extended to study Hilbert spaces with gauge constraints. In the second part of Chapter \ref{Ch7} we studied the behaviour of a small patch of surface belonging to a larger surface. The typical state of such small patch of surface is a canonical Gibbs state where the role of the Hamiltonian is played by the local area operator. In Chapter \ref{Ch8} we apply the same idea to the study of the boundary of a large $3$D region of quantum space. By tracing out the ``internal'' degrees of freedom we compute the typical state of the boundary: it is a microcanonical state with an entropy proportional to the area of the boundary surface. 

We also study the relation which leads to the emergence of the boundary typicality. This has an exponential form which tames the fluctuations around the typical state. In the semiclassical regime the exponential behaviour becomes a threshold condition: If there is too much boundary-bulk entanglement, the state of the boundary becomes a thermal state with an entropy proportional to the area of the surface. We argue that such state is the quantum counterpart of a classical horizon as it has a thermal behaviour with an entropy proportional to the average value of its area. For this reason, we interpret the result as an information-theoretic collapse picture which, at the quantum level, mimics the condition for the creation of a black hole: the \emph{Hoop Conjecture}.

\section{Conclusions and Future work}

In this thesis we explored different aspects of the notion of thermal equilibrium. On the one hand we clarified that the emergence of a thermal behaviour should not always be traced back to the emergence of a thermal state. There is a weaker condition, specific for an arbitrary observable, whose validity can explain the observed emergence of thermal behaviour in experiments and numerical simulations. The theory that emerges is based on the Equilibrium equations derived in Chapter \ref{Ch3}. This should not be understood as an attempt to overthrow statistical mechanics. Rather, it is a way to expand the domain of applicability of statistical  techniques, in physics, to situations where the assumption of having a thermal state is not well justified. Adopting such perspective is useful also from the foundational point of view, as it allows to reconcile smoothly the maximum entropy principle and its predictions with the dynamics of isolated quantum systems: eigenvalues probability distributions with maximum entropy can emerge under the action of a unitary propagator. 

A connection with more popular approaches is provided by the notion of the Hamiltonian Unbiased Observables. These are peculiar solutions of the Equilibrium Equations which, if sufficiently degenerate, satisfy the Eigenstate Thermalization Hypothesis. We believe their correct interpretation is as follows. They are the ``prototypical observables'' which satisfy the ETH as their behaviour is very similar to the one given by Random Matrices. Unfortunately, because of that, they are completely insensitive to the energy scales of the model. While this can be used to our advantage, to argue that these observables will always thermalize to the predictions of the microcanonical ensemble, it also means that they are not always a good approximation to physical observables. Observables which thermalize are expected to exhibit a smooth dependence on the energy. Because of that, we think about the Hamiltonian Unbiased Observables as a first approximation to a complete characterization of observables which exhibits the ETH. Due to the fact that in several cases Random Matrix Theory provides a sufficiently good approximation to describe aspects of quantum equilibration and thermalization\cite{Torres-Herrera2016,Torres-Herrera2015,Borgonovi2016}, we formulate the following \emph{HUO proximity conjecture}: \emph{In the space of operators acting on a Hilbert space, observables which satisfy the ETH should be sufficiently close to observables which are Hamiltonian Unbiased Observables}. The validity of such conjecture would fit well both with observed experimental data and with numerical simulations. It would also reinforce the idea that the emergence of thermal behaviour should not necessarily be traced back to a thermal state of the system.\\

{\bf For future work}, we will certainly aim at refining the HUO picture to obtain a more detailed characterization of observables which satisfy the ETH. The first step is to understand how the HUBs condition is modified by small Hamiltonian perturbations. Moreover, the HUO are only one particular solution of the Equilibrium Equation derived in Chapter \ref{Ch3}. Other solutions are present and we believe it is worth to investigate their physical meaning. The main hope of the author is that this approach can be used to build a full theory to predict the equilibrium behaviour of quantum observables, beyond the thermal case. Moreover, in this work we mainly focused on characterising the equilibrium behaviour. The out-of-equilibrium behaviour of the HUOs was studied very briefly in Chapter \ref{Ch6}. It appears that the dynamical behaviour reproduces what we expect from a dynamically thermalising observable. However, the time-scale of equilibration seem to be too fast, for concrete observables. We believe this problem will also be addressed by studying how the HUB condition is modified by perturbation theory. 

There is one more direction which we would like to explore. It is well-known that the entanglement properties have important repercussions on the physics of many-body quantum systems. This has been exploited mostly in equilibrium physics. The analysis presented in Chapter \ref{Ch5} about the behaviour of the local entropy lead us to formulate the concept of \emph{Time-Dependent Entanglement Hamiltonian} as the operator $\tilde{H}(t)$ which is the solution of the equation $\rho_k(t) = \Tr_{N-k}\ket{\psi(t)}\bra{\psi(t)} = \frac{e^{-\tilde{H}(t)}}{\tilde{Z}(t)}$, where $\ket{\psi(t)} = e^{-\frac{i}{\hbar} \Ham t} \ket{\psi_0}$. This is the out-of-equilibrium counterpart of the concept of Entanglement Hamiltonian $H(E_\alpha)$, defined when the initial state is an energy eigenstate: $\rho_k(E_\alpha) = \Tr_{N-k} \ket{E_\alpha}\bra{E_\alpha} = \frac{e^{-H(E_\alpha)}}{Z(E_\alpha)}$. This definition is of practical use as the operator $H(E_\alpha)$ can be regarded as some ``Effective Hamiltonian'' of the system, which is generated by the entangled nature of the state $\ket{E_\alpha}$. Thus, the study of Time-Dependent Entanglement Hamiltonians could unravel interesting phenomenology, for example to study time-dependent phase transitions, and it certainly provides a practical concept to address the dynamical behaviour of entanglement in quantum systems.

To summarize, future work in this direction will focus on:
\begin{itemize}
\item Investigating the \emph{HUO proximity} picture, mainly by joint analytical and numerical techniques. The goal is to give a more refined characterization of quantum observables which satisfy the ETH.
\item Study the out-of-equilibrium behaviour of the HUOs and, possibly, of their more refined characterization resulting from the previous point. 
\item Other solutions of the Equilibrium Equations are possible. What is their physical relevance?
\item Study of entanglement dynamics, through the concept of time-dependent entanglement hamiltonians.
\end{itemize}


In the second part of the thesis we have used the concentration of measure phenomenon to study the local properties of a quantum geometry which are preserved in the macroscopic regime. This idea can be used to investigate the emergence of semiclassical properties, in the macroscopic regime. Indeed, thanks to the concentration of measure phenomenon the properties of the typical local state will be preserved in the macroscopic regime. In this sense, the typical local state is a representative state of the local behaviour, at the macroscopic scale. Focussing on this single state is a huge simplification. Nonetheless, this should allow to single out properties which are macroscopically relevant and study their semiclassical regime, to be compared with predictions from General Relativity.\\

For generic quantum systems, these properties are usually dominated by the existence of a large amount of entanglement between the local subsystem and the rest. This is also true for the Hilbert space of gauge-invariant spin networks. Some consequences of this large amount of entanglement are easy to understand. For example, a small patch of surface will look like it is in a canonical thermal state where the role of the Hamiltonian is played by the local area operator. The entropy associated to a patch of surface is therefore proportional to the area of the surface. It is well-known that such behaviour is part of the phenomenology of General Relativity, due to the work by Ted Jacobson \cite{Jacobson,Jacobson2016}.  This supports the idea that the Hilbert space of spin networks should be able to describe the continuous geometry that we experience at the classical level. 

The idea was pushed forward and another consequence of such ``typical'' large amount of entanglement was studied, in a slightly different context. When we have a large $3$D region of quantum space, the behaviour of the boundary is dominated by a large amount of entanglement with the bulk. Because of that, the state of the boundary becomes typical when the number of degrees of freedom in the bulk exceed a certain threshold, given by the number of degrees of freedom in the boundary. Such result mimics the analysis which leads to the threshold condition for the creation of a horizon. For this reason, we believe this ``local concentration of measure'' argument offers an \emph{information-theoretic} mechanism to understand the creation of a horizon from a semiclassical perspective. Too much entanglement between the bulk and the boundary provides the boundary with thermal features which resemble the ones of a horizon. \\

{\bf Future work} will be focused on refining this picture to make it more accurate. Various setups of physical interest, progressively more general, will be identified and their local properties will be studied by means of the concentration of measure. The first step will be to take an arbitrary spin network, parametrised by the number of vertices, links and loops, and to study the typical state of a link and of a vertex. The local properties will be concentrated around the typical local state and we will be able to study the way they approach the semiclassical regime. Understanding the emergence of the typical local properties would also be important for a comparison with the picture given by Tensor Networks in AdS/CFT. Indeed, the spin networks are simply tensor networks where the microscopic data have a well-defined geometric interpretation. The goal here would be to provide a tensor network model which is able to describe the semiclassical behaviour of a quantum spacetime. Ideally, this could result from a coarse-graining procedure which leads to the emergence of the typical properties at the local level. Moreover, a long-term goal of this investigation would be to understand the origin of the thermal properties of black holes, from a purely quantum perspective. In particular, we are interested in the role played by the boundary/bulk entanglement. Understanding this aspects should lead to a reliable quantum model for a black hole. With this mindset, the possibility to create a toy model to simulate the dynamical emergence of typicality is currently being investigated. 

Eventually, if we look at the issue of quantum gravity from a broader perspective, an important question to address in Loop Quantum Gravity is: ``how can the non-commutativity between area and curvature be amplified at the macroscopic scale?'' The ultimate goal of this research path is to find macroscopic effects of the non-commutativity which we would hopefully be able to measure. The question can be tackled using quantum information ideas and techniques\cite{Li2017a}. The way we think about this issue is built on a parallel with the theory of dynamical systems. There, the concept of Lyapunov exponents captures the rate at which two infinitesimally close trajectories diverge exponentially. Mirroring this behaviour, we will look for mechanisms which, in progressively larger systems, are able to magnify the non-commutativity between area and curvature, making it detectable at the macroscopic scale. 

Thus, future work in this direction will focus on:
\begin{itemize}
\item Find the typical state of a node and of a link, and connect their properties to the phenomenology of General Relativity;
\item Unravel the connection between spin networks and generic tensor network model for AdS/CFT correspondence;
\item Toy model for the dynamical emergence of a thermal behaviour of a horizon, due to its entanglement with the bulk;
\item  Macroscopic effects of the non-commutativity between area and curvature
\end{itemize}

\appendix
\chapter{Equilibrium Equations} \label{App:EE}

Here we present the full derivation of the equilibrium equations, using the Lagrange multipliers' technique. Our purpose is to find the distribution which maximises Shannon entropy of an observable $\mathcal{O}$. The space on which such an optimisation is formulated is the space of the finite-dimensional density matrices, which is the space of positive semi-definite and self-adjoint matrices. The maximisation is constrained by two equations: the first one accounts for the normalisation of the state while the second one accounts for a fixed value of the average energy. \\

If one wants to be absolutely rigorous, there is another constraint which needs to be imposed, which is the non-negativity of the density matrix. This is commonly indicated in the following way: $\rho \geq 0$. In other words, using the above-mentioned constraints there is no guarantee that, if a solution exists, it is a positive density matrix. This is not an actual problem, from the physical point of view, because one simply disregards a solution if it does not give a non-negative matrix. In such a case, the mathematical problem might have a solution, but not the associated physical question, which is the one we are really interested in. We would like to suggest an alternative route in such direction. If one wants to be completely rigorous, the Karush-Kuhn-Tucker (KKT) technique can be used to implement the non-negativity constraint on the density matrix.  However, we do not think we would gain any physical insight from the use of such technique and therefore we are not going to explore such a route.\\

We present here the derivation, in the general case of a mixed state $\rho = \sum_{n} q_n \Ket{\psi_n}\Bra{\psi_n}$ and of a degenerate observable $\mathcal{O} = \sum_{j,s} \lambda_j \Ket{j,s} \Bra{j,s}$, in which $\left\{ \Ket{j,s} \right\}$ is a complete basis of the Hilbert space. Here are the two constraints:

\begin{subequations}
\begin{align}
&\mathcal{C}_N \equiv \mathrm{Tr}\left( \rho\right) - 1 =\sum_{n;j,s} q_n |D_{js}^{(n)}|^2 -1 \\
&\mathcal{C}_E \equiv \mathrm{Tr} \left(\rho T \right) - E_0 = \sum_{n;j,s;j',s'} q_n \overline{D}_{js}^{(n)} T_{js,j' s'} D_{j' s'}^{(n)} - E_0
\end{align}
\end{subequations}

in which $D_{js}^{(n)}\equiv \braket{j,s}{\psi_n}$. Moreover $T_{js;j' s'} = \Bra{j,s}\hat{T}\Ket{j',s'}$. Exploiting Lagrange's multipliers technique one defines an auxiliary function $\Lambda_{\mathcal{O}}$, specific for the $\mathcal{O}$ observable, that can be freely optimised

\begin{align}
&\Lambda_{\mathcal{O}} [\hat{\rho},\lambda_E,\lambda_N] \equiv H_{\mathcal{O}}[\hat{\rho}] + \lambda_N \mathcal{C}_N + \lambda_E \mathcal{C}_E
\end{align}

The derivatives with respect to the Lagrange's multipliers $\lambda_E$ and $\lambda_N$ enforce the validity of the constraints 

\begin{subequations}
\begin{align}
&\frac{\delta \Lambda_{\mathcal{O}}}{\delta \lambda_N} = 0  \quad \Rightarrow \quad \mathcal{C}_N = 0\\
&\frac{\delta \Lambda_{\mathcal{O}}}{\delta \lambda_E} = 0  \quad \Rightarrow \quad \mathcal{C}_E = 0
\end{align}
\end{subequations}

while the derivatives with respect to the field variables and with respect to the statistical coefficients $q_n$ gives three equations, of which only two are independent:

\begin{align}
&\frac{\delta \Lambda_{\mathcal{O}}}{\delta D_{js}^{(n)}} = 0 && \frac{\delta \Lambda_\mathcal{O}}{\delta \bar{D}_{js}^{(n)}} = 0  
\end{align}

where

\begin{subequations}
\begin{align}
&\frac{\delta \Lambda_{\mathcal{O}}}{\delta D_{jn}} = - q_n \bar{D}_{js}^{(n)} \left[ \log \left( \sum_{m,p} q_m |D_{jp}^{(m)}|^2 \right) + (1-\lambda_N) \right] + \lambda_E q_n \sum_{i,p} T_{js,ip} \bar{D}_{ip}^{(n)}\\
&\frac{\delta \Lambda_\mathcal{O}}{\delta \bar{D}_{jn}} = - q_n D_{js}^{(n)} \left[ \log \left( \sum_{m,p} q_m |D_{jp}^{(m)}|^2 \right) + (1-\lambda_N) \right] + \lambda_E q_n \sum_{i,p} D_{ip}^{(n)} T_{ip,js} 
\end{align}
\end{subequations}

Instead of using these equations, we use the two following independent linear combinations

\begin{subequations}
\begin{align}
&D_{js}^{(n)} \frac{\delta \Lambda_{\mathcal{O}}}{\delta D_{js}^{(n)}} - \bar{D}_{js}^{(n)} \frac{\delta \Lambda_{\mathcal{O}}}{\delta \bar{D}_{js}^{(n)}}= 0  \\
& \frac{1}{2} \left(D_{js}^{(n)} \frac{\delta \Lambda_{\mathcal{O}}}{\delta D_{js}^{(n)}} + \bar{D}_{js}^{(n)} \frac{\delta \Lambda_{\mathcal{O}}}{\delta \bar{D}_{js}^{(n)}}= 0    \right)= 0 
\end{align}
\end{subequations}

which give the two equilibrium equations that we used in the main text

\begin{subequations}
\begin{align}
&\overline{\mathcal{E}}_n(j,s) = \mathcal{E}_n(j,s) \\
&- |D_{js}^{(n)}|^2 \log \left(\sum_{m,q} q_m  |D_{jq}^{(m)}|^2\right)  =   \left( 1-\lambda_N \right)|D_{js}^{(n)}|^2    -  \lambda_E \mathcal{E}_n(j,s) 
\end{align}
\end{subequations}

\chapter{Levy's lemma} \label{App:Levy}

In order to better understand the result it is useful to look at its most important step, which is the so-called Levy-lemma. Take an hypersphere in $d$ dimensions $S^{d}$, with surface area $V$. Any function $f$ of the point which does not vary too much 

\begin{align}
&f :  S^d \ni \phi \to f(\phi) \in \mathbb{R} &&|\nabla f| \leq 1
\end{align}

will have the property that its value on a randomly chosen point $\phi$ will approximately be close to the mean value.
\begin{align}
\frac{\mathrm{Vol}\left[ \phi \in S^d \, : \, f(\phi) - \MV{f} \geq \epsilon \right]}{\mathrm{Vol}\left[ \phi \in  S^d \right]} \leq 4 \, \mathrm{Exp} \left[ - \frac{d+1}{9 \pi^3} \epsilon^2 \right]
\end{align}

Where $\mathrm{Vol}\left[ \phi \in S^d \, : \, f(\phi) - \MV{f} \geq \epsilon \right]$ stands for ``the volume of states $\phi$ such that $f(\phi) - \MV{f} \geq \epsilon$''. $\MV{f}$ is the average of the function $f$ over the whole Hilbert space and $\mathrm{Vol}\left[ \phi \in  S^d \right]$ is the total volume of the Hilbert space. Integrals over the Hilbert space are performed using the unique unitarily invariant Haar measure.\\

The Levy lemma is essentially needed to conclude that all but an exponentially small fraction of all states are quite close to the canonical state. This is a very specific manifestation of a general phenomenon called ``concentration of measure'', which occurs in high-dimensional statistical spaces \cite{Ledoux2001}.

The effect of such result is that we can re-think about the ``a priori equal probability'' principle as an ``apparently equal probability'' stating that: as far as a small system is concerned almost every state of the universe seems similar to its average state, which is the maximally mixed state $\mathcal{E}_{\mathcal{R}} = \frac{1}{d_{\mathcal{R}}}\mathbb{I}_{\mathcal{R}}$.\\

\chapter{Proof of Theorem~\ref{thm:mainresult}}\label{App:Proof}




In this Appendix we provide the details of the proof of the main result of the paper, Theorem~\ref{thm:mainresult}. In the first subsection we give some background material, concerning the formalism of the generalized Bloch-vector parametrization. Such formalism will be used in the second subsection, where we give the actual proof of Theorem~\ref{thm:mainresult}.

\section{Generalised Bloch-vector parametrization}\label{sm:generalizedbloch}{Subsection~A1}
We start by briefly recalling the formalism of the generalized Bloch-vector parametrization \cite{Bengtsson2008,Bertlmann2008} of a pure quantum state. The standard Bloch-vector parametrization is a well-known way to describe the space of pure-states of a qubit, by using the isomorphism between its two-dimensional projective Hilbert space and a 2-sphere $\mathbb{S}^2$. Such an isomorphism can be easily generalized to arbitrary dimensions and it is well known that the projective space of a $D-$dimensional complex Hilbert space is isomorphic to $\mathbb{S}^{D^2-2}$. This isomorphism can be made explicit by associating to any normalized rank-$1$ projector $\ketbra{\psi}{\psi}$ a generalized Bloch vector $\vec{b}(\psi) \in \mathbb{S}^{D^2-2} \subset \mathbb{R}^{D^2 - 1}$ that fulfills
\begin{align}\label{eq:blochvectorcorrespondence}
\ketbra{\psi}{\psi} = \frac{\mathbb{I}}{D} + \sqrt{\frac{D-1}{D}} \,\, \vec{b}(\psi) \cdot \vec{\gamma}
\end{align}
where $\vec{\gamma}$ is a vector with elements $\gamma_i \coloneqq \hat{\gamma}_i / \sqrt{2}$ and $\hat{\gamma}_i$ are the $D^2-1$ generators of $SU(D)$, with the following properties:
\begin{align}
\hat{\gamma}_i &= \hat{\gamma}_i^\dagger & \Tr(\hat{\gamma}_i) &= 0 & \Tr(\hat{\gamma}_i\,\hat{\gamma}_j) &= 2\, \delta_{ij}
\end{align}
Even though the term ``Bloch vector'' is normally used to identify the $2$-dimensional case, hereafter we will use it for its $D$-dimensional counterpart. The constant prefactor $\sqrt{\frac{D-1}{D}}$ has been put to make the norm of the Bloch vector independent on the dimension of the Hilbert space and always equal to one. The square of the absolute value of the scalar product between two pure states $\Ket{\psi},\Ket{\psi'} \in \Hilb$ is mapped into the scalar product of the two Bloch vectors $\vec{b},\vec{b}'$, plus a constant term
\begin{align} 
  |\braket{\psi}{\psi'}|^2 = \frac{1}{D} + \frac{D-1}{D} \,  \vec{b} \cdot \vec{b}' .
\end{align}

From this relation we can see that mutual unbiasedness is a very natural condition when written in term of the respective Bloch vectors. For any two sets of pure states $\{\Ket{\psi_j}\}_j$ and $\{\ket{\psi'_k}\}_k$, with respective Bloch vectors $\{\vec{b}_j\}_j$ and $\{\vec{b}'_k\}_k$ we have 
\begin{equation}
  |\braket{\psi_j}{\psi'_k}|^2 = \frac{1}{D} \quad \Longleftrightarrow \quad \vec{b}_j \cdot \vec{b}'_k = 0 \quad .
\end{equation}
In other words the sets $\{\Ket{\psi_j}\}_j$ and $\{\ket{\psi'_k}\}_k$ are mutually unbiased if and only if their respective sets of Bloch vectors are orthogonal. Now we look at how the property 
of being a basis of the Hilbert space is written in terms of the Bloch vectors of the basis elements. Let $\{\ket{\psi_j}\}_{j=1}^{D} \subset \Hilb$ be a basis of a Hilbert space of dimension 
$D$, with associated Bloch vectors $\{\vec{b}_j\}_j$. Using Eq.~\eqref{eq:blochvectorcorrespondence} we find that $\{\ket{\psi_j}\}_{j=1}^{D}$ spans all of $\Hilb$ if and only if
\begin{align}
  \1 = \sum_{j=1}^{D} \ketbra{\psi_j}{\psi_j} = \1 + \sqrt{\frac{D-1}{D}} \sum_{j=1}^{D} \vec{b}_j \cdot \vec\gamma .
\end{align}
Since the elements of $\vec\gamma$ are the linearly independent generators of $SU(D)$, this is equivalent to $\sum_{j=1}^{D} \vec{b}_j = 0$.
At the same time, the vectors $\{\ket{\psi_j}\}_{j=1}^{D}$ are orthonormal if and only if $\forall j,k \in \{1,\dots,D\}$
\begin{align}
  \delta_{jk} = |\Scal{\psi_j}{\psi_k}|^2 = \frac{1}{D} + \frac{D-1}{D} \, \vec{b}_j \cdot \vec{b}_k ,
\end{align}
which is equivalent to
\begin{align}
 \vec{b}_j \cdot \vec{b}_k = \frac{D}{D-1} \delta_{jk} - \frac{1}{D-1} .
\end{align}
In summary we obtain that $\{\ket{\psi_j}\}_{j=1}^{D}$ is a complete orthonormal basis if and only if their Bloch vectors $\{\vec{b}_j\}_j$ satisfy the two following conditions
\begin{subequations}
\begin{align} \label{eq:conditionsforbeingabasis}
  &&\sum_{k=1}^{D} \vec{b}_k &= 0\\
  &\text{and} & \vec{b}_h \, \cdot \, \vec{b}_k &= \frac{D}{D-1} \delta_{hk} - \frac{1}{D-1} =
                                    \begin{cases}
                                      1 & \text{if } $h=k$\\
                                      -\frac{1}{D-1} & \text{if } h \neq k 
                                    \end{cases}
\end{align}\label{eq:BASE} 
\end{subequations}

\section{Proof of Theorem~\ref{thm:mainresult}}\label{sec:proofofmainresult}\label{sm:proofofmaintheorem}{Subsection~A2}
Here we present a detailed proof of Theorem~\ref{thm:mainresult} from the main text.
In order to do this we first introduce a well known theorem from geometry and the notions necessary to state it.
We then show how the generalized Block vector parametrization together with this theorem and properties of simplices allow to prove Theorem~\ref{thm:mainresult}.

In $\mathbb{R}^n$ an $n$-simplex is the generalization of the $2D$ triangle and the $3D$ tetrahedron to arbitrary dimensions.
A \emph{regular simplex} is a simplex which is also a regular polytope.
For example, the regular $2$-simplex is the equilateral triangle and the regular $3$-simplex is a tetrahedron in which all faces are equilateral triangles.
A $n$-simplex can be constructed by connecting a new vertex to all vertices of an $n-1$-simplex with the same distance as the common edge distance of the existing vertices.
This readily implies that the convex hull of any subset of $n$ out of the $n+1$ vertices of an $n$ simplex is itself a $n-1$-simplex, a so called \emph{facet} of the simplex.    
For $n=2$ they are the sides of the triangle, for $n=3$ they are the two dimensional triangles building the boundary surface of the tetrahedron.
To each facet we can associate a \emph{facet vector} defined as the vector orthogonal to the facet and with Euclidean length equal to the volume of the facet.
The result we need about these objects is the following theorem.

\begin{theorem}[Minkowski(-Weyl) Theorem \cite{poly}] \label{thm:minkowskiweyl}
  For any set of $n+1$ non co-planar vectors $\vec{V}_i \in \R^n$ that span $\R^n$ with the property 
  \begin{align}\label{eq:sumzero}
    \sum_{i=1}^{n+1} \vec{V}_i = 0 
  \end{align}
  there is a closed convex $n$-dim polyhedron whose facet vectors are the $\vec{V}_i$.  The converse is also true, for any closed convex polyhedron the facets vectors sum to zero.
\end{theorem}

If we apply the theorem to an $n$-simplex, whose facets vector are all of equal magnitude it can be easily seen that the (all equal) dihedral angles $\alpha$ between two facet vectors
are such that $\cos \alpha = - \frac{1}{n}$. This fact will be used in the proof of Theorem~\ref{thm:mainresult}. Projecting Eq.~\eqref{eq:sumzero} onto the direction of one vector $\vec{V}_k$ and using the fact that all dihedral angles have the same magnitude $\alpha$ in a simplex we have $\sum_{i=1}^{n+1} \vec{V}_k \cdot \vec{V}_i = 1 + n \cos \alpha = 0$. Which gives $\cos \alpha = - \frac{1}{n}$. We can now proceed with the proof of Theorem~\ref{thm:mainresult}.

\begin{proof}[Proof of Theorem~\ref{thm:mainresult}]
  If $\mathcal{H}$ is a $D$-dimensional Hilbert space, take an arbitrary decomposition $\mathcal{H}=\oplus_{j=1}^n \mathcal{H}_j$ and call $P_j$ the projectors onto $\mathcal{H}_j$. Define $p_{m,j} \coloneqq \bra{\psi_m}\,P_j\,\ket{\psi_m}$.
  For every $\ket{\psi_m}$ let
  \begin{equation}
    \ket{\psi_m^{(j)}} \coloneqq
    \begin{cases}
      P_j\,\ket{\psi_m}/\sqrt{p_{m,j}} & \text{if } p_{m,j} \neq 0\\
      0 & \text{otherwise}
    \end{cases}
  \end{equation}
  be the normalized projection onto the subspace associated with $P_j$ or the zero vector if $\ket{\psi_m}$ is orthogonal to that subspace.
  Now, for any vector $\ket{\varphi} \in \Hilb_j$ we can write $| \braket{\psi_m}{\varphi} |^2 = | \bra{\psi_m}\,P_j\,\ket{\varphi} |^2 = p_{j,k}\, | \braket{\psi_m^{(j)}}{\varphi} |^2$.
  As both $\ket{\psi_m^{(j)}}$ and $\ket{\varphi}$ are contained in $\Hilb_j$, via the construction described in \ref{sm:generalizedbloch}, they have associated generalized Bloch vectors $\vec{b}_m^{(j)}$ and $\vec{b}$ in $\mathbb{S}^{D_j^2-2}$.
  Using Eq.~\eqref{eq:blochvectorcorrespondence} we thus have
  \begin{equation}
    | \braket{\psi_m}{\varphi} |^2 = p_{m,j} \, \frac{1}{D_j} + p_{m,j} \, \frac{D_j-1}{D_j} \,  \vec{b} \cdot \vec{b}_m^{(j)} .
  \end{equation}
  We conclude that $\ket{\varphi} \in \Hilb_j$ has the desired property (Eq.~\eqref{eq:THEO}) of the basis vectors $\ket{j,k}$ if and only if $\vec{b}$ is orthogonal to all the $\vec{b}_m^{(j)}$.
  For any given $j$, in the worst case, all the $M$ vectors $\vec{b}_m^{(j)}$ are linearly independent, leaving a subspace of dimension $D_j^2-2-M$ for picking $\vec{b}$.
  Now, we don't want to pick just one vector $\vec{b}$ from this subspace, but $D_j$ many such vectors, which moreover satisfy the conditions in \eqref{eq:conditionsforbeingabasis} so that their associated state vectors form an orthonormal basis for $\Hilb_j$.
  The Minkowski(-Weyl) Theorem (Theorem~\ref{thm:minkowskiweyl}) tells us that this can be achieved by taking them to be the facet vectors $\vec{V}_i$ of a regular simplex in this subspace, as long as the subspace has sufficiently high dimension.
  More precisely, the first condition from \eqref{eq:conditionsforbeingabasis} is always satisfied for facet vectors $\vec{V}_i$ of general polytopes and the second condition can be achieved by using the facet vectors of a regular simplex, scaled so that they have Euclidean norm equal to one.
  This follows because the cosine of the angle between any two facet vectors of an $n$-simplex is $-1/n$.
  So, as long as the space of vectors orthogonal to all the $\vec{b}_m^{(j)}$ is large enough to accommodate for a $D_j-1$-simplex, $D_j$ suitable Bloch vectors of an orthonormal basis $\{\ket{j,k}\}_{k=1}^{D_j} \subset \Hilb_j$ that is unbiased with respect to all $\ket{\psi_m}$ can be found.
  This is the case as long as $D_j^2-2-M \geq D_j-1$.
\end{proof}

\section{Examples}\label{AppB}
In this second Appendix, we give more details about how to apply Theorem~\ref{thm:mainresult} to the three examples given in the manuscript and how to derive the results. We use a one-dimensional spin-1/2 chain as an exemplary case to showcase our result. Moreover, we will always be interested in using the Hamiltonian eigenvectors as a set of vectors for our theorem. This means $M=D$ and $\left\{ \ket{\psi_j} \right\}_{j=1}^D = \left\{ \ket{E_m}\right\}_{m=1}^D$. However, if for some reason one is interested in a limited portion of the energy spectrum, the results can be strengthened by limiting the set of eigenvectors to $M < D$.

\subsection*{Example 1: Local observables}

As first application of our Theorem, we study the emergence of ETH in a local observable which has support on less than half of the whole chain. The total number of spins is $N$ and the Hilbert space is split into tensor products of $k$ and $N-k$ spins: $\mathcal{H} = \mathcal{H}_k \otimes \mathcal{H}_{N-k}$. Local observables $A_{\mathrm{loc}}= A_k \otimes \mathbb{I}_{N/k} = \sum_{j=1}^{2^k} P_j a_j$ have support on $k \leq N-k$ sites. In this case all eigenvalues have degenerate subspaces with the same dimension: $\mathrm{dim} \mathcal{H}_j = \Tr P_j = D_j = 2^{N-k}$. The condition that ensures the validity of the hypothesis of Theorem~\ref{thm:mainresult} is $2^{N-k}(2^{N-k}-1) \geq 2^N + 1$. Applying the $\log $  to both sides and with some algebraic manipulations we obtain
\begin{equation}
2(N-k)\log 2 - N \log 2 \geq \log \left( \frac{1-\frac{1}{2^{N-k}}}{1+\frac{1}{2^N}}\right)
\end{equation}
The right-hand side is always negative. So we request the following (slightly stronger) condition
\begin{equation}
2(N-k)\log 2 - N \log 2 \geq 0 \geq \log \left( \frac{1-\frac{1}{2^{N-k}}}{1+\frac{1}{2^N}}\right)
\end{equation}
The condition arising from the first inequality gives $k \leq \frac{N}{2}$. Therefore, local observables with support on less than half of the chain satisfy the assumptions of our theorem. For them we obtain that there is a basis $\ket{a_j,k}$ that diagonalizes the observable, such that

  \begin{equation} 
    | \braket{E_m}{a_j,k} |^2 = \bra{E_m} P_j \ket{E_m} / 2^{N-k}
  \end{equation}
since $P_j = A_j \otimes \mathbb{I}_{N/k}$ we have $\bra{E_m} P_j \ket{E_m} = \Tr_k \left(A_j \rho_k(E_m)\right)$ where $\rho_k(E_m) = \Tr_{N/k}\ket{E_m}\bra{E_m}$. For small subsystems $k \ll N-k$, if the Hamiltonian eigenstates are highly entangled, which is expected to be true for a non-integrable system in the bulk of the spectrum, the von Neumann entropy of the reduced state is close to the maximum value $k \log 2 - S_{\mathrm{vN}}(\rho_k(E_m)) \leq \epsilon_{k}(E_m)$ with $\epsilon_k(E_m) \geq 0$. Using Pinsker's inequality and the fact that the relative entropy with respect to the maximally mixed state is just the difference between the two entropies we have 
\begin{equation}
|| \rho_k(E_m)- \frac{\mathbb{I}}{2^k}||^2 \leq \frac{1}{2} (k \log 2 - S_{\mathrm{vN}}(\rho_k(E_m))) \leq \frac{\epsilon_k(E_m)}{2} \, .
\end{equation}
Whenever $\epsilon_k(E_m) \ll 1$, which is expected to be true in the bulk of the spectrum, we have
\begin{equation}
| \braket{E_m}{a_j,k} |^2 = \frac{\bra{E_m} P_j \ket{E_m} }{2^{N-k}} = \frac{\Tr_k A_j \rho_k(E_m)}{2^{N-k}} = \frac{\bra{a_j} \rho_k(E_m)\ket{a_j}}{2^{N-k}} \simeq \frac{1}{2^N} \, .
\end{equation} 
We can therefore conclude that entanglement in the energy eigenstate is the feature that makes local observables be HUOs. Provided certain mild assumptions, which have been discussed in the paper, are satisfied, this guarantees that they satisfy ETH. We conclude that, if the energy eigenstates are highly entangled in a certain energy window $I_0 = \left[E_{a} , E_{b} \right]$, as it is expected to happen in a non-integrable model, ETH will hold for all local observables, in the same energy window.

\subsection*{Example 2: Extensive observable - Global magnetization }
In this second example we study the consequences of our theorem for an observable which is the extensive sum of local observables: the global magnetization $M_z=\sum_{n=1}^N \sigma_n^z$. Its spectral decomposition is $M_z=\sum_{j=-N}^{N} j \, P_j$ so the Hilbert space is decomposed as the direct sum of the $\mathcal{H}_j$, which are the images of the $P_j$: $\mathcal{H}=\bigoplus_{j=-N}^{N} \mathcal{H}_j$. Their dimension $D_j = \Tr \, P_j$ can be computed using combinatorial arguments: $D_j = C^N_{j}\equiv {N \choose \frac{N-j}{2}}$.  At fixed size $N$, $D_j \in [1,{N \choose N/2}]$. The inequality $D_j \geq 1+\frac{2^N+1}{D_j}$ selects a subset $j \in [-j_{*} ,j_{*} ]$ of subspaces $\mathcal{H}_j$ for which the theorem will hold. Note that the interval $[-j_{*},j_{*}]$ is symmetric with respect to zero because $D_{j}=D_{-j}$. In order to find how $j_{*}$ scales with the system size, we numerically compute how many subspaces $\mathcal{H}_j$ meet the condition $D_j \geq 1+\frac{2^N+1}{D_j}$. We call this number $q(N)$. Since the eigenvalues are given by the relative number $j \in \mathbb{Z} \cap [-N,N]$ and they are equally spaced, we have $q(N)=2 j_{*}+1$. Which means $j_{*}=\frac{q(N)-1}{2}$. In Fig.\ref{fig:Macro} we can see that it scales linearly with the system size: $q(N) \sim 1.56 N$. This gives $j_{*}(N) \sim 0.78 N$.

\begin{figure}[h]
\centering
\includegraphics[width=3.6 in]{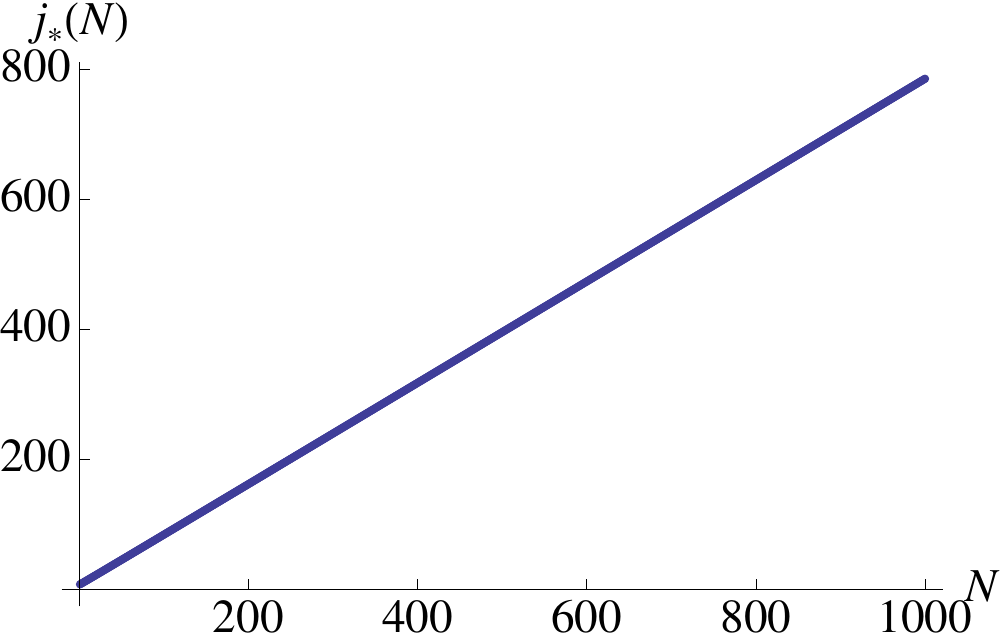}
\caption[Number of subspaces satisfying the hypothesis of the theorem]{Scaling of the number of subspaces $\mathcal{H}_j$ which meet the condition $D_j \geq 1+ \frac{2^N+1}{D_j}$.}\label{fig:Macro}
\end{figure}
The picture that we obtain is the following. States with ``macroscopic magnetization'', i.e. around the edges of the spectrum of $M_z$, have very small degeneracy and the theorem is not going to hold for them. In the bulk of the spectrum, however, there is a large window $j\in [-j_{*}(N),j_{*}(N)]$ where the respective subspaces $\mathcal{H}_j$ meet the conditions for the validity of the theorem. In summary, if we apply the theorem to the global magnetization we obtain:

\begin{equation}
\forall j \in \mathbb{Z} \cap [-j_{*}(N),j_{*}(N)]\,, \qquad  |\braket{j,s}{E_m}|^2 = \frac{\bra{E_m} P_j \ket{E_m}}{D_j} \,.
\end{equation}
We know that the relation we are interested in is the Hamiltonian Unbiasedness, which would be $\frac{\bra{E_m} P_j \ket{E_m}}{D_j} \simeq \frac{1}{2^N}$. For this reason we study the relation
\begin{equation}
\frac{\bra{E_m} P_j \ket{E_m}}{D_j} \simeq \frac{1}{2^N} \qquad \Rightarrow \qquad  \bra{E_m} P_j \ket{E_m} \simeq \frac{D_j}{2^N}\,,
\end{equation}
which in turn means to study how $\frac{D_j}{2^N}$ behaves. For this goal, in the large $N$ regime we can use Stirling's approximation. As it is known, there is not a unique way of using it. Rather, there are different ways, depending on the number of sub-leading terms that one is willing to use. Here we focus on the leading term. Note that Stirling's approximation can be used throughout the whole window $[-j_{*}(N),j_{*}(N)]$, as long as $N \gg 10$. This is true because $j_{*}(N) \sim 0.78 N$ so $|j| \in [0,0.78 N]$ and $\frac{N-j}{2}\sim 0.1*N$. Therefore, as long as $0.1 N \gg 1$, we can use Stirling's formula for all the factorials involved in $D_j$. It can be shown that if $n\geq k \gg 1$, at the leading order we have ${n \choose k} \sim 2^{n H_2(\frac{k}{n})}$ where $H_2(x)\equiv -x\log_2 x - (1-x) \log_2 (1-x)$ is the binary entropy. Using this we get 
\begin{equation}
D_j \approx 2^{N\,H_2\left(\frac{N-j}{2N}\right)} = 2^{N\,H_2\left(\frac{1}{2} - \frac{j}{2N}\right)} \,\,\, ,
\end{equation}
which in turn gives 
\begin{equation}
\frac{D_j}{2^N} \approx 2^{-N [1-H_2\left(\frac{1}{2} - \frac{j}{2N}\right)]} \,\,\, .
\end{equation}
We have the size of the system $N$ which multiplies a function which is a binary
relative entropy. If we call $p_{\mathrm{mix}}\coloneqq \left( \frac{1}{2}, \frac{1}{2} \right)$ and $p(j)=\left(\frac{1}{2} - \frac{j}{2N} ,\frac{1}{2} + \frac{j}{2N} \right)$ we have 

\begin{equation}\label{eq:ldb}
1-H_2[\left(\frac{1}{2} - \frac{j}{2N}\right)] = H_2\left[ p(j)||p_{\mathrm{mix}} \right] \quad \Rightarrow \quad \frac{D_j}{2^N} \approx 2^{-N H_2\left[ p(j)||p_{\mathrm{mix}} \right]}\,.
\end{equation}
Eq.~\eqref{eq:ldb} has a very interesting form.
It is telling us that the statistics of the eigenvalues $a_j$, induced by the eigenstates $\ket{E_m}$, satisfies a large deviation bound. The rate function is given by the binary Kullback-Leibler divergence $H_2[p(j)|\!| p_{\mathrm{mix}}]$. Now we can formulate a clear statement. Choose a subspace $\mathcal{H}_j$ with $|j| < j_{*}$ where the hypothesis of our theorem hold. \emph{If there is a $k \in \mathbb{N}, \,\,\, k < j_{*}$ such that for all $j \in [-k,k]$ we have $\bra{E_m}P_j \ket{E_m} \approx 2^{-N H_2\left[ p(j)||p_{\mathrm{mix}} \right]}$, the global magnetization $M_z$ will be an HUO and satisfy the ETH in the subspaces $\bigoplus_{|j|<k} \mathcal{H}_j$.} Concretely, this will happen if the measurement statistics generated by the energy eigenstates $\ket{E_m}$ on the eigenvalues $a_j$ satisfies a large deviation bound.

To build our intuition on what this means we evaluate $H_2\left(\frac{N-j}{2N}\right)$ in two regimes allowed by our Theorem: $\frac{|j|}{N} \ll 1$ and $\frac{|j| - j_{*}}{N} \ll 1$. In the first case, calling $x=\frac{|j|}{N}$ we can Taylor-expand $H_2(\frac{1-x}{2} )$ around $x \ll 1$ to obtain 
\begin{equation}
H_2\left(\frac{1-x}{2}\right) \stackrel{x \ll 1}{\approx} 1 - \frac{x^2}{2} \qquad \Rightarrow \qquad H_2\left(\frac{1-|j|/N}{2}\right) \approx 1 - \frac{j^2}{2N^2} \qquad |j|/N \ll 1\,. \label{eq:expansion}
\end{equation}

In the regime $|j|\approx j_{*}$ we have a better way to estimate $D_j$. Indeed in such regime $D_j \approx D_{j_{*}}$, which satisfies $D_{j_{*}} \approx 1+ \frac{2^N-1}{D_{j_{*}}}$. Solving for $D_{j_{*}}$ and taking the leading order in $N$ we obtain $D_{j_{*}} \approx 2^{N/2}$. Moreover, using the expression in Eq.~\eqref{eq:expansion} we can find how $D_j$ deviates from $D_{j_{*}}$. Indeed expanding $H_2(\frac{1-|j|/N}{2})$ around $j_{*}$ we get
\begin{equation}
H_2\left(\frac{1-|j|/N}{2}\right) \approx H_2\left(\frac{1-j_{*}/N}{2}\right) - \left. \frac{dH_2}{dx} \right\vert_{x=\frac{1-j_{*}/N}{2}} \frac{|j|-j_{*}}{2N} \approx \frac{1}{2} - \frac{3}{2} \left(\frac{|j|-j_{*} }{N}\right)\,.
\end{equation}
In summary, when $N \gg 10$
\begin{align}
& \frac{D_j}{2^N} \approx  \left\{ \begin{array}{ll} 
 2^{- \frac{j^2}{2N}} & \quad \frac{|j|}{N} \ll 1  \\
& \\
 2^{-\frac{N}{2} - \frac{3}{2}(|j|-j_{*})} & \quad \frac{|j| - j_{*}}{N} \ll 1
  \end{array} \right. 
\end{align}
This means that when we approach the thermodynamic limit $N \to \infty$, the eigenvalues with higher magnetization will be exponentially suppressed in the system size. This is indeed what we expect to be true at the macroscopic level.

\subsection*{Example 3: Macroscopic equilibrium - Normal typicality and von Neumann's Quantum H-theorem}

In this last example we investigate the connection of our theorem with the notion of Macro-observables proposed by von Neumann in his work on the Quantum H-theorem \cite{Neumann1929,Tumulka2010}.
This in turn is strictly related with the notion of Normal typicality developed in a series of more recent works by Goldstein \emph{et al.} \cite{Goldstein2006,Goldstein2010,Goldstein2010a}.
Again, we start by decomposing our Hilbert space $\mathcal{H}$ as a direct sum of subspaces $\mathcal{H}_j$.
The index $j$ runs over a finite number of values that identify different macroscopic properties of the system.
One could say that it identifies different ``macrostates'', characterized by the expectation value of commuting macroscopic observables.
In the original idea by von Neumann, in a classical system we measure position and momentum, which commute.
His point was that there are some coarse-grained approximation of actual position and momenta which can be ``rounded'' to obtain a set of commuting macro-observables.
Such set of commuting Macro-observables provides a decomposition of the Hilbert space $\mathcal{H} = \bigoplus_{j=1}^n \mathcal{H}_j$ where the index $j$ runs over all the possible different macrostates.
Each one of these spaces $\mathcal{H}_j$ is hugely degenerate and we assume here that we can use our Theorem for all of them.
Using the concentration of measure phenomenon it can be shown \cite{Goldstein2010,Goldstein2010a,Goldstein2006} that for most $t$, $\bra{\psi(t)}P_j \ket{\psi(t)} \simeq \frac{D_j}{D}$ for all $j$, for most Hamiltonians in the sense of Haar and for all $\psi(0)$.

Concretely, this happens for all $\psi(0)$ if and only if $\bra{E_m}P_j \ket{E_m} \simeq \frac{D_j}{D}$ for all $j$ and $m$.
Such a relation can be proven to hold in the same sense as before.
For most Hamiltonians in the sense of Haar
\begin{equation}
  \bra{E_m}P_j \ket{E_m} \simeq \frac{D_j}{D} \qquad \forall  j,m\,.
\end{equation}
The unitary for which this ``most'' holds is the one connecting the Hamiltonian eigenbasis to the basis giving the decomposition of the Hilbert space into ``commuting macro-observables''. We can now see that the connection of these ideas with ETH is unraveled by our theorem 1 and by the notion of HUO. Indeed, using our theorem, we can write 
\begin{equation}
\bra{E_m}P_j \ket{E_m} = D_j \left| \braket{j,s}{E_m}\right|^2
\end{equation}
Therefore
\begin{equation}
  \bra{E_m}P_j \ket{E_m} \simeq \frac{D_j}{D} \qquad \Longleftrightarrow \qquad \left| \braket{j,s}{E_m}\right|^2 \simeq \frac{1}{D}
\end{equation}
From this we conclude that \emph{for most Hamiltonians, in the sense of Haar, that the basis $\left\{\ket{j,s}\right\}$ which diagonalizes all the ``commuting macro-observables'' giving the decomposition $\mathcal{H}=\bigoplus_j \mathcal{H}_j$ is an Hamiltonian Unbiased Basis (HUB)}.
Moreover, thanks to the fact that each subspace $\mathcal{H}_j$ is highly degenerate and that the decomposition $\mathcal{H}=\oplus_j \mathcal{H}_j$ is generated by Macro-observables, this proves that all Macro-observables built in this way are HUO. Again, provided certain mild assumptions, which have been discussed in the main text, are satisfied, this guarantees that they satisfy ETH.

\chapter{Intertwiner Hilbert spaces}\label{App:Dimensions}

The main argument of the paper relies on the study of the dimensions of all Hilbert spaces involved, which we give here. We start with the constrained Hilbert space

\begin{align}
&\mathcal{H}_{\mathcal{F}_R} = \mathrm{Inv}_{SU(2)} \left[ \bigotimes_{e=1}^{E_{\partial R}} V_{j_e} \bigotimes_{l=1}^{L_{R}} \left( V_{j_l} \otimes \overline{V}_{j_l} \right) \right]
\end{align}
which can be embedded in a larger Hilbert space 
\begin{align}\label{space}
\mathcal{H}_{\mathcal{F}_R} \subseteq \mathcal{H} \coloneqq \mathcal{H}_{\partial R} \otimes  \mathcal{H}_R,
\end{align} 
where the unconstrained boundary and bulk Hilbert spaces are defined as
\begin{align}
&\mathcal{H}_{\partial R} \coloneqq \bigotimes_{e=1}^{E_{\partial R}} V_{j_e} && \mathcal{H}_R \coloneqq \bigotimes_{l=1}^{L_R} V_{j_l} \otimes \overline{V}_{j_l} \,\,\, .
\end{align}
We remember that $V_j$ is the space of the $j$-th irreducible representation of SU(2) and $\overline{V}_{j}$ is its dual. Their dimension is $\mathrm{dim} V_j = 2j+1$. Therefore
\begin{align}
&d_{\partial R} = \prod_{e=1}^{E_{\partial R}} (2j_e+1) &&d_{R} = \prod_{l=1}^{L_{R}} (2j_l+1)^2
\end{align}
The dimension of the constrained Hilbert space depends on the specific irreps which colours the graph. It can be given implicitly as the result of an integral over the characters of SU(2):
\begin{align}
&d_{\mathcal{F}_R} = \int_{\mathrm{SU(2)}} dg \, \prod_{e=1}^{E_{\partial R}} \prod_{l=1}^{L_{R}} \chi_e(g) (\chi_l(g))^2
\end{align}
where $\chi_e(g)$ is the SU(2) character in the $j_e$ irrep. Unfortunately there is no closed expression for such quantity. For such a reason we restricted our analysis to the case where $j_e=j_l=j_0$. In this case we have

\begin{align}
&d_{\partial R} = (2j_0+1)^{E_{\partial R}} &&d_{R} = (2j_l+1)^{2L_{R}} \, .
\end{align}
The constrained Hilbert space becomes

\begin{align}
&\mathcal{H}_{\mathcal{F}_R} = \mathrm{Inv}_{SU(2)} \left[ V_{j_0}^{E_{\partial R}+2L_R} \right]
\end{align}
and its dimension has been computed in \cite{Livine2006,Livine2008} in more general terms. We summarise the argument here. Suppose we have a tensor product of $n$ copies of the same representation space $V_j$. We can decompose it into its irreducible representations, each of them with a degeneration space:
\begin{align}\label{eq:dec}
&\bigotimes_{n} V_j = \bigoplus_k V_k \,\, \cdot \, {}^jF_k^n 
\end{align}
Here ${}^jF_k^n$ is such degeneration space and  ${}^{j} d_k^{n}\coloneqq \mathrm{dim} \,{}^{j} F_k^{n}$ counts the number of times the space $V_k$ appear in the decomposition.\\

The request of gauge-invariance applied to this space would select the trivial representation space:
\begin{align}
&\mathrm{Inv}_{SU(2)} \left[\bigotimes_{n} V_j \right] = V_0 \,\, \cdot \, {}^jF_0^n 
\end{align}
Since $\mathrm{dim}V_0 = 1$, this space is isomorphic to ${}^jF_0^n$. Applying this result to the constrained Hilbert space 
$\mathcal{H}_{\mathcal{F}_R}$ one can write down the following isomorphism ($\simeq$)
\begin{align}
&\mathcal{H}_{\mathcal{F}_R} \simeq {}^{j_0}F_0^{E_{\partial R}+2L_{R}}
\end{align}
There is no closed form for ${}^jd_k^n$, but it can be given as an integral over the SU(2) characters:
 \begin{align}
{}^{j} d_k^{n} & = \int_{SU(2)} dg \,\,  \chi_k(g) (\chi_{j}(g))^n =\\
& =  \frac{2}{\pi}\int_{0}^{\pi} d\theta \sin \theta \sin (2k+1)\theta \left( \frac{\sin (2j+1)\theta}{\sin \theta}\right)^n \nonumber
\end{align}
An expression for this quantity in the regime $n\gg k$ has been given in \cite{Livine2006,Livine2008}:
\begin{align}\label{eq:asy}
&{}^{j} d_k^{n} \sim \frac{(2j+1)^n(k+1)}{[j(j+1)n]^{3/2}}
\end{align}
Since we are interested in the regime $E_{\partial R},2L_{R}\gg 1$ we can use Eq.(\ref{eq:asy}) to obtain
\begin{align}\label{eq:dR}
&d_{\mathcal{F}_R}= {}^{j_0}d_0^{E_{\partial R}+2L_{R}} \sim \frac{(2j_0+1)^{E_{\partial R}+2L_{R}}}{[j_0(j_0+1)(E_{\partial R}+2L_{R})]^{3/2}}  
\end{align}
Using this expression and $d_{\partial R}$ we obtain the bound in Eq.(12) of the main text:
\begin{align}
&\frac{d_{\partial R}}{\sqrt{d_{\mathcal{F}_R}}} \sim (2j_0+1)^{\frac{E_{\partial R}}{2}-L_{R}} \left[ j_0(j_0+1)(E_{\partial R}+2L_{R})\right]^{3/4}
\end{align}
We now consider the decomposition of the constrained Hilbert space given in Eq.(6) of the main text. First we split the tensor product into
edges and loops:
\begin{align}\label{eq:gaugespace}
&\mathcal{H}_{\mathcal{F}_R} = \mathrm{Inv}_{SU(2)} \left[ \bigotimes_{E_{\partial R}} V_{j_0} \,\, \otimes \,\, \bigotimes_{2L_{R}} V_{j_0}\right]
\end{align}
Using the decomposition in Eq.(\ref{eq:dec}) we can write 
\begin{align}
&\bigotimes_{E_{\partial R}} V_{j_0} = \bigoplus_k V_k \cdot {}^{j_0}F_k^{E} \\
& \bigotimes_{2L_{R}} V_{j_0} = \bigoplus_y V_y \cdot {}^{j_0}F_y^{2L}
\end{align}
Plugging this into Eq.(\ref{eq:gaugespace}) and considering that the request of gauge invariance forces the two indices $k,y$ to be equal, we get:
\begin{align}
\mathcal{H}_{\mathcal{F}_R}  &= \bigoplus_{k} V_k^{(E)} \otimes V_k^{(L)}  \,\, \cdot {}^{j_0}F_{k}^{E} \otimes {}^{j_0}F_{k}^{2L} \nonumber \\
& = \bigoplus_k \mathcal{H}_k^{(E)} \otimes \mathcal{H}_k^{(2L)}
\end{align}
From this we can easily identify 
\begin{align}
&\mathcal{H}_k^{(E)} = V_k^{(E)} \,\, \cdot {}^{j_0}F_{k}^{E} \\
&\mathcal{H}_k^{(L)} = V_k^{(L)} \,\, \cdot {}^{j_0}F_{k}^{2L} 
\end{align}
Moreover, ${}^{j_0}F_{k}^{E}$ and ${}^{j_0}F_{k}^{2L} $ are the degeneration spaces which we called $\mathcal{D}_k^{E}$ and $\mathcal{D}_k^{L}$ in the main text. This also implies that
\begin{align}
&\mathrm{dim} \mathcal{D}_k^E = {}^{j_0} d_k^E &&\mathrm{dim} \mathcal{D}_k^L = {}^{j_0} d_k^{2L}
\end{align}

\addcontentsline{toc}{chapter}{Bibliography}
\bibliography{library}        
\bibliographystyle{unsrt}  

\end{document}